\documentclass[%
 reprint,
 amsmath,amssymb,
 aps,
]{revtex4-2}

\usepackage{graphicx}
\usepackage{dcolumn}
\usepackage{bm}
\usepackage{hyperref}

\usepackage{amsthm}
\usepackage{physics}
\usepackage{xcolor}
\usepackage{algorithm}
\usepackage{algpseudocode}
\usepackage{mathrsfs}
\usepackage{enumitem}
\usepackage{mathtools}

\newtheorem{theorem}{Theorem}
\newtheorem{lemma}[theorem]{Lemma}

\newtheorem{proposition}[theorem]{Proposition}
\newtheorem{definition}{Definition}
\newtheorem{claim}{Claim}
\newtheorem*{remark}{Remark}

\newcommand{\imp}[1]{\textcolor{black}{#1}}

\begin{document}

\preprint{APS/123-QED}

\title{Unifying information propagation models on networks and influence maximization}

\author{Yu Tian}
\email{yu.tian@maths.ox.ac.uk}
\author{Renaud Lambiotte}%
\email{renaud.lambiotte@maths.ox.ac.uk}
\affiliation{%
 Mathematical Institute, University of Oxford, Oxford, OX2 6GG, United Kingdom}%

\date{September 16, 2022}

\begin{abstract}
    Information propagation on networks is a central theme in social, behavioral, and economic sciences, with important theoretical and practical implications, such as the influence maximization problem for viral marketing. Here, we consider a model that unifies the classical independent cascade models and the linear threshold models, and generalise them  by considering continuous variables and allowing feedback in the dynamics. We then formulate its influence maximization as a mixed integer nonlinear programming problem and adopt derivative-free methods. Furthermore, we show that the problem can be exactly solved in the special case of linear dynamics, where the selection criterion is closely related to the Katz centrality, and propose a customized direct search method with local convergence. We then demonstrate the close-to-optimal performance of the customized direct search numerically on both synthetic and real networks. 
\end{abstract}

\maketitle


\section{\label{sec:introduction}Introduction}
The rapid growth of online social networks, such as Facebook and Twitter, allows hundreds of millions of people worldwide to interact with each other, providing access to a vast source of information on an unprecedented scale.  The \textit{propagation} of information, opinion, innovation, rumor, etc., is a critical component to explain, e.g., how a piece of information could quickly become pervasive through a network via ``word-of-mouth" \cite{Bakshy_InfoRole_2012,Centola_InfoDiffBehav_2010,Nekovee_RumorDiff_2007}. Accordingly, understanding how information spreads in social networks is a central theme in social, behavioral, and economic sciences, with theoretical and practical implications, such as the adoption of political viewpoints in presidential elections and the \textit{influence maximization} problem for viral marketing \cite{Bovet_Twitter_2019,Chen_InflMaxViral_2010,Leskovec_ViralMarket_2007}, attracting expertise from various fields including mathematics, physics, and biology \cite{Mossel_Roch_2010,Pastor-Satorras_epidemic_2015,pastor-satorras_internet_2004}. 

Specifically in the context of influence maximization, one can distinguish two main classes of information propagation models: the independent cascade (IC) model, and the linear threshold (LT) model, where nodes adapt their behavior from each neighbor independently, or from the collective influence of the whole neighborhood, respectively \cite{Kempe_influence_2003,Shakarian_models_2015}. It is well known that both models, which we refer to as ``classic models," suffer from limitations. First of all, the state space of both models is binary, as nodes can only have states either active or not, while various levels of influence and of confidence could coexist among agents in real cases. Moreover, there is no feedback in both processes, as nodes can only stay active after being activated, thus may not influence back nodes that influenced them as in real life. These limits have called, and still call for more general models allowing us to consider dynamics with feedback between \imp{nodes} and more heterogeneity in the agents' behavior.

In parallel to this line of work, simple and complex contagions have attracted much research interests in mathematical sociology and physics \cite{Centola_ccontagion_2007,Guibeault_Complcontag_2018}. Essentially, complex contagion considers situations when the reinforcement of a signal favors its future adoption, which can be modeled via deterministic threshold models. However, their focus is usually on understanding the effects of specific structures, the presence of \imp{shortcuts} or the density of triangles for example, and they also tend to consider binary state variables. For models with continuous variables, more has been done within the field of opinion dynamics, where linear models build on the heat equation \cite{Ma_heat_2008}, such as the DeGroot model \cite{degroot_model_1974}, and nonlinear models include the bounded confidence model, for example \cite{Deffuant_boundconf_2000}. A first contribution of this work is to introduce a nonlinear, deterministic model for information propagation which relaxes the constraints of binary variables while allowing feedback between \imp{nodes}, and also possesses the classic models as limiting cases, thus providing a unifying framework for the information propagation.

As a second contribution, we consider the important problem of influence maximization (IM), known to have potential applications in various domains \cite{Leskovec_ViralMarket_2007,Leskovec_outbreak_2007,Song_recomm_2006}. Given a network and an associated information propagation process, the classic IM problem consists in selecting a small set of nodes to activate initially with the aim of maximizing the overall \textit{influence spread}, commonly defined as the number of activated nodes at the end of the process. Kempe \textit{et al.} \cite{Kempe_influence_2003} formulated it as a stochastic combinatorial optimization under the classic models, and proposed a greedy algorithm with theoretical approximation guarantees. The near-optimal asymptotic bounds of this seminal work has triggered a vast amount of research in this direction, mostly to further reduce the running time \cite{Goyal_CELF_2011,Leskovec_outbreak_2007}. There are also heuristic solutions, such as centrality-based methods and genetic algorithms \cite{Banerjee_IMsurvey_2020}. 

However, the aforementioned theoretical guarantees are obtained \imp{under certain assumptions on the influence spread function}, which could be strict in practice \cite{Li_IMsurvey_2018}. Hence, to consider a general information propagation model, we have developed a general framework, where we formulate the IM problem as a \textit{mixed integer nonlinear programming} (MINLP). The exact methods for MINLP are mostly based on \imp{relaxing the problem to be linear in an appropriate manner}, and require first-order information \cite{Belotti_MINLP_2013,Boukouvala_GOMINLPCDFM_2016,Burer_MINLPs_2012}, which is not generally available in the IM problem. For this reason, we treat the objective function as a black box and adopt \textit{derivative-free methods} \cite{Boukouvala_GOMINLPCDFM_2016} \imp{where only evaluations of the the objective function are required}, with a \textit{mesh adaptive direct search} method as a general solution \cite{Abramson_MADSmv_2009}. Furthermore, we propose a customized method with local convergence, specifically designed for the proposed general class of information propagation model.


The rest of our paper is organised as follows. In Sec.~\ref{sec:related_work}, we discuss in detail the IC model and the LT model. In Sec.~\ref{sec:general_diffusion_model}, we propose the general class of information propagation model and illustrate its salient features. In Sec.\ref{sec:influence_maximization}, we propose the general framework for the IM problem, and the customized method. The numerical results of both the proposed model and the proposed method are discussed in Sec.\ref{sec:experiments_results}. Finally, we conclude with some implications for future work in Sec.\ref{sec:conclusion}. In the Appendixes, we include detailed theoretical results, and further features of the proposed model and the IM problem. 

\section{\label{sec:related_work}Related work}
We first describe two classic information propagation models, the IC model and the LT model, in more detail \cite{Kempe_influence_2003,Shakarian_models_2015}. \imp{Throughout this paper, we consider a social network $G(V, E)$ that is directed, connected, and weighted \footnote{\footnotesize Note that the cases when the network is undirected or disconnected or unweighted can be treated similarly.}}, where $V = \{v_1,v_2,\dots, v_n\}$ is the node set, and $E = \{(v_i,v_j): \text{there is an edge from node $v_i$ to node $v_j$}\}$ is the edge set. Each edge $(v_i, v_j)$ is considered as a channel connecting nodes $v_i$ and $v_j$ along which the information flows, and can be associated with a weight $W_{ij}> 0$, for example, indicating the strength of the interaction or the level of trust between the agents; $W_{ij} = 0$ if $(v_j, v_j)\notin E$. \imp{Here, we use $x_i(t)\in\{0,1\}$ to represent the state value of node $v_i$}, either $0$ for being inactive or $1$ for being active, at each discrete time step $t\ge 0$. \imp{In both models, $x_i(t) \le x_i(t+1),~\forall v_i\in V$}, i.e.,~nodes can switch from being inactive to being active but not vice versa, thus the propagation is progressive. Hence, there is no feedback between nodes, as a node can not be influenced by others that it has influenced before.

\imp{Specifically, the LT model further requires $\sum_{i}W_{ij} \le 1,~\forall v_j\in V$}, and each node $v_j$ chooses a threshold $\theta_j$ uniformly at random from the range $[0,1]$, which represents the critical influence weight for node $v_j$ to be activated. Given a random choice of the thresholds, and an initial set of active nodes $\mathcal{A}_0$, the propagation process unfolds deterministically in discrete time steps, where at $t>0$, all nodes that are active in step $t-1$ remain active, and an inactive node $v_j$ will be activated if the total weight of its active neighbors is at least its threshold $\theta_j$,
\begin{equation}
    \sum_{v_i\in\mathcal{A}_{t-1}}W_{ij} \ge \theta_j,
	\label{equ:LT-classic_update}
\end{equation}
where set $\mathcal{A}_{t-1}$ contains the active nodes in step $t-1$. \imp{However, in the IC model, each $W_{ij}$ corresponds to the probability that an active node $v_i$ can influence its inactive neighbor $v_j$ in one step}, and each node only has a single chance to influence others when it first become active. If $v_j$ is successfully activated, it will have value $1$ in the next time step, but whether or not $v_i$ succeeds, it cannot further attempt to activate others in the subsequent rounds. More general cascade and threshold models have also been considered \cite{Kempe_influence_2003,Mossel_Roch_2010}. In particular, Kempe \textit{et al.} \cite{Kempe_influence_2003} considered variants of both models with feedback through constructing a multilayer network with each layer for each time step, but they required a predetermined depth of the propagation. We refer the reader to \cite{Chen_Info_2013} for a comprehensive survey on information propagation. 

\imp{In the two classic models, the influence spread is defined as the number of active nodes at the end of the process, $\lim_{t\to\infty}\sum_ix_i(t)$, and the IM problem is then to maximize it subject to limited number of nodes that one can activate at the beginning of the process, $\abs{\mathcal{A}_0}$} \cite{Banerjee_IMsurvey_2020,Li_IMsurvey_2018}. The IM problem under the two classic models is NP hard, and the key algorithmic breakthrough lies in the approximation guarantees for the greedy hill-climbing algorithms \cite{Kempe_influence_2003}. Subsequently, numerous methods have been proposed to further improve the efficiency of the greedy algorithms, maintaining the same approximation guarantees \cite{Goyal_CELF_2011,Leskovec_outbreak_2007}, or not exactly \cite{Borgs_IMoptimal_2014,Chen_scalable_2010,Wang_scalable_2012}. One vital assumption here is that the information spread is \textit{submodular}. Specifically, a function $f: P(U)\to \mathbb{R}^+\cup\{0\}$, where $P(U)$ is the power set of a finite set $U$, is submodular, if 
\begin{equation*}
    f(S\cup\{v\}) - f(S) \ge f(T\cup\{v\}) - f(T),
\end{equation*}
for all element $v\in U$ and $S\subseteq T\subseteq U$, \imp{i.e.,~the marginal gain from activating one more node initially is larger if the original set is smaller. However, this is not necessarily true when there are certain threshold effects}. For the LT model, the key correspondence lies in the uniform distribution of thresholds, and we can show that the influence spread under the LT model with deterministic thresholds is not submodular. Further with deterministic thresholds, the IM problem has been shown to be NP-hard to approximate within a factor of $n^{1-\epsilon}$ for any $\epsilon > 0$ \imp{where $n$ is the network size} \cite{Chen_approximability_2009,Kempe_influence_2015}. 

\imp{Finally, we note that there are models with continuous variables and feedback between nodes within the field of opinion dynamics \cite{degroot_model_1974,Deffuant_boundconf_2000}, but the associated IM problem is not a central theme there. Meanwhile, there are some continuous models analyzed within the context of IM, such as the fully linear models \cite{Even-Dar_voter_2007}, but they do not have well-established connections with the classic models and their intuitive mechanisms. We also note that there could be variants of the constraint in the IM problem when nodes take continuous values, e.g.,~on the sum of initial state values rather than the number of activated nodes \cite{Demanine_ContinInf_2014}}. However, in this paper, we are interested in the case when, e.g.,, companies have limited resources to convince more people to buy products, thus maintain the original constraint.

\section{\label{sec:general_diffusion_model}Information propagation}
In this section, we first extend the two classic information propagation models, the IC model and the LT model, for continuous state variables while allowing feedback between \imp{nodes} in Sec.~\ref{sec:extend_model}, and then propose a general class of information propagation model in Sec.~\ref{sec:diff_model}. We show the salient feature of the proposed model that it can be equivalent to the extended IC model for one end and the extended LT model for the other end in Sec.~\ref{sec:diff_props}, and its general properties via the corresponding differences in Sec.~\ref{sec:diff_props-gen}.

\subsection{\label{sec:extend_model}Extending the classic models}
We first extend the two classic models, the IC model and the LT model, to a continuous state space and a deterministic case. Here, we consider a continuous variable $x_j(t)\in \mathbb{R}$ to represent the state value of node $v_j$ at each discrete time step $t\ge 0$, \imp{which can be interpreted as the influence on node $v_j$ at $t$.} \imp{By using continuous variables, we assume that the influence is \textit{additive}, where, e.g.,, people could become more convinced of a piece of news if more friends believe it, or buy more products if more friends make a purchase, either at each time or over time.} Accordingly, a node $v_j$ is \textit{influenced} or \textit{active} at time step $t$ if $x_j(t) > 0$, and we represent the \textit{overall influence} on each node $v_j$ as \begin{equation}
    s_j = \sum_{t=1}^\infty(1-\gamma)^tx_j(t),
    \label{equ:dyn_gen_influ-j}
\end{equation}
where $\gamma\in[0,1)$ is a time-discounting factor which guarantees convergence. $\mathbf{x}(t) = (x_j(t))$ denotes the vector consisting of $x_j(t)$.

We start from the IC model. We assume that (i) the expected value is the actual \imp{influence on each node at each time step}, and (ii) the state values have the \textit{no-memory} property where the ability \imp{either to be influenced or to influence others} at the current time step $t$ is independent of its previous states, thus $x_j(t) = \sum_i W_{ij}x_i(t-1),\ \forall v_j\in V, t>0$. Hence, the extended independent cascade (EIC) model has the following updating function,
\begin{align}
	\mathbf{x}(t) = \mathbf{W}^T\mathbf{x}(t-1),\quad \forall t > 0,
	\label{equ:linear_dynamics}
\end{align}
where $\mathbf{W}$ is the (weighted) adjacency matrix of the network. We note that it is one type of linear dynamics on networks. The following condition on the time-discounting factor $\gamma$ and the spectral radius $\rho(\mathbf{W})$ is required to guarantee the convergence of the overall influence,
\begin{equation}
    \gamma > 1-1/\rho(\mathbf{W}).
    \label{equ:linear_dynamics_req}
\end{equation}


We now proceed to the LT model. Firstly, we maintain the linear activation strategy as in \eqref{equ:LT-classic_update}, \imp{where the linear product of each node's neighbors' state values and the edge weights is computed}. Secondly, with continuous variables, we can \imp{set the activated state value to be the threshold value, thus} control the source of nonlinearity to be only the activation. Note that the state values can then change magnitude over time, thus we impose time-dependent thresholds $\{\theta_{j, t}\}$. Thirdly, we also assume the no-memory property as before. Hence, the extended linear threshold (ELT) model has the following updating function, $\forall t > 0,\ v_j\in V$,
\begin{align}
	x_j(t) = 
	\begin{cases}
		\theta_{j, t},\quad &\sum_iW_{ij}x_i(t-1) \ge \theta_{j, t},\\
		0,\quad &\text{otherwise}.
	\end{cases}
	\label{equ:extreme2_update}
\end{align}


From another perspective in extending the LT model, we can maintain the magnitude of the state value at each time step, but instead of a single value of $1$ for being active, each node $v_j$ can take values in a range $[1, m_j]$ depending on how strong the \imp{influence attempts from its neighbors are}, $\sum_iW_{ij}x_i(t-1)$. Specifically, as in \eqref{equ:LT-classic_update}, a node $v_j$ starts to take positive state value if the sum is at least a threshold, denoted $l'_j$ here, but further, the state value increases from $1$ to the highest possible state value $m_j$ as the sum increases from $l'_j$ to a higher value $h'_j$. This gives a more direct extension to the LT model, and we name it the multi-valued linear threshold (MLT) model. Explicitly, it has the updating function,
\begin{align}
	\label{equ:dyn_gen_time-indp}
	x_j(t) &= f_{j}(\sum_iW_{ij}x_i(t-1)), \quad \forall t > 0,\ v_j\in V,\\
	\text{where, }  &f_{j}(x) = 
	\begin{cases}
		0, &\quad x < l'_{j},\\
		\frac{m_j - 1}{h'_j - l'_j}(x - l'_j) + 1, &\quad l'_{j} \le x < h'_{j},\\
		m_j, &\quad x \ge h'_{j},
	\end{cases} \nonumber
\end{align}
is the time-independent bound function, and $x_j(0) \in\{0\}\cup[1, h'_{j,0}]$ with $h'_{j,0}$ being the upper bound of node $v_j$'s initial state value. \imp{As we will show later, the two extensions to the LT model can be equivalent through their equivalences with the model we will propose later}. 

\subsection{\label{sec:diff_model}General class of information propagation model}
\begin{figure*}
    \centering
    \hspace*{-2em}
    \begin{tabular}{cccc}
         \includegraphics[width=.25\textwidth]{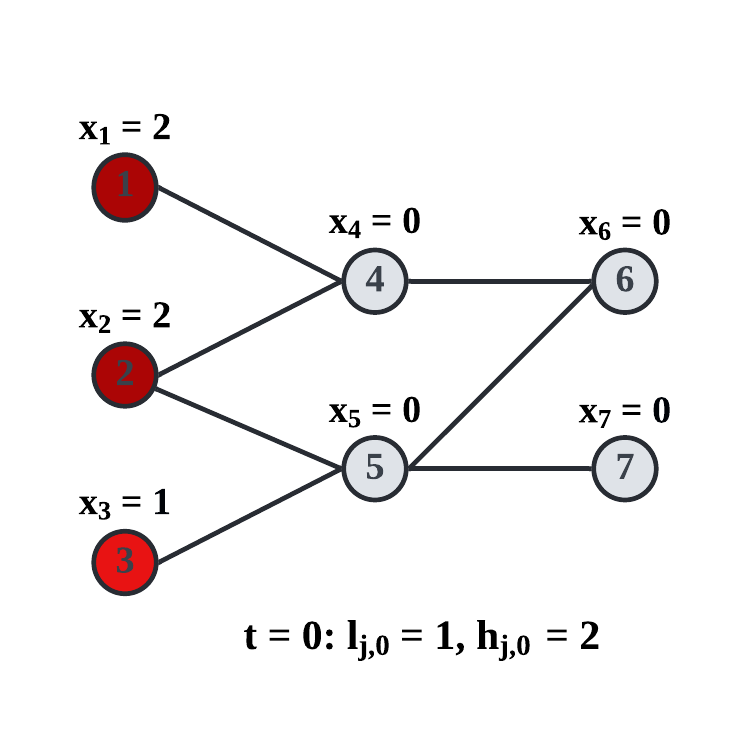} & \includegraphics[width=.25\textwidth]{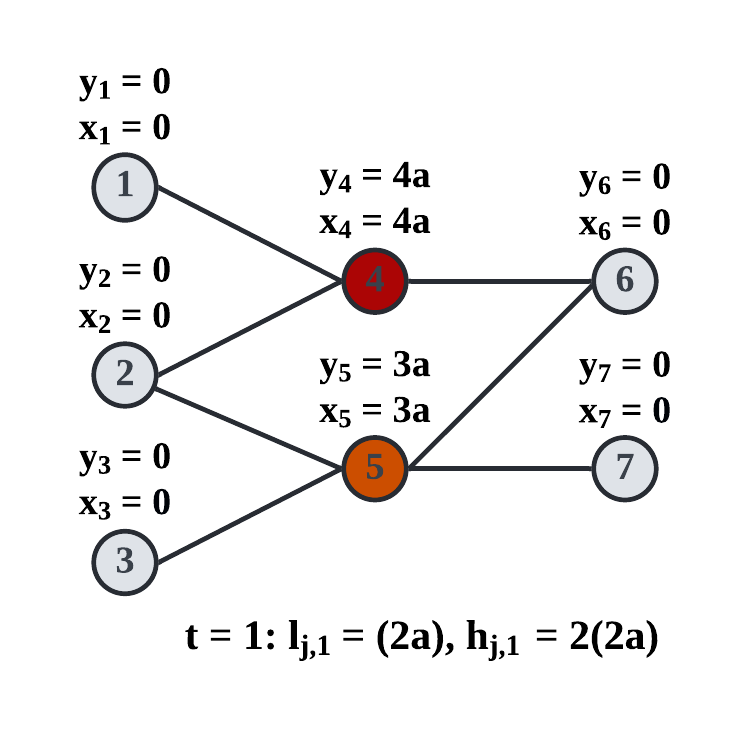} &          \includegraphics[width=.25\textwidth]{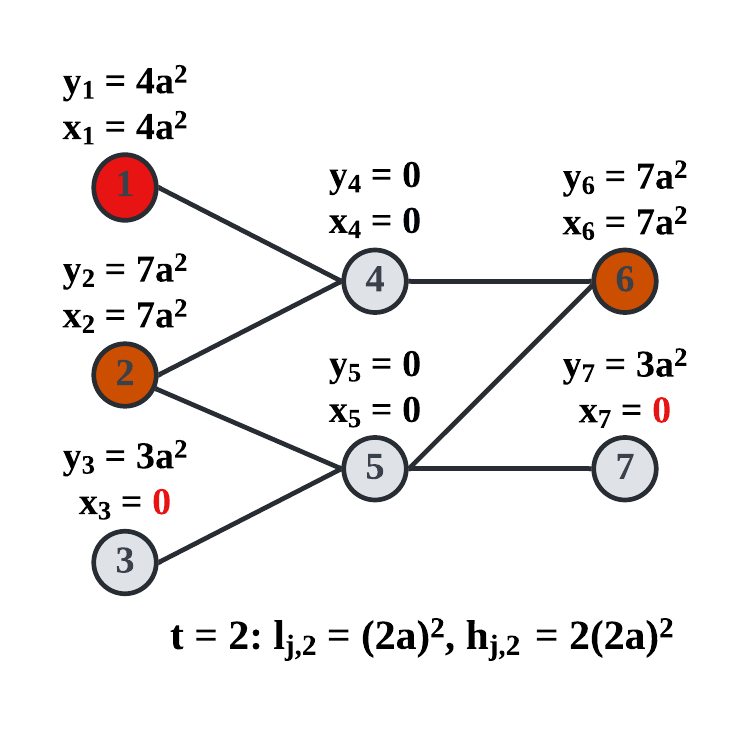} &          \includegraphics[width=.25\textwidth]{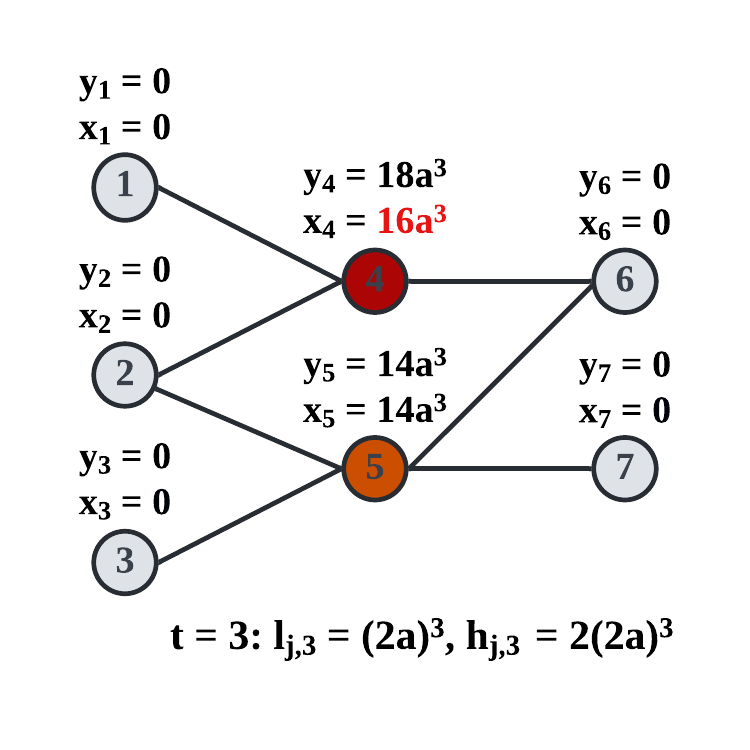}
    \end{tabular}
    \caption{Illustration of the information propagation in the first few steps following the GIP model with bounds $l_{j,t} = (2\alpha)^t, h_{j,t} = 2(2\alpha)^t,\ \forall v_j\in V$, where $y_j(t) = \sum_{i}W_{ij}x_i(t-1)$ is the linear product, $x_j(t) = f_{j,t}(y_j(t))$ is the state value, \imp{each edge in the network is bidirectional} and has uniform weight $\alpha$, and the node color indicates the state value at each time step.}
    \label{fig:illustr_model}
\end{figure*}
In this section, we propose a general class of information propagation (GIP) model unifying the mechanisms underlying the two classic models. (i) Each node $v_i$ can independently \imp{attempt to influence} its neighbors, proportional to the edge weight and its own state value $x_i(t)$, which is consistent with the IC model. (ii) The \imp{actual influence} on each node $v_j$ is based on the \textit{collective} behavior of the whole neighborhood, by applying a nonlinear transformation to $y_j(t) = \sum_i W_{ij}x_i(t-1)$, in order to capture how the \imp{accumulated influence attempts} from all neighbors transform into a change for the state of $v_j$, which is reminiscent of the LT model, but also of nonlinear models for opinion dynamics \cite{Srivastava_bifur_2010}. Specifically, we assume that at each $t>0$, there is a lower bound $l_{j,t}$ corresponding to the critical mass to trigger the propagation s.t.~$x_j(t) = 0$ if $y_j(t) < l_{j,t}$, and also an upper bound $h_{j,t}$ for the saturation effect s.t.~$x_j(t) = h_{j,t}$ if $y_j(t) \ge h_{j,t}$ \cite{Asllani_crowd_2018,Fanelli_crowd_2010} (see Fig.~\ref{fig:illustr_model} for a step-by-step illustration of the underlying process which we will discuss later). Explicitly, the GIP model is a bounded-linear dynamics, 
\begin{align}
	\label{equ:dyn_gen_revis}
	x_j(t) &= f_{j,t}(\sum_iW_{ij}x_i(t-1)), \quad \forall t > 0,\ v_j\in V,\\
	\text{where }  &f_{j,t}(x) = 
	\begin{cases}
		0, &\quad x < l_{j,t},\\
		x, &\quad l_{j,t} \le x < h_{j,t},\\
		h_{j,t}, &\quad x \ge h_{j,t},
	\end{cases} \nonumber
\end{align}
is the time-dependent bound function of each node $v_j$ (see Fig.~\ref{fig:illustr_boundf} for an example), $\mathbf{W} = (W_{ij})$ with $W_{ij} \ge 0$ is the (weighted) adjacency matrix of the underlying network, $\{l_{j,t}\}$ and $\{h_{j,t}\}$ are the time-dependent lower and upper bounds of each node $v_j$, respectively, with $0\le l_{j,t} \le h_{j,t}$. \imp{The bound values leaves extra freedom to characterize the underlying population, and we will show later that the GIP model can recover the classic models by setting specific bound values}. The initial states $\mathbf{x}(0)$ are given, with $x_j(0)\in \{0\}\cup[l_{j,0}, h_{j,0}]$ and $l_{j,0} > 0$.
\begin{figure}
    \centering
    \includegraphics[width=.35\textwidth]{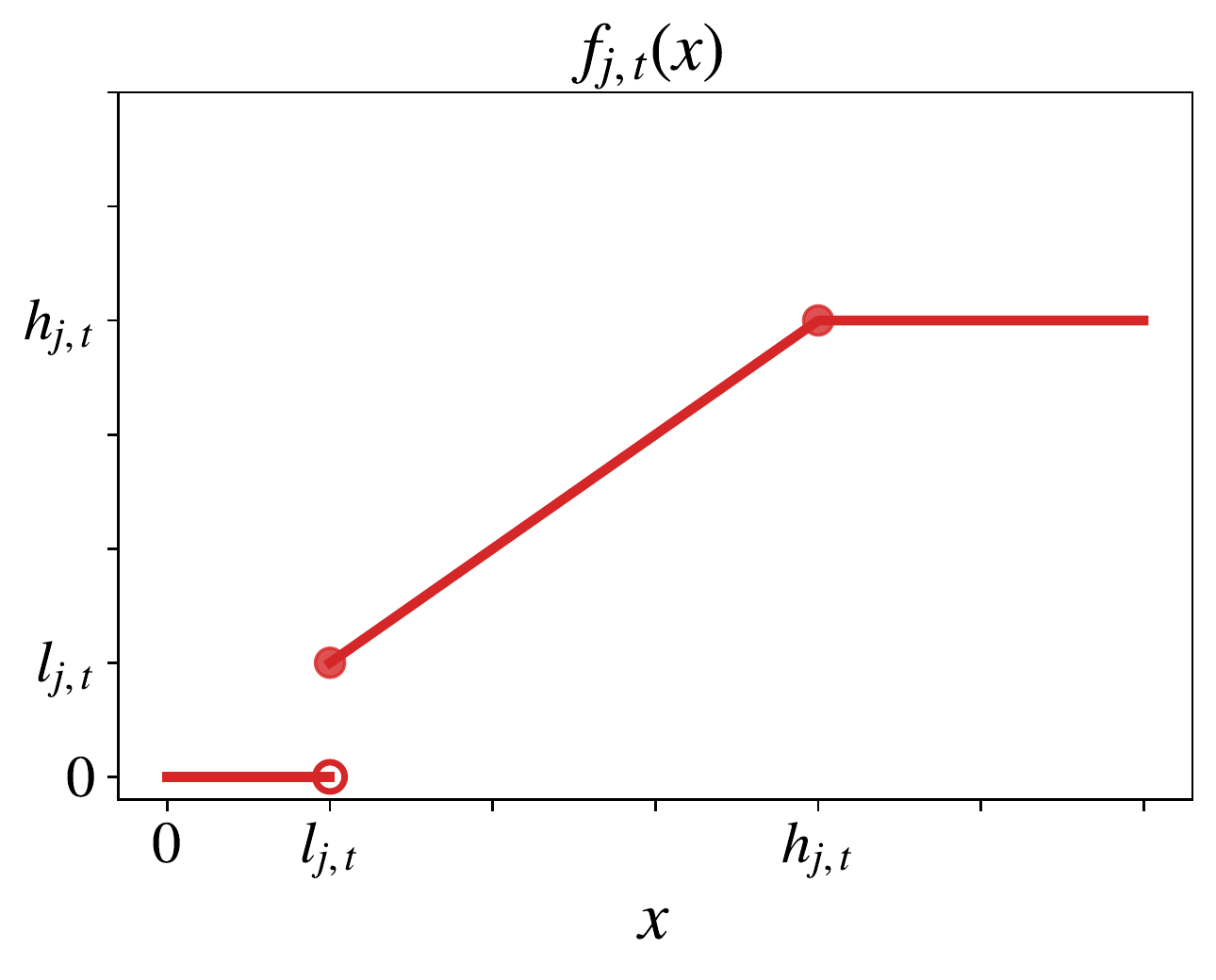}
    \caption{Example bound function $f_{j,t}$ of node $v_j$ at time step $t$ where $h_{j,t} = 4l_{j,t} > 0$.}
    \label{fig:illustr_boundf}
\end{figure}


In order to interpret the GIP model and the underlying process more intuitively, we construct a small social network with seven agents, \imp{each edge being bidirectional} and of uniform weight $\alpha = 0.4$; see Fig.~\ref{fig:illustr_model}. For illustrative purposes, we apply the bounds $l_{j,t} = 0.8^t$ and $h_{j,t} = 2\times 0.8^t$, $\forall t\ge 0,\ v_j\in V$, set $\gamma = 0$, and activate nodes $v_1, v_2$ with value $2$ and node $v_3$ with value $1$ at $t=0$. Then at $t=2$, $v_4$ can independently influence $v_1$, consistent with the IC model, as a result of its high state value, while collective effort is needed to influence $v_7$ so that it does not have positive state value, consistent with the LT model; see Fig.~\ref{fig:illustr_model}. The coexistence of the features in both models is necessary since a social network can have people with heterogeneous levels of activity, where people of high activity are more likely to activate others. Moreover, there is a positive feedback among the nodes $v_2, v_4, v_5, v_6$, as they reinforce their states over time; see Fig.~\ref{fig:diff_model-illustr}. This corresponds to the fact that groups of close friends keep receiving positive feedback from each other, thus reinforcing the information. 
\begin{figure}
	\centering
	\includegraphics[width=.35\textwidth]{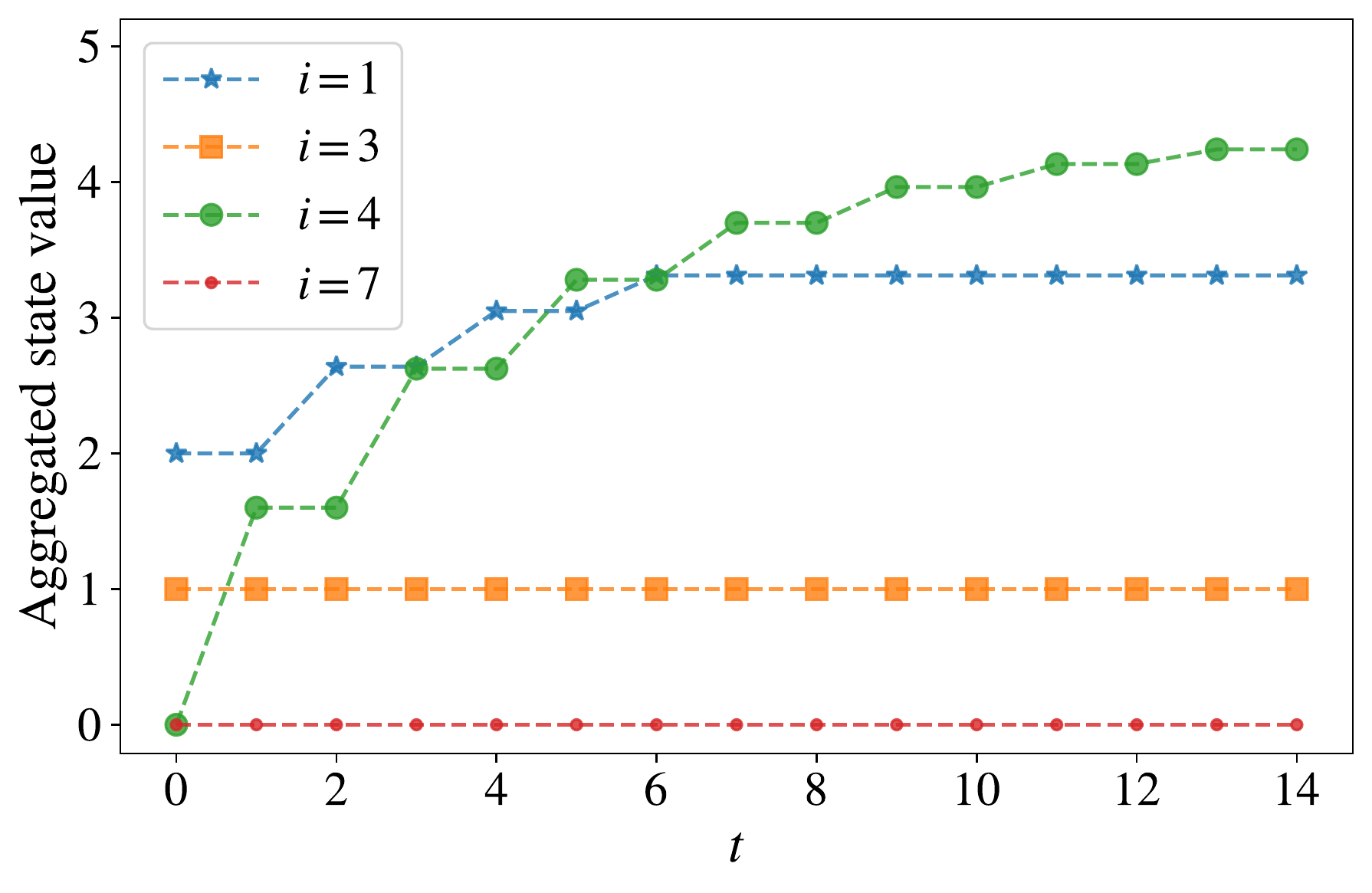}
	\caption{The change of the sum of \imp{influence on} selected nodes along time $t$ (bottom) on the network in Fig.~\ref{fig:illustr_model} with uniform weight $\alpha=0.4$ and $\gamma = 0$.}
	\label{fig:diff_model-illustr}
\end{figure}


\subsection{\label{sec:diff_props}The unifying properties}
In this section, we illustrate \imp{the unifying feature of the GIP model where it possesses the EIC and ELT models as limiting cases.} For one end, the GIP model is equivalent to the EIC model when \imp{all upper bounds are sufficiently large while all lower bounds are sufficiently small} (see the detailed proof in Appendix \ref{app:proofs_model}). 
\begin{lemma}
	If $l_{j,t} \le l_{\min, 0}w^t \le h_{j,t},\ \forall t > 0,\ v_j\in V$, where $l_{\min, 0} = \min_jl_{j,0}$ and $w = \min_{ij: W_{ij}>0}W_{ij}$, in the GIP model, then there is no threshold effect from the lower bounds,
	$\forall t > 0,\ v_j\in V$ s.t. $\sum_{i}W_{ij}x_i(t-1) > 0$, 
	\begin{equation}
		\sum_{i}W_{ij}x_i(t-1) \ge l_{j,t}.
		\label{equ:special_weightsum}
	\end{equation} 
	\label{lem:linear_equivalence-lower}
\end{lemma}	
\begin{theorem}
	If $l_{j,t} \le l_{\min, 0}w^t \le \mathbf{h}_{0}^T\mathbf{W}_{:, j}^t \le h_{j,t}$, $\forall t > 0, v_j\in V$, where $\mathbf{h}_0 = (h_{j,0})$, $\mathbf{W}_{:, j}^t$ is the $j$-th column of $\mathbf{W}^t$, and $l_{\min,0},w$ are the same as in Lemma \ref{lem:linear_equivalence-lower}, the GIP model is equivalent to the EIC model.
	\label{the:linear_equivalence}
\end{theorem}

For the other end, the GIP model is equivalent to the ELT model if all upper bounds are equal to the corresponding lower bounds, $l_{j,t} = h_{j,t} = \theta_{j,t},\ \forall t>0, v_j\in V$. To further disentangle the relationships, we \imp{introduce the upper and lower bound thresholds, $\theta_{l,j}, \theta_{h,j}$ ($0\le \theta_{l,j}\le \theta_{h,j}$), respectively, which are time-independent, as in Eq.~\eqref{equ:LT-classic_update} in the classic LT model}, and propose the following \textit{threshold-type} bounds, where $\forall t > 0,\ v_j\in V$,
\begin{align}
	\begin{split}
		l_{j, t} &= (\theta_{l,j}\alpha)^{t}l_{j,0},\\
		h_{j, t} &= \theta_{h,j}\theta_{l,j}^{t-1}\alpha^{t}h_{j,0},
	\end{split}
	\label{equ:dyn_gen_thres-bBt}
\end{align}
\imp{and the time evolution of the bounds is encoded through different powers of the mean weight $\alpha = \sum_{(v_i, v_j)\in E}W_{ij}/\abs{E}$ (and $\theta_{l,j}$)} \footnote{\footnotesize With the threshold-type bounds, the condition $\sum_iW_{ij}x_i(0) \ge l_{j,1}$ at $t=1$ is equivalent to $\sum_i(W_{ij}/\alpha)(x_i(0)/l_{j,0}) \ge \theta_{l,j}$. Hence, $\alpha$ will not affect the activation so long as the relative weight $W_{ij}/\alpha$ does not change (e.g.,, $W_{ij}/\alpha = 1$ if the network has uniform edge weight).}. With such bounds, the GIP model is equivalent to the ELT model if $l_{j,0} = h_{j,0}$ and $\theta_{l,j} = \theta_{h,j},\ \forall v_j\in V$. Furthermore, we can show that the GIP model can also be equivalent to the MLT model with the help of such bounds. 
\begin{theorem}
	If the threshold-type bounds have uniform thresholds s.t. $\forall v_j\in V$,
	\begin{align}
	    \theta_{l,j} = \theta_l, \theta_{h,j} = \theta_h,
	    \label{equ:dyn_gen_thres-uni}
	\end{align}
	the GIP model with such bounds and $l_{j,0} = 1,\ \forall v_j\in V$, is equivalent to the MLT model with $l'_j = \theta_{l}\alpha$, $h'_j = \theta_{h}\alpha h_{j,0}$, $m_j = (\theta_{h}h_{j,0})/\theta_{l}$, and $h'_{j,0} = h_{j,0},\ \forall v_j\in V$, in terms of the overall influence where if we denote the time-discounting factors for the GIP model and the MLT model as $\gamma, \gamma'$, respectively, we set $\gamma' = 1 - (1 - \gamma)\theta_{l}\alpha$. Specifically, if we denote the state values from the GIP model by $x_j(t)$ and those from the MLT model as $x'_j(t)$, 
	\begin{align}
	    x_j(t) = (\theta_l\alpha)^tx'_j(t),\quad \forall t\ge0,\  v_j\in V. 
	    \label{equ:dyn_gen_LT-equ}
	\end{align}
	\label{the:LT_equivalence}
\end{theorem}
Further, if the network has uniform weight $\alpha$, then setting $h'_j = l'_j = \theta_{l}\alpha$ in the MLT model is equivalent to requiring $\theta_l$ neighbors to have positive state values for activations, which is reminiscent of the constant threshold model \cite{Banerjee_IMsurvey_2020}. The equivalence can pass to the GIP model and the ELT model. Hence, $\theta_l = \theta_h = 1$ corresponds to simple contagions, where a node can be influenced by a single active neighbor, and $\theta_l = \theta_h > 1$ corresponds to complex contagions, where collective effort from the neighborhood is required to influence a node. Due to the desired correspondence illustrated here, we will consider exclusively the threshold-type bounds hereafter.


\subsection{\label{sec:diff_props-gen}The general properties}
\imp{We now illustrate the general properties of the GIP model by deviating it from the two limiting cases, the EIC model and the ELT model, where it can incorporate the features of each other}. We start from the imposed threshold effect on top of the EIC model (i.e.,~linear dynamics) when increasing the lower bounds, where reinforcement within the neighborhood will be necessary to activate nodes, thus to reach a higher influence on the whole network. We show it through the lens of stochastic block models (SBMs), specifically a two-block planted SBM, $SBM(p_{in}, p_{out})$, where it has two communities and the probabilities for an edge to occur inside each community and between the two communities are $p_{in}$ and $p_{out}$, respectively (see Appendix \ref{app:proofs_model} for the detailed proof).  
\begin{claim}
	With $l_{j,0} = h_{j,0} = l_0,\ \forall v_j\in V$, and $\{h_{j,t}\}$ as in Theorem \ref{the:linear_equivalence} in the GIP model, $SBM(p_{in}, p_{out})$ with two equally sized communities, $\mathcal{B}_1, \mathcal{B}_2$, and uniform weight $\alpha$ has the following properties at $t=1$:
	\begin{enumerate}
		\item when $l_{j,1} \le l_1^* = l_0\alpha,\ \forall v_j\in V$, the expected influence $\mathbb{E}[\sum_j(1-\gamma)x_j(1)]$ \footnote{\footnotesize Here we consider the distribution of the SBM, and the expectation is take over this distribution.} from initially activated node set (i) $\mathcal{A}_0 = \{v_{i_1}, v_{i_2}\} \subset \mathcal{B}_1$, is the same as that from (ii) $\mathcal{A}_0 = \{v_{j_1}, v_{j_2}\}$ with $v_{j_1}\in \mathcal{B}_1, v_{j_2}\in \mathcal{B}_2$; 
		\item when $l_1^* < l_{j,1} \le 2l_1^*,\ \forall v_j\in V$ \footnote{\footnotesize{In the specific case here, the upper bound $2l_1^*$ is equivalent to require that at most two initially activated neighbors are needed to activate a node.}}, if 
		\begin{align}
    		p_{in} \ne p_{out},
    		\label{equ:extreme1_devSBM-cond}
		\end{align}
		the expected influence $\mathbb{E}[\sum_j(1-\gamma)x_j(1)]$ from set (i) is larger than that from set (ii). 
	\end{enumerate}
	\label{cla:linear_difference}
\end{claim}

From the opposite end, we now demonstrate the added locally linear effect on top of threshold models when increasing the upper bounds, where a single active node is able to activate its neighbors at certain time steps in the propagation process. We consider it through the threshold-type bounds \eqref{equ:dyn_gen_thres-bBt} in Claim \ref{cla:extreme2_devi-tree} with detailed proofs in Appendix \ref{app:proofs_model}, and leave further discussions of the GIP model's general features, including the limited depth such single-source activation can lead to, to Appendix \ref{app:diff_coexistence}.
\begin{claim}
	 With $l_{j,0} = h_{j,0} = l_0,\ \forall v_j\in V$, and the threshold-type bounds (\ref{equ:dyn_gen_thres-bBt}) satisfying (\ref{equ:dyn_gen_thres-uni}) applied in the GIP model, suppose at a particular time $t'\ge 0$, there is a treelike structure of the active nodes $\mathcal{A}_{t'} = \{v_i: x_i(t') > 0\}$ and some currently inactive node $v_{j_*}$ s.t.  
	\begin{align*}
		\exists! v_{j_0}\in\mathcal{A}_{t'}\ s.t.\ W_{j_0j_*}>0,
	\end{align*} 
	where $\mathbf{W}$ is the (weighted) adjacency matrix. Then if the network has uniform weight $\alpha$,
	\begin{enumerate}
		\item when $\theta_h = \theta_l > 1$, node $v_{j_*}$ can never have positive state value at $t = t'+1$;
		\item when $\theta_h > \theta_l > 1$, node $v_{j_*}$ can have positive state value at $t = t'+1$ given a sufficiently large $\theta_h$.
	\end{enumerate}
	\label{cla:extreme2_devi-tree}
\end{claim}
As in Claim \ref{cla:linear_difference}, we can consider the case when increasing $\theta_h$ from $2$ to $4$ while maintaining $\theta_l = 2$ in a two-block planted SBM, and compare the performance of the following two larger sets of initially activated nodes: (i) $\mathcal{A}_0 = \{v_{i_1}, v_{i_2}, v_{i_3}, v_{i_4}\} \subset \mathcal{B}_1$; (ii) $\mathcal{A}_0 = \{v_{j_1}, v_{j_2}, v_{j_3}, v_{j_4}\}$ with $v_{j_1}, v_{j_2}\in \mathcal{B}_1$ and $v_{j_3}, v_{j_4}\in \mathcal{B}_2$. We show that the increase in the expected influence (in one time step) from (i) is more than that from (ii) in Appendix \ref{app:SBM}, which is consistent with Claim \ref{cla:extreme2_devi-tree} since nodes in $\mathcal{B}_1$ have a higher probability to reach a higher or the highest state value in the propagation from (i). \imp{This locally linear effect corresponds to the fact that people with high activity are more likely to influence their friends in social networks}.

\section{\label{sec:influence_maximization} Influence maximization}
Now, we proceed to a key algorithmic problem associated with information propagation, the \textit{influence maximization} (IM) problem, i.e., to maximize the overall influence on the nodes at the end of the process, and here we are interested in the constraint of a limited number of initially activated nodes, determined by the \textit{budget size}, corresponding to the limited resources to influence more people as discussed in Sec.~\ref{sec:related_work}. In this section, we will first introduce a general formulation of the IM problem in Sec.~\ref{sec:IM_general_problem}, and give \imp{general solution methods} to this task in Sec.~\ref{sec:IM_general_solution}. In these two sections, we focus on the general features that the IM problem could have in practice. Then we turn to the special cases when the dynamics are governed by the GIP model in Sec.~\ref{sec:IM_special_cases}, and further propose a customized algorithm in Sec.~\ref{sec:IM_customized_direct_search}.

\subsection{\label{sec:IM_general_problem} Problem formulation}
With a given information propagation process, and a given function $s_j(\cdot)$ for the overall influence on each node $v_j$, the overall influence on the whole network is naturally, 
\begin{align}
	s(\mathbf{x}(0)) =  \sum_{j}s_j(\mathbf{x}(0)),
	\label{equ:dyn_gen_obj}
\end{align} 
where $\mathbf{x}(0)$ is the initial state vector. \imp{For example, the function for individual influence can encode the state of the nodes at the end of the propagation process, $\lim_{t\to\infty}x_j(t)$, and then the IM problem recovers the one in the literature associated with the classic models, while it can also be Eq.~\eqref{equ:dyn_gen_influ-j} associated with the GIP model we proposed in Sec.~\ref{sec:general_diffusion_model}}. The IM problem is then to maximize $s(\mathbf{x}(0))$ with respect to the $\mathbf{x}(0)$, subject to the constraint of limited budget size, 
\begin{align}
	\abs{\{v_j: x_j(0) > 0\}} \le k,
	\label{equ:dyn_gen_constr}
\end{align}	
where $k\in\mathbb{Z}^+$ is the budget size.

With objective (\ref{equ:dyn_gen_obj}) and constraint (\ref{equ:dyn_gen_constr}), we then formulate the IM problem as a \textit{mixed-integer nonlinear programming} (MINLP), 
\begin{align}
	\begin{split}
		\max_{\mathbf{x},\mathbf{z}} \quad &s(\mathbf{x})\\
		s.t. \quad & x_j \le h_{j,0}z_j,\\
		& x_j \ge l_{j,0}z_j,\\
		&\sum_j z_j \le k,\\
		&x_j\in \mathbb{R},\ z_j\in \{0,1\},\ \forall j,
	\end{split}
	\label{equ:opt_dyn_gen}
\end{align}
where $0 < l_{j,0}\le h_{j,0}$ restrict the initial level of influence on node $v_j$, $k\in \mathbb{Z}^+$ is the budget size, and the objective function $s(\cdot)$ is the overall influence on the whole network as in Eq.~(\ref{equ:dyn_gen_obj}). The variables in vector $\mathbf{x}$ correspond to the initial states, while the extra variables in vector $\mathbf{z}$, of the same dimension, \imp{correspond to whether to set positive initial state values ($z_j = 1$) or not ($z_j = 0$),} and are added to appropriately impose the constraint (\ref{equ:dyn_gen_constr}).

The difficulty of the optimization problem lies in the objective function $s(\mathbf{x})$. Take the GIP model together with the function \eqref{equ:dyn_gen_influ-j} for individual influence as an example. (i) $s(\mathbf{x})$ is not always smooth and even discontinuous, since each $f_{j,t}(x)$ in (\ref{equ:dyn_gen_revis}) can be nonsmooth at $h_{j,t}$ and discontinuous at $l_{j,t}$. (ii) A closed-form of $s(\mathbf{x})$ cannot be obtained generally, except when $f_{j,t}(x) = x,\ \forall t > 0,\ v_j\in V$, in (\ref{equ:dyn_gen_revis}). (iii) The derivative information is rarely very useful in finding a maximal point, as discussed in Appendix \ref{app:diff_derivative}. However, even in this case, we can show that the evaluation of the objective function can be solved \imp{efficiently}, as in Theorem \ref{the:evaluation_complexity}. This is a bonus in the deterministic setting \cite{Lu_compldLT_2012}. Hence, it is necessary to treat the objective function as an input-output (black-box) system and resort to derivative-free methods (DFMs) for general solutions, which we will discuss in the following section. 
\begin{theorem}
	Given a network $G(V, E)$ with the weight matrix $\mathbf{W}$ and an initial state $\mathbf{x}(0)$, and with the GIP model governing the information propagation process and Eq.~\eqref{equ:dyn_gen_influ-j} as the function for individual influence, the problem of computing the objective function $s(\mathbf{x}(0))$ in the MINLP \eqref{equ:opt_dyn_gen} [i.e., Eq.~\eqref{equ:dyn_gen_obj}] can be solved in $O(|E|t_{\epsilon})$ time, where $t_\epsilon$ is the number of time steps required for the convergence with tolerance $\epsilon > 0$, where
	\begin{equation*}
	    \abs{(1-\gamma)^t\mathbf{x}(t)} < \epsilon, \quad \forall t\ge t_\epsilon.
	\end{equation*}
	\label{the:evaluation_complexity}
\end{theorem}
\begin{proof}
	The time complexity follows from Algorithm \ref{alg:influ_evaluate}. In each iteration $t$, each nonzero element of the weight matrix $\mathbf{W}$ has only one chance to be used to potentially adjust the state value $\mathbf{x}^{(t)}$, and there are overall $O(\abs{E})$ such elements. Therefore, the time complexity of each iteration is $O(\abs{E})$, and the overall evaluation has time complexity $O(\abs{E}t_\epsilon)$, dependent on the number of steps towards convergence, $t_\epsilon$.  
\end{proof}
\begin{algorithm}[H]
	\caption{Influence evaluation.}
	\label{alg:influ_evaluate}
	\begin{algorithmic}[1]
		\State Input: A network $G(V,E)$ with its weight matrix $\mathbf{W}$ where $W_{ij} > 0$ if $(v_i,v_j)\in E$, parameters $\{l_{j,t}\},\ \{h_{j,t}\}$ in the GIP model, time-discounting factor $\gamma$, the initial state $\mathbf{x}(0)$ where $x_j(0) \in [l_{j,0}, h_{j,0}]$ if and only if $v_j\in\mathcal{A}_0$ ($0$ otherwise), and the tolerance $\epsilon$.
		\State Output: The value of the objective function in the MINLP (\ref{equ:dyn_gen_obj}), $s$. 
		\State Set $t \leftarrow 0$, $\mathbf{x}^{(0)} \leftarrow \mathbf{x}(0)$, and $s \leftarrow 0$.
		\State Mark all the out-neighbors of $\mathcal{A}_0$ as potentially activated nodes, $\mathcal{N}_0 \leftarrow \bigcup_{v_j\in\mathcal{A}_0}\mathcal{N}^{out}(v_j)$.
		\While{$\abs{(1-\gamma)^t\mathbf{x}^{(t)}} > \epsilon$}
		\State{$\mathcal{A}_{t+1}, \mathcal{N}_{t+1} \leftarrow \emptyset$, and $\mathbf{x}^{(t+1)} \leftarrow \mathbf{0}$;}
		\For{each potentially activated node $v_j\in\mathcal{N}_{t}$}
		\State{$x^{(t+1)}_j = f_{j, t}(\sum_{i\in\mathcal{A}_{t}}W_{ij}x^{(t)}_i)$;}
		\If{$x^{(t+1)}_j > 0$}
		\State{$\mathcal{A}_{t+1} \leftarrow \mathcal{A}_{t+1} \cup \{v_j\}$;}
		\State{$\mathcal{N}_{t+1}\leftarrow \mathcal{N}_{t+1}\cup \mathcal{N}^{out}(v_j)$;}
		\State{$s \leftarrow s + (1-\gamma)^{t+1}x^{(t+1)}_j$;}
		\EndIf
		\EndFor
		\State{$t\leftarrow t+1$;} 
		\EndWhile
	\end{algorithmic}
\end{algorithm}

\subsection{\label{sec:IM_general_solution}\imp{General solution methods}}
\imp{There are two main classes of methods in DFMs, \textit{model-based} methods and \textit{direct-search} methods.}
Since we cannot assume the objective to fall in a simple family, e.g.,~polynomials, model-based methods are not appropriate in this problem. Among the direct-search algorithms, the mesh adaptive direct search (MADS) method is one of the few that has local convergence analysis when the objective function is not necessarily Lipschitz continuous \cite{Vicente_discont_2012}. Therefore, we consider the MADS for mixed variables (MV) \cite{Abramson_MADSmv_2009} as a general solution to the IM problem, which can be implemented by the software NOMAD \cite{Audet_NOMAD4_2021,LeDigabel_NOMAD_2011}; see Appendix \ref{app:MADS} for a brief overview. 

To understand local convergence, here we introduce the important notion of \textit{local optimality} for mixed variables, and accordingly, the \textit{local neighborhood}. We partition each vector into its continuous and discrete components, $\mathbf{y} = (\mathbf{y}^c, \mathbf{y}^d)\in\Omega$, where $\Omega$ is the domain. \imp{For the MINLP \eqref{equ:opt_dyn_gen}, $\mathbf{y}^c = \mathbf{x}$ and $\mathbf{y}^d = \mathbf{z}$}. For the continuous variables of maximum dimension $n^c$, the neighborhood is well defined as the open ball, $B_{\epsilon}(\mathbf{y}^c) = \{\mathbf{y}_1^c\in\mathbb{R}^{n^c}: \norm{\mathbf{y}_1^c - \mathbf{y}^c} < \epsilon\}$ with $\epsilon > 0$. However, different notions of the discrete neighborhood exist. One common choice for integer variables is $\mathcal{N}(\mathbf{y}) = \{\mathbf{y}_1\in\Omega: \mathbf{y}_1^c = \mathbf{y}^c, \norm{\mathbf{y}_1^d - \mathbf{y}^d}_1\le 1\}$. With a user-defined discrete neighborhood, the classical definition of local optimality can be extended to mixed variable domains as follows. 
\begin{definition}
	A point $\mathbf{y} = (\mathbf{y}^c; \mathbf{y}^d)\in\Omega$ is said to be a local maximizer of a function $f$ on $\Omega$ with respect to the set of neighbors $\mathcal{N}(\mathbf{y})\subset\Omega$ if there exists an $\epsilon > 0$ such that $f(\mathbf{y})\ge f(\mathbf{y}_2)$ for all $\mathbf{y}_2$ in the set 
	\begin{align*}
		\Omega\cap\left(\bigcup_{\mathbf{y}_1\in\mathcal{N}(\mathbf{y})}B_{\epsilon}(\mathbf{y}_1^c)\times \mathbf{y}_1^d\right),
	\end{align*}
	where $B_{\epsilon}(\mathbf{y}^c) = \{\mathbf{y}_1^c\in\mathbb{R}^{n^c}: \norm{\mathbf{y}_1^c - \mathbf{y}^c} < \epsilon\}$ with $\epsilon > 0$ is an open ball, and $\mathcal{N}(\mathbf{y})$ is a user-defined discrete neighborhood.
	\label{def:local_maximizer}
\end{definition}

As mentioned before, MADS is among the few algorithms that can relax the assumptions for convergence analysis to include discontinuous functions. To conclude the \imp{overview of applicable DFMs}, we mention that there are also many heuristic algorithms \cite{Laguna_scatsearch_2014} but without theoretical performance guarantee, and refer the reader to the work of Boukouvala \textit{et al.} \cite{Boukouvala_GOMINLPCDFM_2016} for a thorough review of DFMs in conjunction with MINLP problems.

\subsection{\label{sec:IM_special_cases}Special cases}
In the previous sections, we have analyzed the highly general features of the IM problem, and given general solution methods accordingly. Hereafter, we turn our attention to the IM problem with the GIP model governing the information propagation process and Eq.~\eqref{equ:dyn_gen_influ-j} as the function for individual influence. In this section, we consider two special cases of the the GIP model, in order to shed light on other more general cases.

The first special case is when the lower bounds, $\{l_{j,t}\}$, are sufficiently small in the GIP model, where we can show that the objective function is continuous and concave with respect to the continuous variables $\mathbf{x}$ as in Theorem \ref{the:spe_contin-concave} (see Appendix \ref{app:proofs_IM} for the detailed proof). In this case, any local maximum is a global maximum, therefore the MADS method can have global convergence, though it is only with respect to the continuous variables since the optimality of the integer part is still local, from Definition \ref{def:local_maximizer}.
\begin{theorem}
	If $\{l_{j,t}\}, \{h_{j,t}\}$ are as in Lemma \ref{lem:linear_equivalence-lower}, then the objective function $s(\cdot)$ in the MINLP (\ref{equ:opt_dyn_gen}) is continuous and concave w.r.t.~the continuous variables $\mathbf{x}$. 
	\label{the:spe_contin-concave}
\end{theorem}


The other special case is when not only $\{l_{j,t}\}$ are sufficiently small but $\{h_{j,t}\}$ are sufficiently large in the GIP model, i.e.,~the extreme of the EIC model as in Sec.~\ref{sec:diff_props}. In this case,
\begin{align*}
	\mathbf{x}(t) = \mathbf{W}^T\mathbf{x}(t-1) = \left(\mathbf{W}^T\right)^t\mathbf{x}(0), 
\end{align*}
and the objective function is then, 
\begin{align}
    \begin{split}
        s(\mathbf{x}(0)) 
        &= \sum_{j}\sum_{t=1}^{\infty}(1-\gamma)^t x_j(t)\\
        &= \sum_{t=1}^{\infty}\mathbf{1}^T\left((1-\gamma)\mathbf{W}^T\right)^t\mathbf{x}(0)\\ 
        &= \mathbf{1}^T\left(\left(\mathbf{I} - (1-\gamma)\mathbf{W}^T\right)^{-1} - \mathbf{I}\right)\mathbf{x}(0) = \mathbf{c}^T\mathbf{x}(0),
    \end{split}
	\label{equ:dyn_lineard_obj}
\end{align}
where $\mathbf{c} = \{[\mathbf{I} - (1-\gamma)\mathbf{W}]^{-1} - \mathbf{I}\}\mathbf{1}$ is the Katz centrality with factor $(1-\gamma)$, $\mathbf{I}$ is the identity matrix, and the penultimate equation is obtained given that condition (\ref{equ:linear_dynamics_req}) is true in this extreme. Hence, the objective function is linear, thus (Lipschitz) continuous, concave and smooth. The exact solution(s) in this case is achievable as in Theorem \ref{the:lineard_exact_sol}, and we defer the detailed proof to Appendix \ref{app:proofs_IM}. 
\begin{theorem}
	When $\{l_{j,t}\},\ \{h_{j,t}\}$ are as in Theorem \ref{the:linear_equivalence}, then the exact solution(s) to the MINLP (\ref{equ:opt_dyn_gen}) is 
	\begin{align}
	    x^*_j = 
	    \begin{cases}
	        h_{j,0},\ &\text{if } j\in \mathcal{A},\\
	        0,\ &\text{otherwise},
	    \end{cases}
	    \ z^*_j = 
	    \begin{cases}
	        1,\ &\text{if } j\in \mathcal{A},\\
	        0,\ &\text{otherwise},
	    \end{cases}
	    \label{equ:lineard_exact_sol}
	\end{align}
	where $\mathcal{A} = \{j_1, ..., j_k\}$ s.t. $h_{i,0}c_i \le h_{j,0}c_{j},\ \forall i\notin\mathcal{A}, j\in\mathcal{A}$, $\mathbf{c} = \{[\mathbf{I} - (1-\gamma)\mathbf{W}]^{-1} - \mathbf{I}\}\mathbf{1}$ is the Katz centrality with factor $(1-\gamma)$, and the uniqueness of the solution depends on the uniqueness of set $\mathcal{A}$. 
	\label{the:lineard_exact_sol}
\end{theorem}
Hence, the exact solution(s) when the GIP model is in the extreme of the EIC model (i.e.,~linear dynamics) is to activate the $k$ nodes of the highest product of its Katz centrality and its maximum initial value. This relates the IM problem to a well-studied centrality measure in networks, the Katz centrality. Furthermore, this solution can serve as a warm start in the following search algorithm for the MINLP \eqref{equ:opt_dyn_gen}, as what we will do in the following section, with the search depth potentially proportional to the distance of the underlying propagation from the linear dynamics.

\subsection{\label{sec:IM_customized_direct_search} Customized direct search method}
Here, we exploit one feature of the objective, that $s(\mathbf{x})$ is nondecreasing in $\mathbf{x}$, which is inherited in the proof of Theorem \ref{the:lineard_exact_sol} \imp{(see Appendix \ref{app:proofs_IM})}, and propose a customized direct search method for the MINLP \eqref{equ:opt_dyn_gen} accordingly.

Because of this feature, maximizing the objective $s(\mathbf{x})$ with respect to $\mathbf{x}$ and $\mathbf{z}$ in the MINLP (\ref{equ:opt_dyn_gen}), is equivalent to the maximization with $\mathbf{x}$ and $\mathbf{z}$ at their highest possible values, particularly $x_j = h_{j,0}z_j$ and $\sum_j z_j = k$. Hence the problem is effectively reduced to the following problem only w.r.t.~the binary vector $\mathbf{z}$,
\begin{align}
	\begin{split}
		\max_{\mathbf{z}} \quad &s(\mathbf{h}_0\odot \mathbf{z})\\
		s.t. \quad &\sum_j z_j = k,\\
		&z_j\in \{0,1\}, \forall j,
	\end{split}
	\label{equ:opt_dyn_gen_comb}
\end{align}
where $\mathbf{h}_0 = (h_{j,0})$, and $\odot$ denotes the element-wise (Hadamard) product. Then the domain $\Omega^d$ is a natural mesh to search at each iteration $r$, 
\begin{align}
	M_r = \Omega^d = \{\mathbf{z}\in \{0,1\}^n: \sum_j z_j = k\}.
	\label{equ:opt_spec_mesh}
\end{align}
The constraints are incorporated in the domain, and are treated by the extreme barrier approach $s_{\Omega^d}$, where $s_{\Omega^d}(\mathbf{z}) = s(\mathbf{h}_0\odot \mathbf{z})$ if $\mathbf{z}\in{\Omega^d}$ and $-\infty$ otherwise. We define the neighborhood function of binary variables $\mathbf{z}$ to be,
\begin{align}
	\mathcal{N}(\mathbf{z}) = \{\mathbf{y}\in \{0,1\}^n: \norm{\mathbf{y} - \mathbf{z}}_1 \le d\}, 
	\label{equ:opt_spec_neigh}
\end{align}
where $d\in Z^+\backslash\{1\}$, since $\norm{\mathbf{y} - \mathbf{z}}_1 \ge 2$ if $\mathbf{y} \ne \mathbf{z}$ and $\mathbf{y}, \mathbf{z} \in \Omega^d$, where the shortest distance of $2$ occurs when exchanging only one element of value $1$ with another of value $0$.

We then propose the following \textit{customized direct search} (CDS) algorithm for the revised problem (\ref{equ:opt_dyn_gen_comb}). In this algorithm, we start from an exact solution (in Theorem \ref{the:lineard_exact_sol}) when the GIP model is at the extreme of the EIC model. Then at each iteration $r$ in the \texttt{poll} step, we search the local neighborhood of the current candidate $\mathbf{z}^{(r)}$, until a point with sufficient improvement in the objective value has been found or all points have been exhausted. In the termination check, if an improved point has been found, the algorithm will go back to the optional \texttt{search} step, but will decrease the required improvement if a sufficiently improved point has not been found; if no improvement has been found, the algorithm outputs the current iterate and terminates; see Algorithm \ref{alg:mv-mads_custom} for more details. The default parameter values are set to be $\zeta = 0.1$, $\delta = 0.5$ and $d = 2$.

Therefore, local convergence is directly guaranteed in the termination step, by Definition \ref{def:local_maximizer}. Global convergence could be obtained with a sophisticatedly developed \texttt{search} step and a better understanding of the landscape of the objective function, in order not to be trapped in bad local optima. However, the downside of a global method is its time complexity, thus we leave the \texttt{search} step optional. Instead, the CDS method incorporates the problem's features and circumvents the worst-case complexity by initializing with an exact solution when the GIP model reaches the extreme of the EIC model, and we postulate that the local optima near this special solution are sufficiently good. We leave the detailed discussion of the time complexity to Appendix \ref{sec:appendix_CDS_timecomplex}.
\begin{algorithm}[H]
	\caption{Customized direct search (CDS).}
	\label{alg:mv-mads_custom}
	\begin{algorithmic}[1] 
		\State{initialization: Set $0 < \zeta, \delta < 1$. Let $\mathbf{z}^{(0)}\in \Omega^d$ such that $z_j^{(0)} = 1$ if node $j\in \mathcal{A} = \{j_1, ..., j_k\}$ where $h_{i,0}c_i \le h_{j,0}c_{j},\ \forall i\notin\mathcal{A}, j\in\mathcal{A}$, and $\mathbf{c} = \{[\mathbf{I} - (1-\gamma)\mathbf{W}]^{-1} - \mathbf{I}\}\mathbf{1}$ is the Katz centrality. Set iteration $r = 0$.} 
		\State{\texttt{SEARCH} step (optional): Evaluate $s_{\Omega^d}$ on a finite subset of trial points on the mesh $M_r$ (\ref{equ:opt_spec_mesh}), until a sufficiently improved mesh point $\mathbf{z}$ is found, where $s_{\Omega^d}(\mathbf{z}) > (1 + \zeta)s_{\Omega^d}(\mathbf{z}^{(r)})$, or all points have been exhausted. If an improved point is found, then the \texttt{SEARCH} step may terminate, skip the next \texttt{POLL} step and go directly to step 4.} 
		\State{\texttt{POLL} step: Evaluate $s_{\Omega^d}$ on the set $\Omega^d\cap \mathcal{N}(\mathbf{z}^{(r)}) \subset M_r$ as in (\ref{equ:opt_spec_neigh}), until a sufficiently improved mesh point $\mathbf{z}$ is found, where $s_{\Omega^d}(\mathbf{z}) > (1 + \zeta)s_{\Omega^d}(\mathbf{z}^{(r)})$, or all points have been exhausted.} 
		\State{Termination check: If an improvement is found, set $\mathbf{z}^{(r+1)}$ as the improved solution, while decreasing $\zeta\leftarrow \delta\zeta$ if a sufficient improvement has not been found, increment $r\leftarrow r+1$, and go to step 2. Otherwise, output the solution $\mathbf{z}^{(r)}$.}  
	\end{algorithmic}
\end{algorithm}

From the current CDS method, there are two dimensions to further improve the quality of the output. We note that the current problem is equivalent to selecting a set of nodes to give value $1$ (and others $0$). Accordingly, there are two known methods of global convergence: (i) \textit{brute force}, where all node sets of size $k$ are evaluated in order to choose an optimal one, and (ii) \textit{random sampling}, where randomly chosen node sets are evaluated, and this method has global convergence asymptotically if it samples densely enough. The two dimensions of improvement are motivated by these two methods. On the one hand, we can enlarge the distance in defining the neighborhood, which necessarily searches more points in the domain. Further, if the neighborhood is as large as the whole domain, it reduces to the brute-force method. On the other hand, we can restart the searching process, i.e.,~steps 2, 3, and 4 in Algorithm \ref{alg:mv-mads_custom}, from other unexplored points randomly, which works in the same logic as the \texttt{search} step. This strategy will give global convergence asymptotically, similar to the random sampling method.

\section{\label{sec:experiments_results} Numerical experiments}
In this section, we experimentally illustrate the rich behavior of the GIP model, and evaluate the performance of the CDS method for the IM problem in both small and large, both synthetic and real networks. Throughout the section, $l_{j,0} = h_{j,0} = 1,\ \forall v_j\in V$, $\gamma = 0$, and we apply exclusively the threshold-type bounds (\ref{equ:dyn_gen_thres-bBt}) with condition (\ref{equ:dyn_gen_thres-uni}), thus the lower bounds vary according to the lower bound threshold $\theta_l$ and the upper bounds also change with the upper bound threshold $\theta_h$.

\subsection{\label{sec:experiments_model}Information propagation}
We start from the general features of the GIP model. In accordance with Sec.~\ref{sec:general_diffusion_model}, we show that the GIP model can have both the threshold effect and the locally linear effect via tuning lower and upper bounds, respectively. Such effects cannot happen simultaneously in either the (E)IC or the (E)LT model, but may coexist in real systems. 

Specifically, we consider simple networks generated from the two-block planted $SBM(0.9, 0.1)$, where an edge, \imp{being bidirectional}, is placed between the nodes in the same community with probability $p_{in} = 0.9$ and in the different communities with $p_{out} = 0.1$. The networks have size $n=50$ and $n_c = 2$ communities, where we label the nodes in communities one and two as $0$ to $24$ and $25$ to $49$, respectively; see Fig.~\ref{fig:Diff_SBM_net} for one realization. 
\begin{figure}[ht]
	\centering
	\includegraphics[width=.3\textwidth]{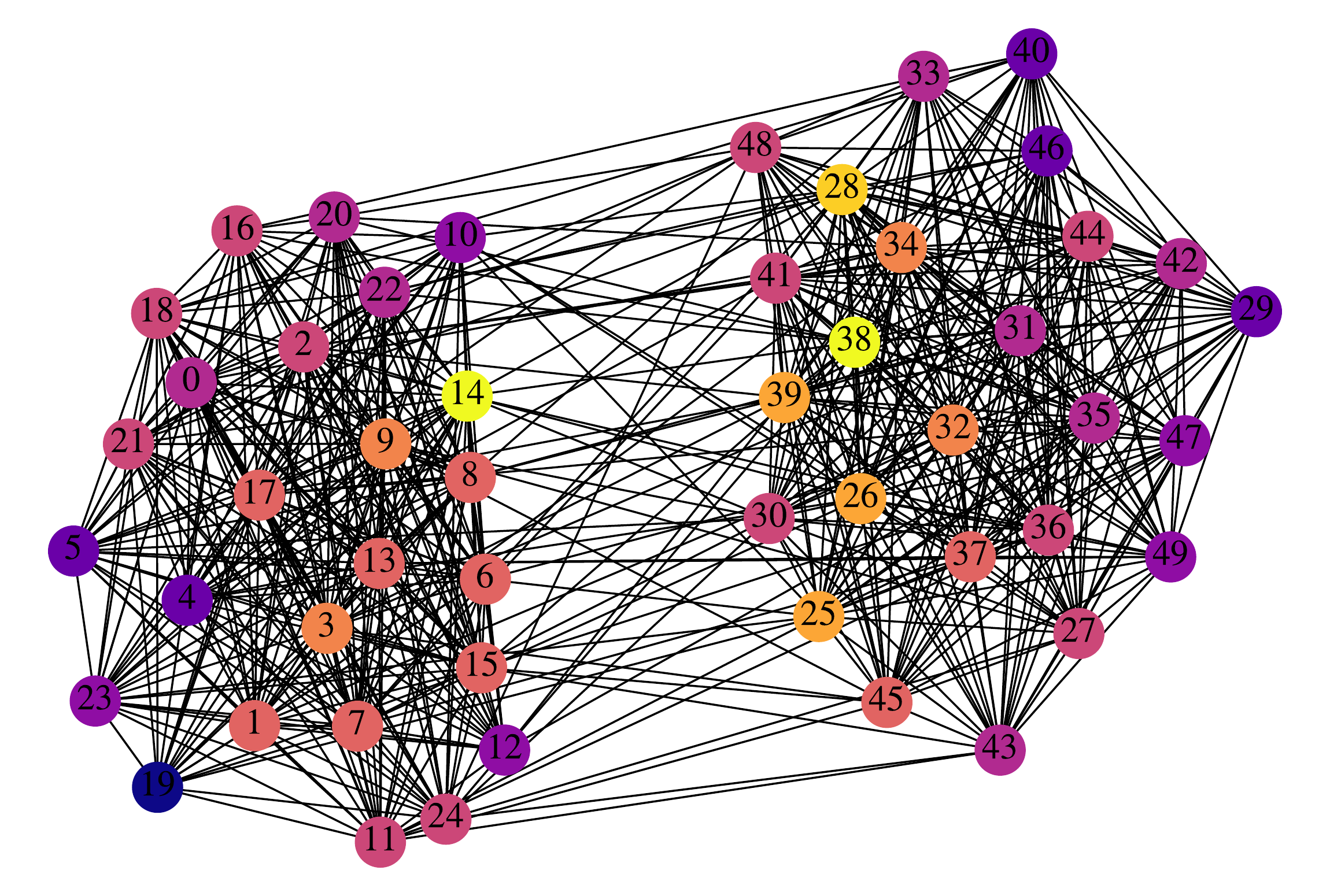}
	\caption{One realization of the two-block $SBM(0.9, 0.1)$, where the lighter the color of a node is, the higher the degree of a node is.}
	\label{fig:Diff_SBM_net}
\end{figure}
We choose these values to enlarge difference between different types of node sets, for visualising purposes, \imp{but noting that the results are expected to hold for any probabilities s.t.~$p_{in} \ne p_{out}$ and on networks of any size $n$ (cf.~Claims \ref{cla:linear_difference} and \ref{cla:extreme2_devi-tree})}. We assign a uniform weight $\alpha=0.1$, to account for \imp{moderate level of trust among agents}. Therefore, $\theta_l = 1$ corresponds to the critical lower bounds for the linear-dynamics extreme, where any $\theta_l > 1$ cannot always result in linear dynamics. The difference in the propagation behavior will be quantified by the \imp{time-dependent influence} on all nodes, 
\begin{align}
	s(t) = \sum_{j}\sum_{t'=0}^{t}(1-\gamma)^{t'}x_j(t'),
	\label{equ:Diff_aggstate_tot}
\end{align}
where $\mathbf{x}(t') = (x_j(t'))$ is the state vector at time step $t'$ with the updating function (\ref{equ:dyn_gen_revis}) and a given initial state vector $\mathbf{x}(0) = (x_j(0))$. Note that $\lim_{t\to\infty}s(t) - \sum_{j}x_j(0)$ is the objective in the IM problem. 

\begin{figure*}[ht]
	\centering
	\begin{tabular}{cc}
		\includegraphics[width=.35\textwidth]{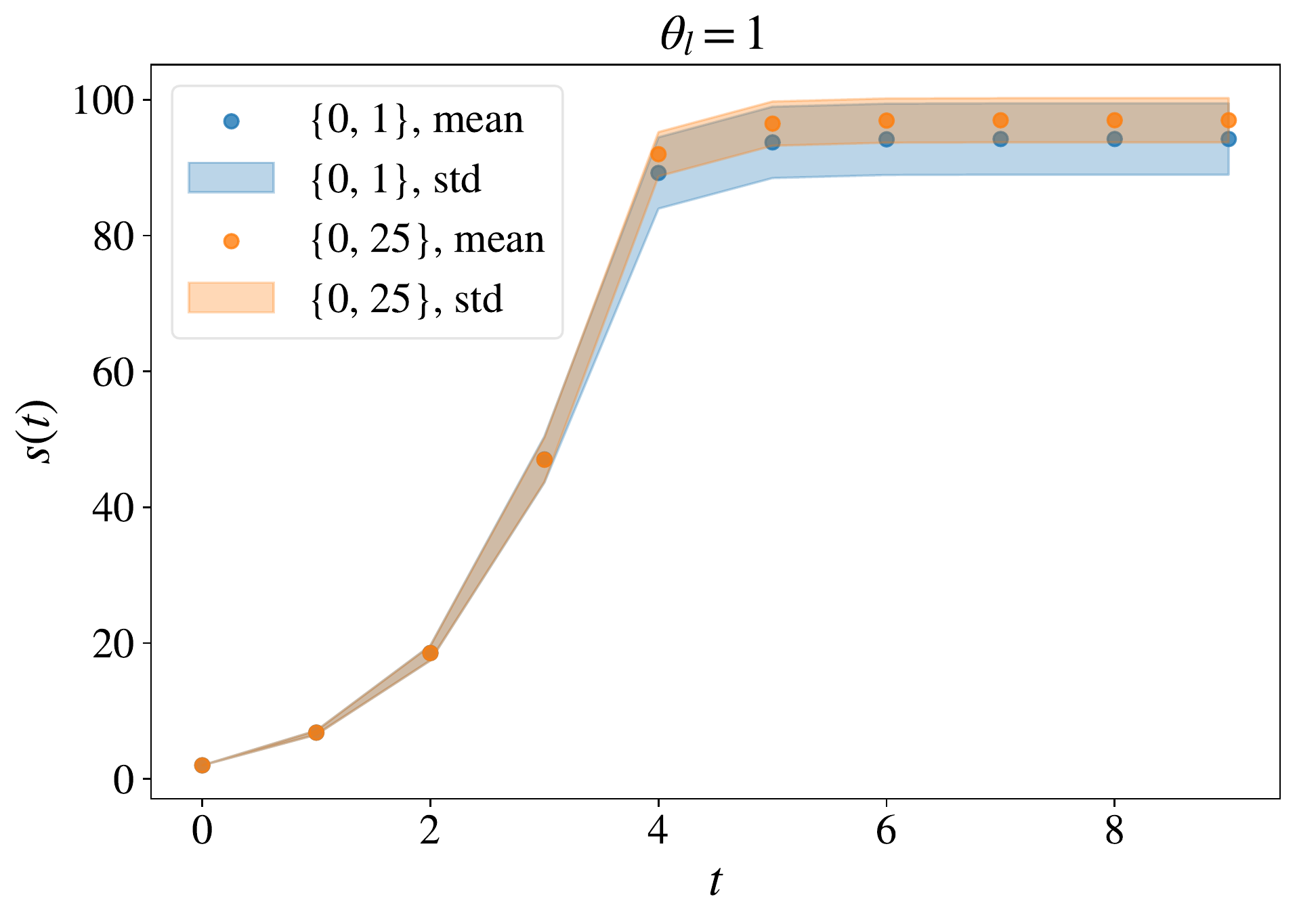} & \includegraphics[width=.35\textwidth]{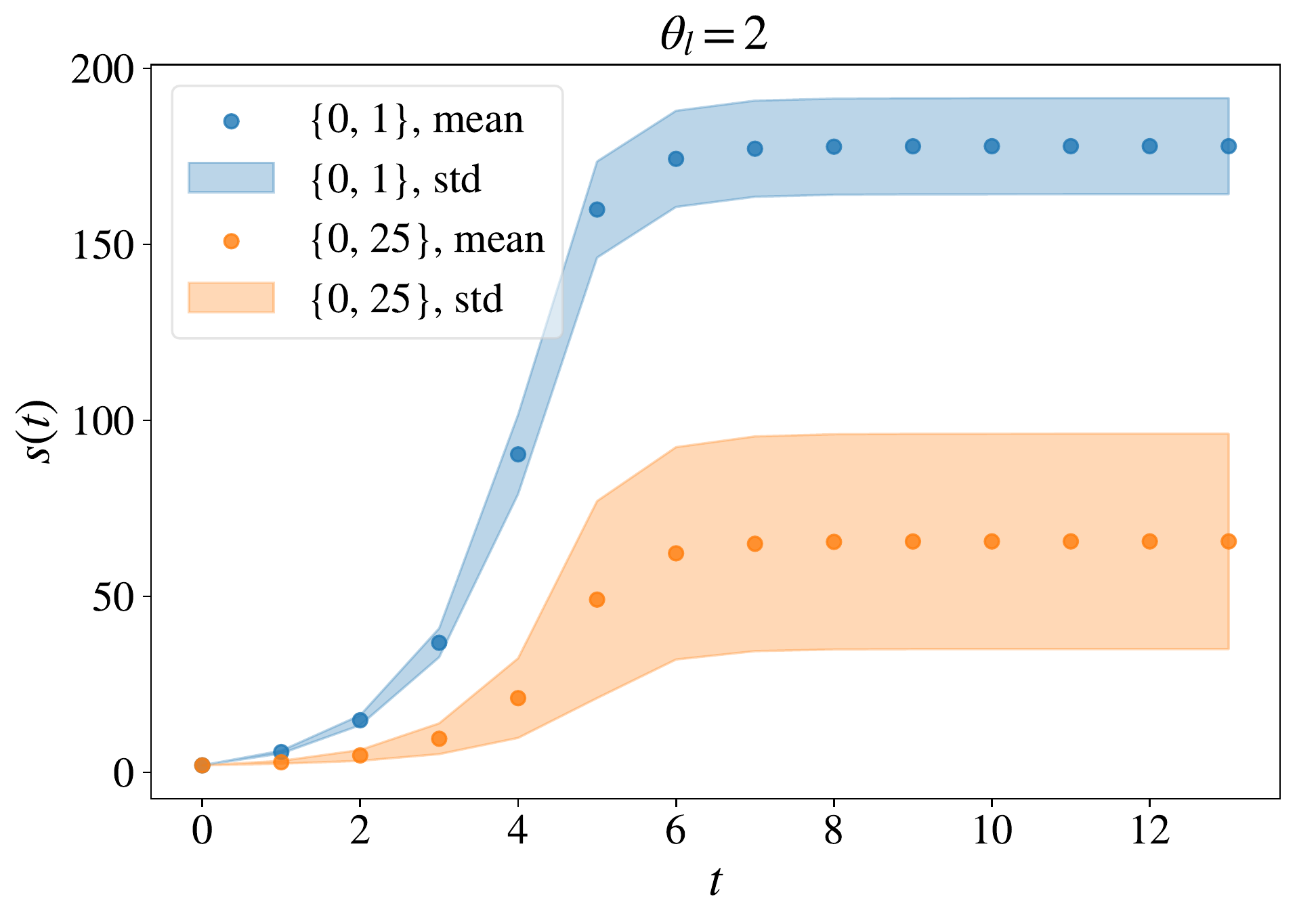}
	\end{tabular}
	\caption{The time-dependent influence, from two different initially activated node sets, with $\theta_l = 1$ (left, the critical value for the extreme of the EIC model) and $\theta_h = 2$ (right, a larger value) while $\theta_h$ being large (here $9000$), on $1000$ samples of $SBM(0.9,0.1)$.}
	\label{fig:Diff_SBM_abs_ns1000}
\end{figure*}
We first show the threshold effect imposed on top of the EIC model (i.e.,~linear dynamics), by tuning the lower bounds of the GIP model to gradually deviate it from the extreme of the EIC model, as in Sec.~\ref{sec:diff_props-gen}. Specifically, we consider the following two initially activated node sets as in Claim \ref{cla:linear_difference}: (i) $\{0,1\}$ from the same community and (ii) $\{0,25\}$ from the different communities. The results numerically verify such effect, since the \imp{time-dependent influence} from the two node sets are similar when $\theta_l = 1$, while set (i) triggers a propagation with higher influence as $\theta_l$ slightly increases; see Fig.~\ref{fig:Diff_SBM_abs_ns1000}. 
We note that the performance of set (ii) has larger variance, and this is because it largely depends on the intercommunity edges whose existence has much smaller probability. Here, the results from set (ii) are concentrated at the values slightly above the mean [not as high as set (i)], while they also contain values substantially lower than the mean.

\begin{figure*}[ht]
	\centering
	\begin{tabular}{cc}
		\includegraphics[width=.35\textwidth]{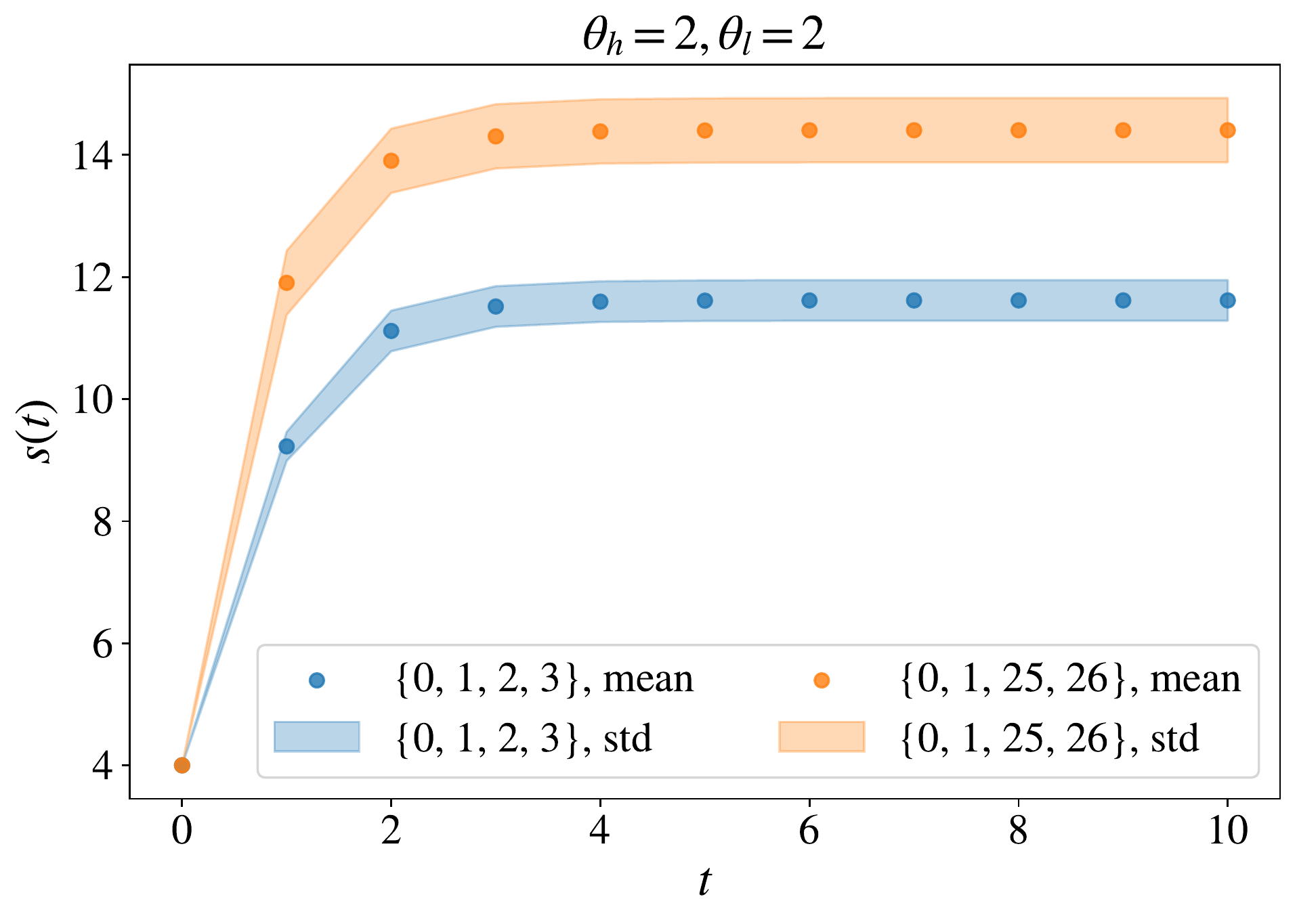} & \includegraphics[width=.35\textwidth]{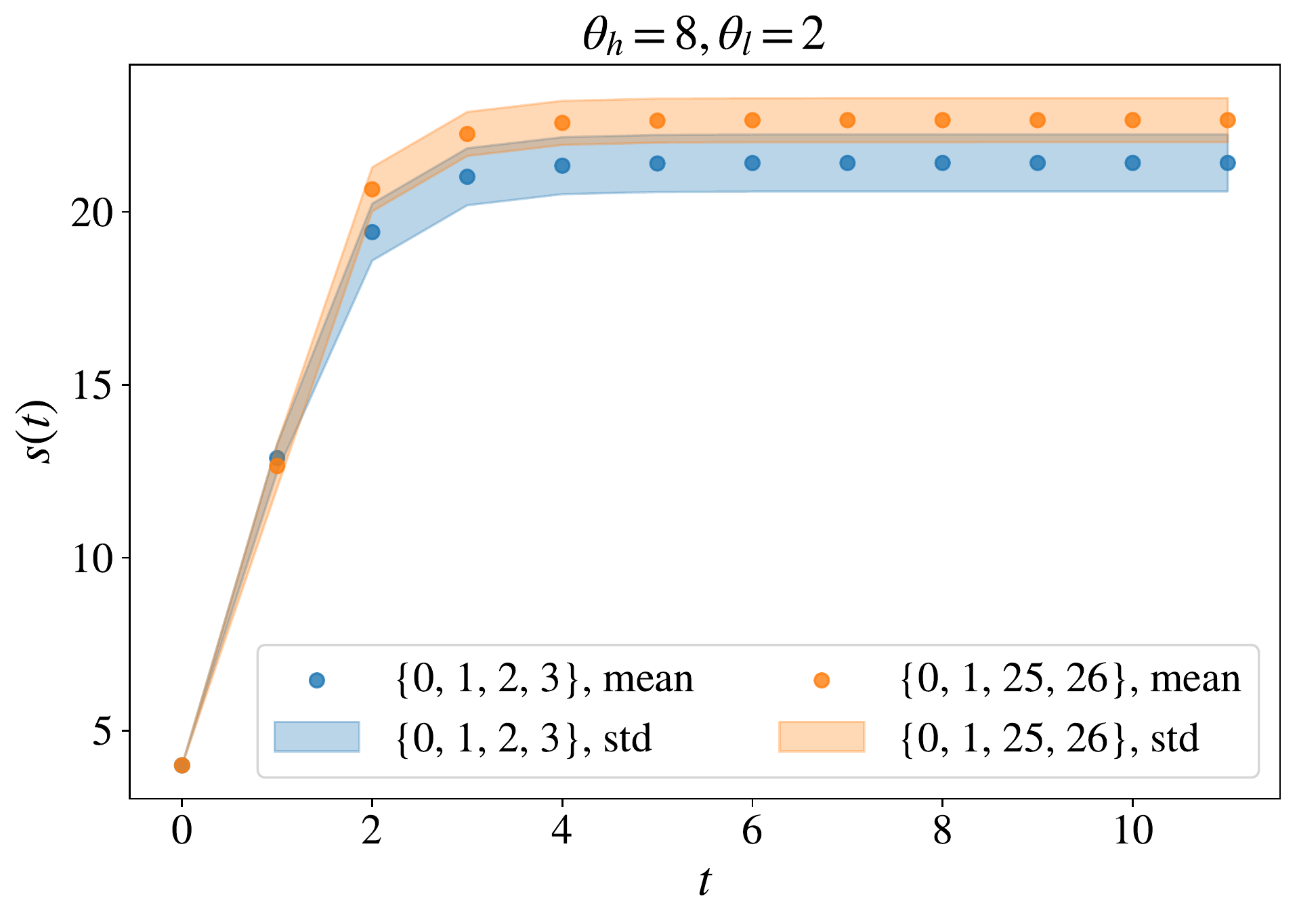}  
	\end{tabular}
	\caption{The time-dependent influence, from two different initially activated node sets, with $\theta_l = 2$ while $\theta_h = 2$ (left, the extreme of the ELT model) and $\theta_h = 8$ (right, a larger value), on $1000$ samples of $SBM(0.9,0.1)$.}
	\label{fig:Diff_SBM_aBs_ns1000}
\end{figure*}
We then illustrate the locally linear effect added on top of the ELT model, by tuning the upper bounds of the GIP model to differentiate it from the extreme of the ELT model. Here, we fix $\theta_l = 2$, the smallest integer for the threshold effect to take place, while increasing $\theta_h$ from $2$ to a higher value, as in Sec.~\ref{sec:diff_props-gen}. We then consider the following two larger initially activated node sets: (i) $\{0,1,2,3\}$ from the same community; (ii) $\{0,1,25,26\}$ evenly distributed in the two communities. The numerical results justify such effect, because the \imp{time-dependent influence} from set (i) has larger increase than the other when the upper bound threshold $\theta_h$ increases; see Fig.~\ref{fig:Diff_SBM_aBs_ns1000}. 
We also observe that the propagation triggered by set (ii) has consistently higher influence than the other.

Finally, we integrate the two aspects and provide a whole picture of these general features of the GIP model, by changing the upper and lower bounds simultaneously. Here, we consider both node-set pairs, $\{0,1\}$ versus $\{0,25\}$ and $\{0,1,2,3\}$ versus $\{0,1,25,26\}$, and quantify their relative behavior by the \textit{ratio} $\delta$ of the overall influence \eqref{equ:dyn_gen_influ-j} following the GIP model. We observe consistent patterns of the previously analyzed features: (i) the node sets in the same community have consistently higher influence as $\theta_l$ exceeds the critical value for the other set, $1$ for $\{0,25\}$ and $2$ for $\{0,1,25,26\}$ (which are equal to the numbers of nodes distributed in each community); (ii) the overall influences from the node sets in each pair are increasingly similar as $\theta_h$ increases in general; see Fig.~\ref{fig:Diff_SBM_heatmap-as_ns1000}.
We also notice that there is a regime where increasing the upper bounds will enlarge the (relative) difference between the node sets in each pair, which emphasizes the nonlinearity in the model. 
\begin{figure*}[ht]
	\centering
	\begin{tabular}{cc}
		\includegraphics[width=.48\textwidth]{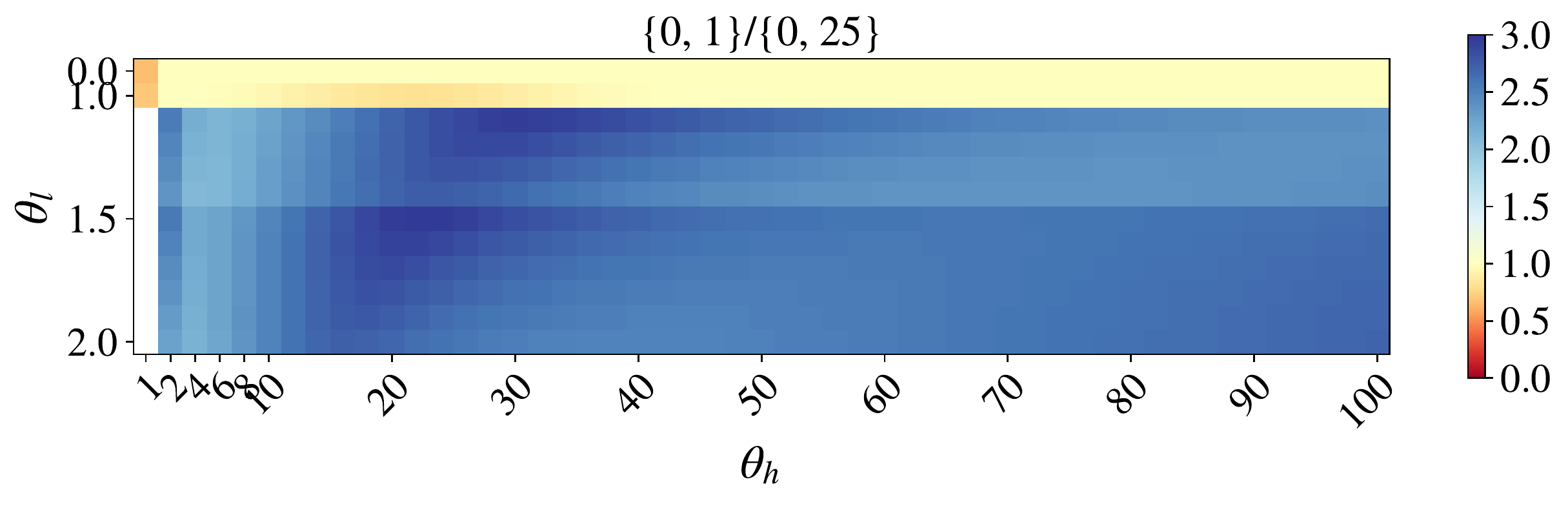} & \includegraphics[width=.48\textwidth]{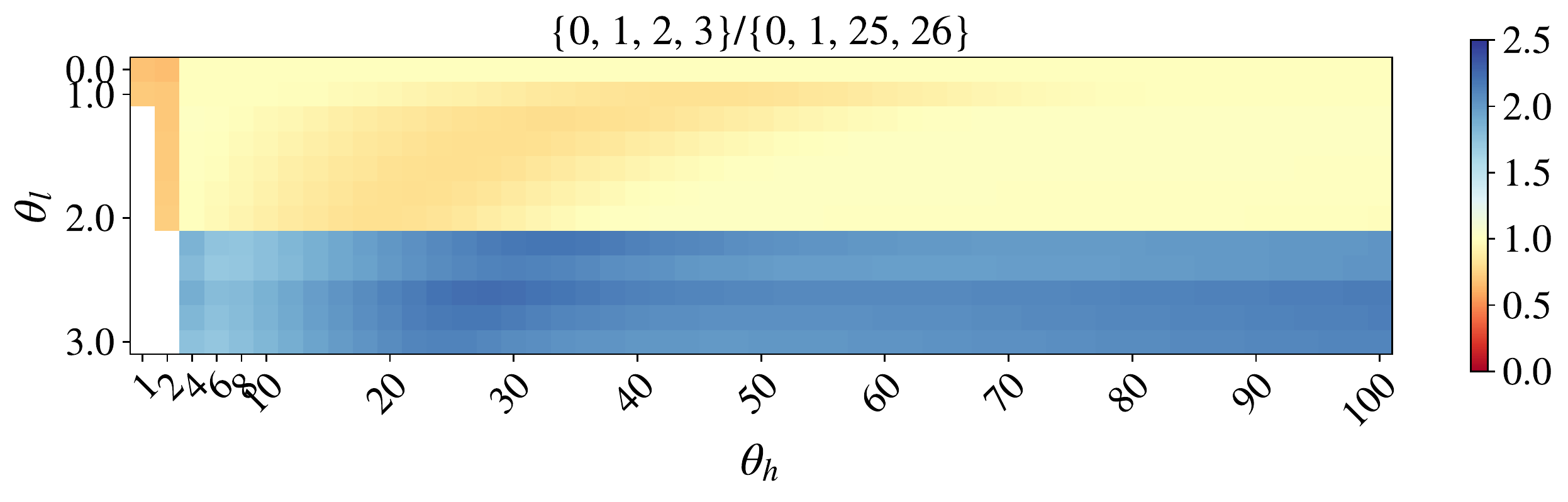}
	\end{tabular}
	\caption{The ratios $\delta$ of the overall influence from two pairs of initially activated node sets, $\{0,1\}$ to $\{0,25\}$ (left) and $\{0,1,2,3\}$ to $\{0,1,25,26\}$ (right), with changing upper ($x$-axis) and lower bound ($y$-axis) thresholds, on $1000$ samples of $SBM(0.9, 0.1)$.}
	\label{fig:Diff_SBM_heatmap-as_ns1000}
\end{figure*}

\subsection{Influence maximization}
We now proceed to the proposed CDS method for the IM problem, with the GIP model governing the information propagation process and Eq.~\eqref{equ:dyn_gen_influ-j} as the function for individual influence. In the following, we first test the accuracy of the CDS method on networks of a relatively small scale. We then examine its performance on relatively large networks, both real and synthetic, by comparison with other state-of-the-art approaches.

\subsubsection{\label{sec:experiments_accuracy} The CDS method}
From Sec.~\ref{sec:IM_customized_direct_search}, we know that the CDS method is a local algorithm, thus it is important to explore how good the output is with respect to the global optima. However, algorithms with global convergence, such as the brute-force method, require $O(n^k)$ evaluations of the objective, which prohibits their application to large networks. Hence, we exclusively consider networks of a relatively small scale in this section. 

In order to measure the \textit{goodness} of the output from the CDS method, we consider the following two measures: accuracy and rank. The \textit{accuracy} is defined as the relative value to a global optimum, 
\begin{align*}
	\tau(s; s^*) = s / s^*,
\end{align*}
where $s$ is an output from the proposed algorithm, and $s^*$ is a global optimum. $\tau(\cdot)\in [0,1]$ and a higher accuracy implies a better solution. Since the current problem is purely combinatorial in terms of the initially activated node set, we also consider the output's \textit{rank},
\begin{align*}
    \phi&(\mathcal{A}_0; V, k) = \\
    &\frac{\abs{\{\mathcal{A}\subset V: \abs{\mathcal{A}} = k, s(\mathbf{h}_0\odot \mathbf{z}_{\mathcal{A}})> s(\mathbf{h}_0\odot \mathbf{z}_{\mathcal{A}_0}) \}} + 1}{\abs{\{\mathcal{A}\subset V: \abs{\mathcal{A}} = k \}}}, 
\end{align*}
where $G(V, E)$ is the underlying network, $k$ is the budget size, $\mathcal{A}_0$ is the set of initially activated nodes corresponding to the output, $\mathbf{z}_{\mathcal{A}}$ is the binary vector whose $j$th element is $1$ if and only if $v_j\in \mathcal{A}$, and $s(\mathbf{h}_0\odot \mathbf{z})$ is the objective function of the revised problem (\ref{equ:opt_dyn_gen_comb}). $\phi(\cdot)\in (0,1]$ and a lower rank implies a better solution. Hence, in order to obtain the parameters in the measures, we select the brute force as the reference global algorithm. 

Specifically, we consider the two cases: (i) differentiating $\theta_l$ and $\theta_h$ while maintaining $k$, and (ii) varying $k$ while fixing $\theta_l$ and $\theta_h$. In case (ii), we also evaluate the performance of the initial point of the CDS method, i.e.,~selecting nodes by their Katz centrality (since $h_{j,0} = 1,\ \forall v_j\in V$). As a representative example, we show results exclusively from two-block SBMs, and defer the results from other types of networks to Appendix \ref{app:CDS-performance}. The SBM considered here has the same size and weights as in Sec.~\ref{sec:experiments_model} and \imp{the edges are still bidirectional}, but of different probabilities in the two communities, $p_1 = 0.3$ in one and $p_2 = 0.12$ in the other, in order to distinguish nodes in different communities. The connecting probability between the communities is set to be a smaller value $p_{12} = 0.01$ here.

When $k=4$, the CDS method can find a solution either globally optimal or fairly close to optimal with all different combinations of the upper and lower bound thresholds; see Fig.~\ref{fig:IM_SBM_4_acc-com}. There are only five cases when the CDS method cannot find a globally optimal solution. 
However, in such worst-case scenarios, the solutions still have accuracy over $0.95$ and rank less than $0.001\%$ in overall $230,300$ possibly initially activated sets, i.e.,~the CDS method can still output a top $2$ set. Since different combinations of the upper and lower bounds correspond to various properties of the underlying propagation process as discussed in Sec.~\ref{sec:general_diffusion_model}, these demonstrate that the proposed method can capture the general properties of the IM problem.  
\begin{figure*}
	\centering
	\begin{tabular}{cc}
		\includegraphics[width=.48\textwidth]{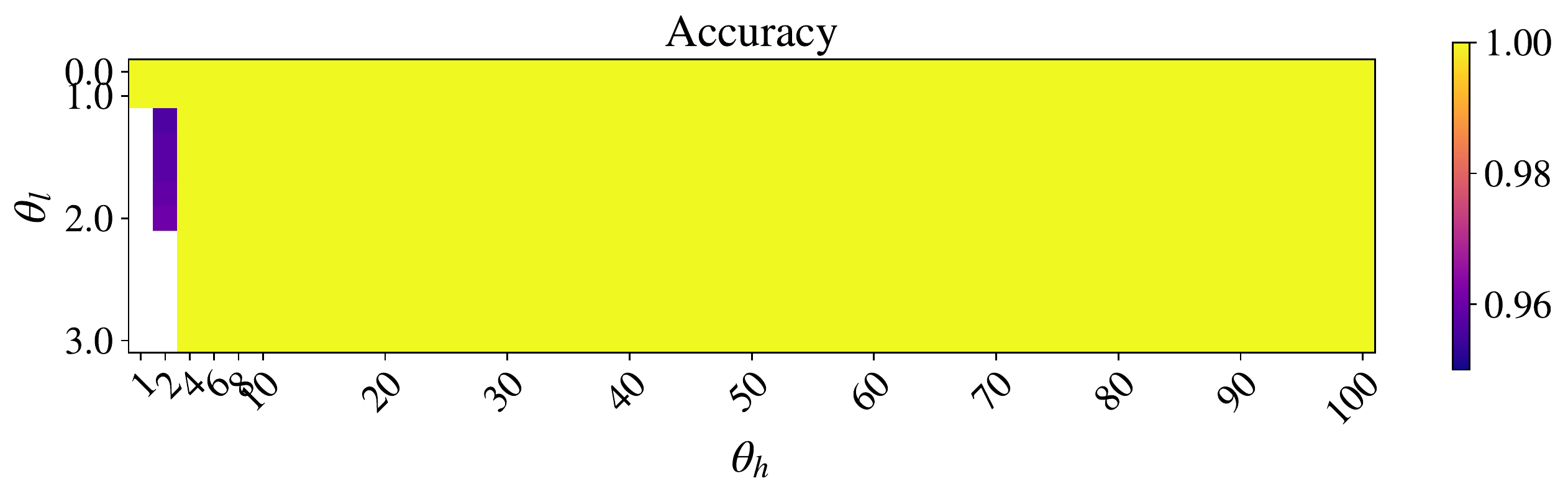} & \includegraphics[width=.48\textwidth]{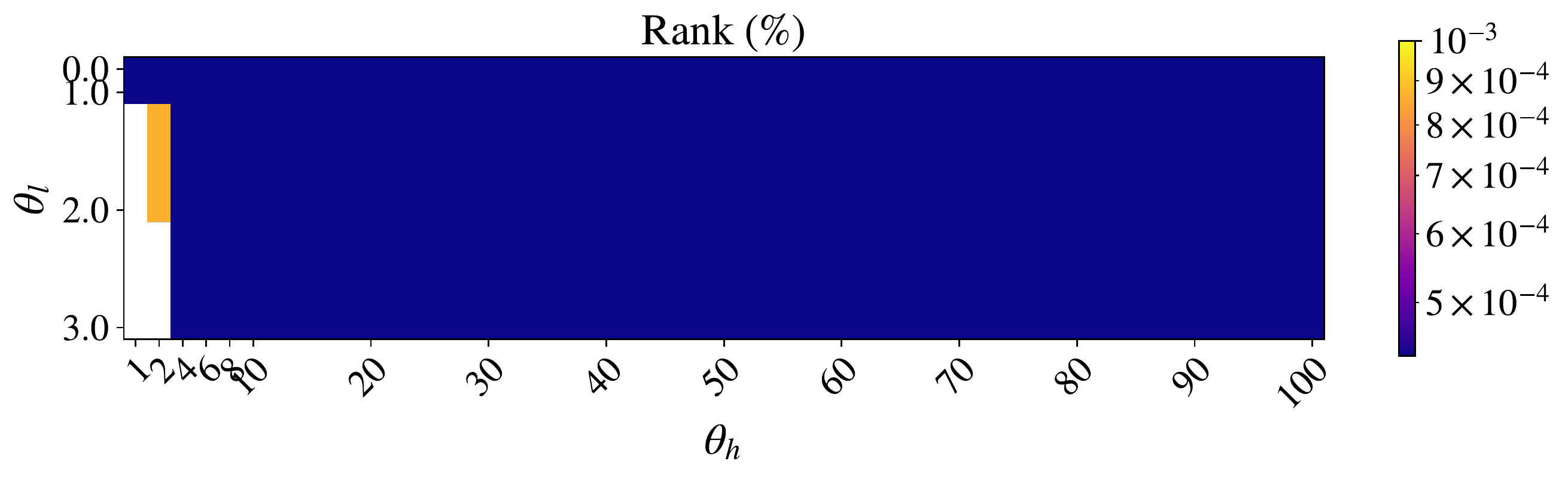}
	\end{tabular}
	\caption{Performance of the CDS method on the SBM in terms of the output's accuracy (left) and rank (right), subject to changing upper ($x$-axis) and lower ($y$-axis) bound thresholds of the GIP model, $\theta_h$ and $\theta_l$, respectively, when $k=4$.}
	\label{fig:IM_SBM_4_acc-com}
\end{figure*}

We further explore the performance of the CDS method when increasing the budget size $k$ and $\theta_l = 2, \theta_h = 2$, one of the pairs with worst-case performance when $k=4$. We observe that the outputs from the CDS method generally have high accuracy (greater than $0.9$), and consistently low rank; see Fig.~\ref{fig:IM_SBM_au3_Ks_acc-time}. There is a drop in accuracy when $k$ becomes larger. Apart from the drop in performance of the initial point, it is also partially because the fixed neighborhood size, $2$ here, becomes increasingly restrictive. The former is because the current dynamics are far from \imp{being linear (i.e.,~the EIC model)}. The latter is inherited in the CDS method being local, since one needs to define the radius of local neighborhood. As discussed in Sec.~\ref{sec:IM_customized_direct_search}, there are two dimensions to further improve the algorithm: (i) to enlarge the radius of neighborhood as $k$ increases, e.g.,~an \textit{adaptive} neighborhood and (ii) to \textit{restart} the search from some other points, preferably far from the original initial point. See Appendix \ref{app:CDS-performance} for the strategy we propose to improve the performance. 
\begin{figure*}
	\centering
	\begin{tabular}{cc}
		\includegraphics[width=.35\textwidth]{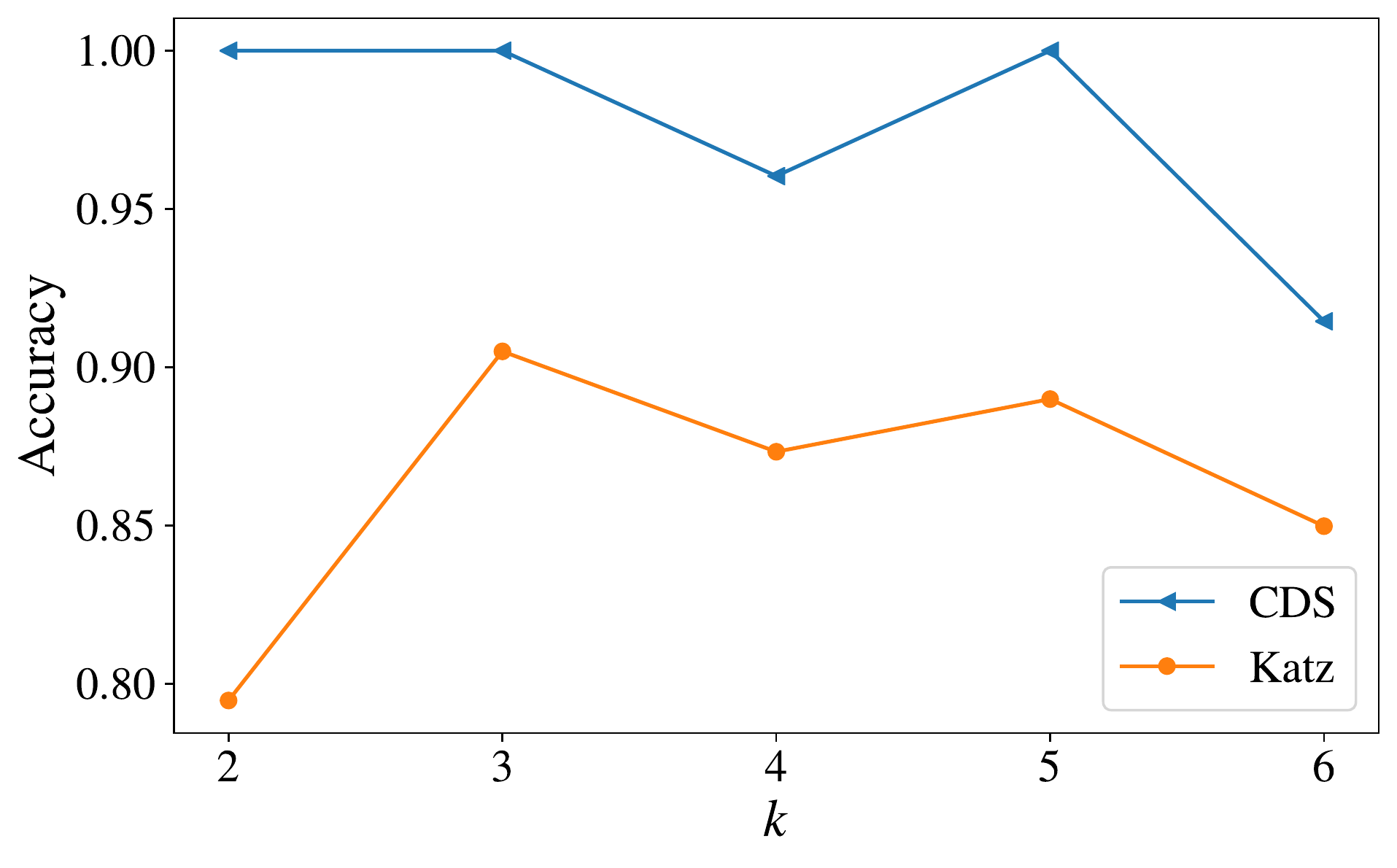} & 
		\includegraphics[width=.35\textwidth]{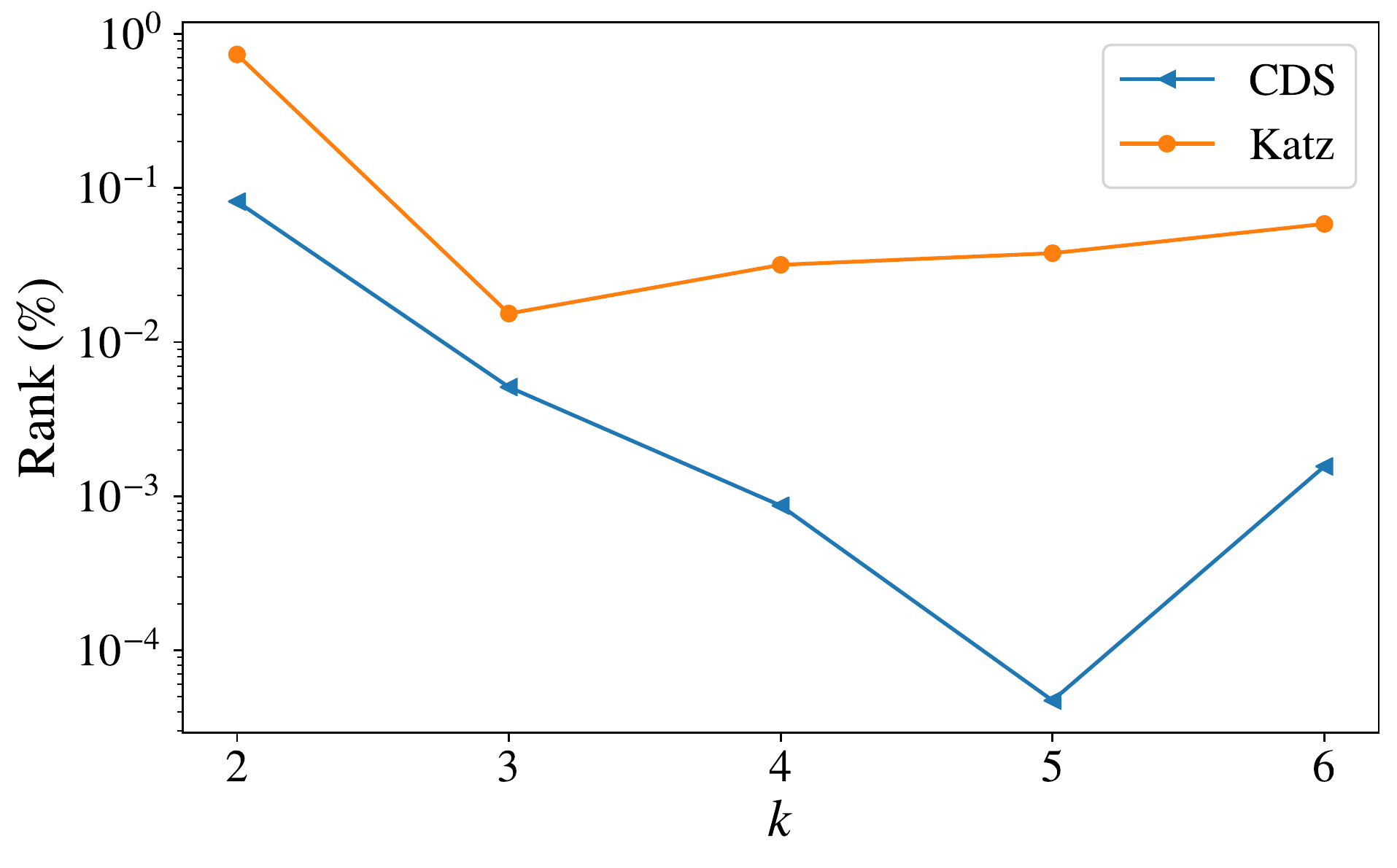}
	\end{tabular}
	\caption{Performance of the CDS method on the SBM in terms of the output's accuracy (left) and rank (right), subject to changing budget size $k$ ($x$-axis) when $\theta_l = 2, \theta_h = 2$.}
	\label{fig:IM_SBM_au3_Ks_acc-time}
\end{figure*}

\subsubsection{\label{sec:experiments_comparison} Comparison between methods}
In the following, we examine the performance of the CDS method on both synthetic and real networks on a relatively large scale, by comparing them with the state-of-the-art methods. Hereafter, we denote the set of initially activated nodes, $\mathcal{A}_0$, as the \textit{seed set}. We note that the vast majority of greedy algorithms following the work of Kempe \textit{et al.} \cite{Kempe_influence_2003} are not applicable here, because every node sets that are not large enough \footnote{\footnotesize{It corresponds to the node sets of sizes less than $\theta_l$ in networks with uniform weights, while in weighted networks, we also need to incorporate the exact weights around each node compared to the mean weight $\alpha$.}} will return an overall influence $0$ and then these algorithms lack an appropriate approach to select the first few nodes, which has significant consequences in their following steps. We instead compare the proposed method with the following ones.
\begin{enumerate}[label=(\roman*)]
	\item \textbf{Random sampling} (``Random"). Randomly selecting $k$ nodes of the network, and returning them as the seed set. We repeat the process $n_s$ times, and output the one with the highest objective value. 
	\item \textbf{Degree-centrality method} (``Degree"). Centrality is an important measure which quantifies the significance of nodes in networks \cite{Banerjee_IMsurvey_2020}. The degree-centrality method is to select the $k$ nodes of the highest degree in the network as the seed set. 
	\item \textbf{Katz-centrality method} (``Katz"). Finding the $k$ nodes of the highest Katz centrality in the network as the seed set. 
\end{enumerate}  
Since method (i) has random components, we will repeat the method $n_r$ times and analyze their averaged performance. In this section, we compare the performance of different methods directly through the overall influence $s$, i.e.,~the objective value in the MINLP \eqref{equ:opt_dyn_gen} (or equivalently \eqref{equ:opt_dyn_gen_comb}).

The networks under consideration are composed of a large two-block SBM, and a real collaboration network. On the one hand, SBMs are considered here because community structure is a common feature in many real networks, and also to extend the previous analysis in Sec.~\ref{sec:experiments_accuracy} to a larger scale. Collaboration networks, on the other hand, are extensively used in IM experiments, because researchers believe that such networks capture key features of social networks \cite{Newman_collabor_2001}. The specific one we select also has ground-truth communities as metadata. In the following experiments, we choose $n_s = 100$ and $n_r = 10$.

\paragraph{Stochastic block model (SBM).} Consider a relatively large network of size $n=1000$, generated by a two-block SBM with different probabilities in the two communities, $p_1 = 0.015$ and $p_2 = 0.006$, and the connecting probability between the two being $p_{12} = 0.0005$. These values are chosen to maintain roughly the same mean degree as the SBM in Sec.~\ref{sec:experiments_accuracy}, and we assign the same uniform weight $\alpha = 0.1$. \imp{We still consider each edge being bidirectional to emphasize the feedback among the nodes}. We observe that the CDS method outperforms all the reference algorithms in all possible budget sizes; see Fig.~\ref{fig:SBM_Ks}. Further, the performance of the degree centrality is similar to the Katz centrality, and we then prove that it is theoretically expected in Appendix \ref{app:SBM}.
\begin{figure*}
	\centering
	\begin{tabular}{cc}
		\includegraphics[width=.35\textwidth]{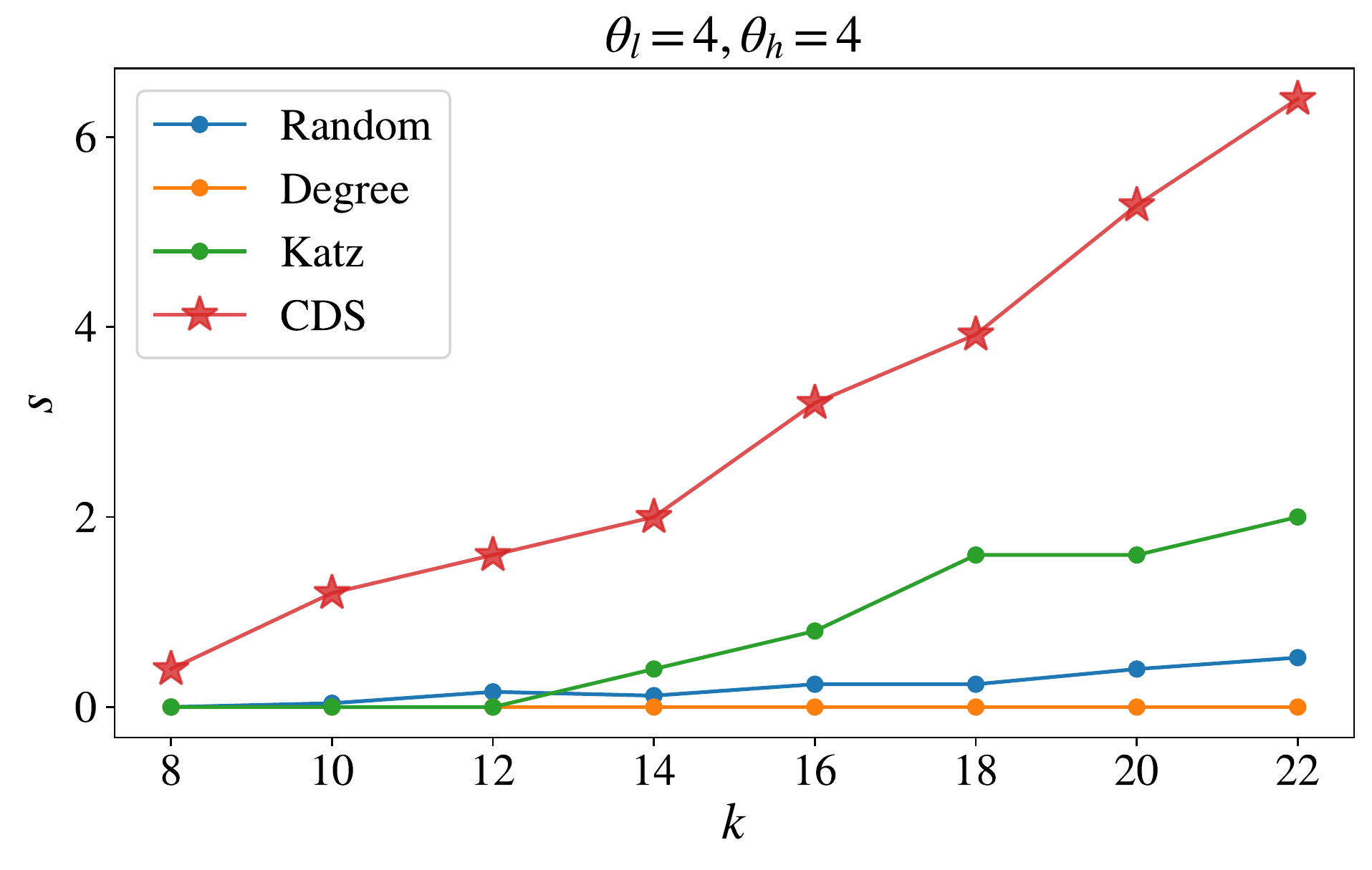} & 
		\includegraphics[width=.35\textwidth]{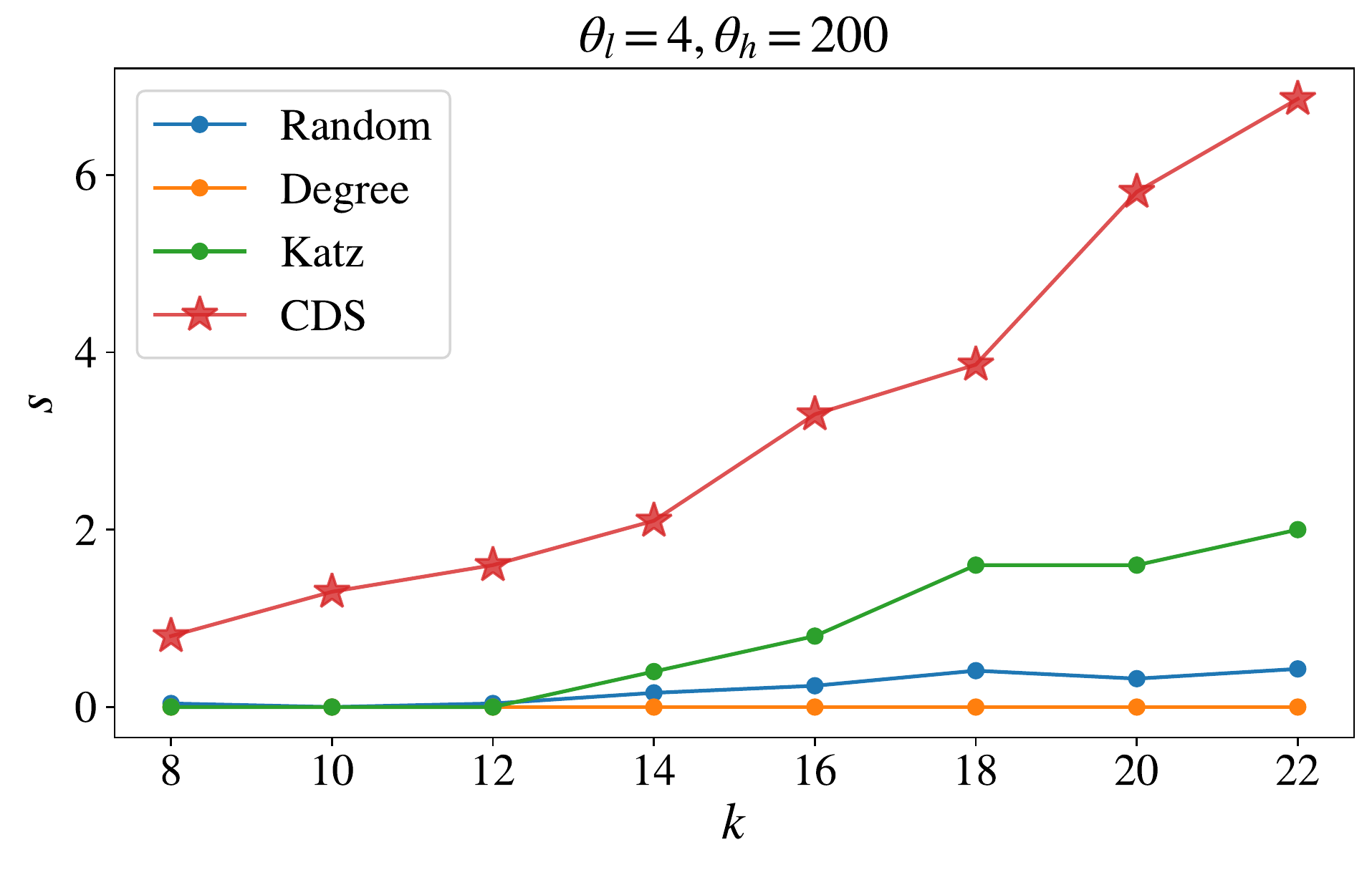}
	\end{tabular}
	\caption{Overall influence $s$ from different node selection algorithms applied to the SBM, subject to changing budget size ($x$-axis), when the GIP model has $\theta_l = \theta_h = 4$ (left) and $\theta_l = 4$, $\theta_h = 200$ (right).}
	\label{fig:SBM_Ks}
\end{figure*}

\paragraph{Collaboration network.} The collaboration network is constructed by a comprehensive list of research papers in computer science provided by DBLP computer science bibliography, where there is an edge from one author to another if they publish at least one paper together \cite{Yang_dblpdata_2018}, \imp{thus each edge is bidirectional} \footnote{\footnotesize Note that even through higher-order interactions are likely to occur in collaboration networks, the simple version (as what we considered here) still remains a classic example in the IM problem.}. There are intrinsic communities defined by the publication venue, e.g.,~journals or conferences. Here, we randomly select two such venues (no.~6035 and 6335) whose sizes are over $500$ and the related authors are relatively densely connected so that the reinforcement within groups is more likely to happen. The resulting network is unweighted and connected, containing $n=1016$ nodes and $|E| = 3469$ edges. Here, we maintain a uniform weight $\alpha=0.1$. Since the collaboration network contains several nodes of much higher degrees than the others, the performance of the linear dynamics is largely dominated by these nodes. Accordingly, we observe that both degree-centrality and Katz-centrality methods perform competitively to the CDS method when $\theta_h$ is relatively far from $\theta_l$. However, the CDS method still outperforms others, and the distance is relatively larger as $\theta_h$ becomes smaller; see Fig.~\ref{fig:colab_Ks}. 
\begin{figure*}[ht]
	\centering
	\begin{tabular}{cc}
		\includegraphics[width=.35\textwidth]{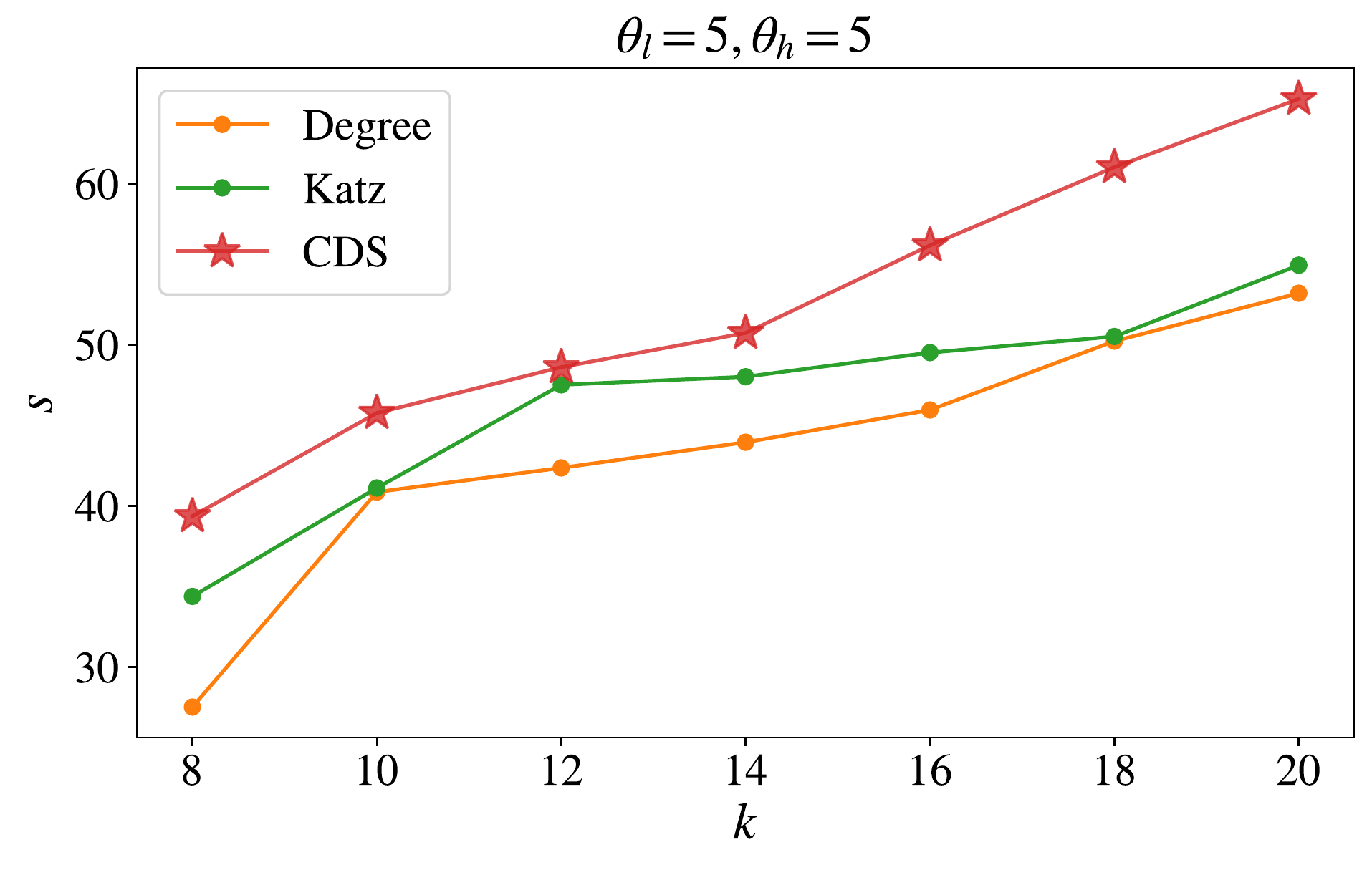} & 
		\includegraphics[width=.35\textwidth]{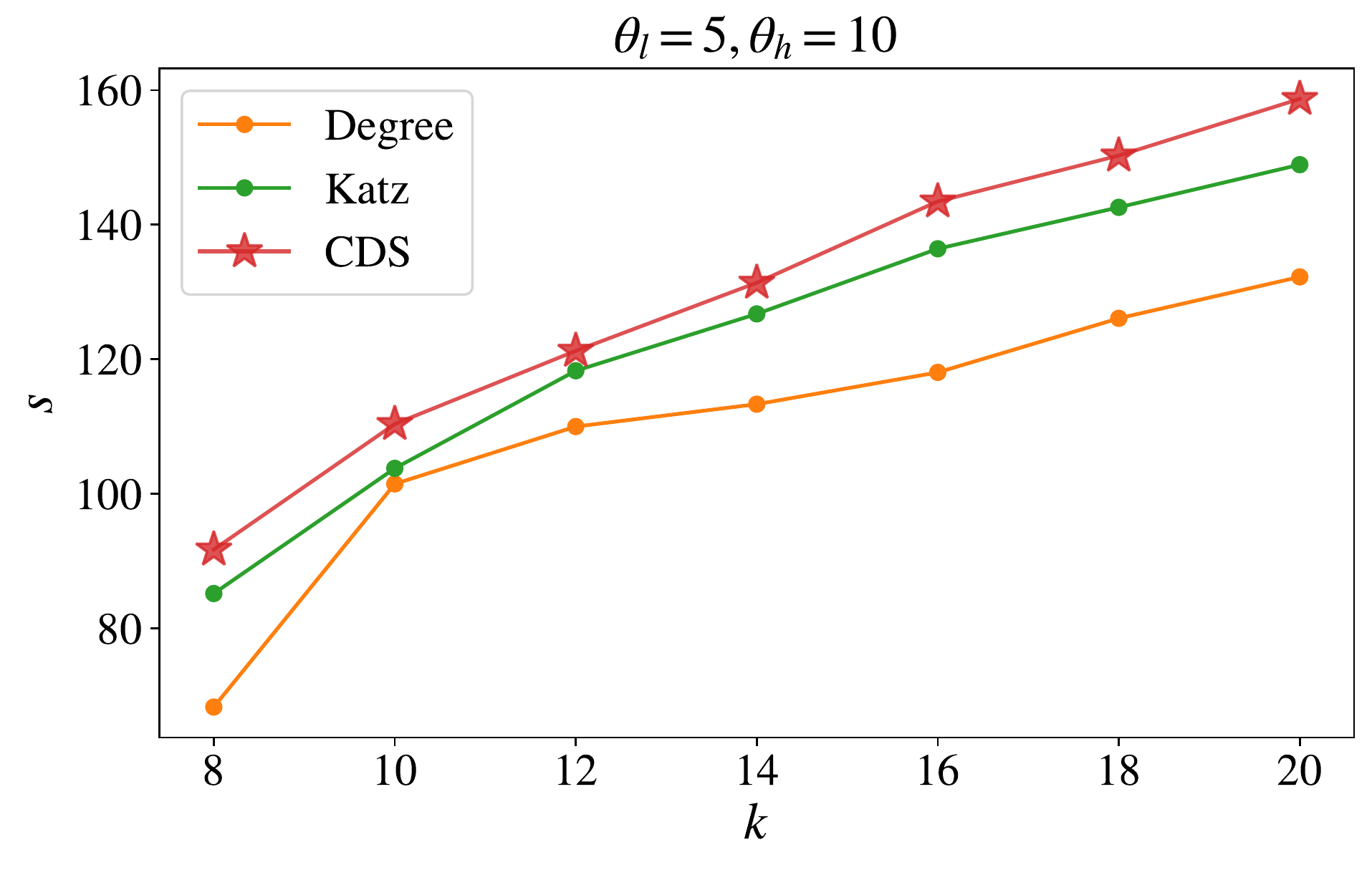}
	\end{tabular}
	\caption{Overall influence $s$ from different node selection algorithms applied to the collaboration network, subject to changing budget size ($x$-axis), when the GIP model has $\theta_l = \theta_h = 5$ (left) and $\theta_l = 5$, $\theta_h = 10$ (right), where the results from random sampling is ignored because they are always close to $0$.}
	\label{fig:colab_Ks}
\end{figure*}

\section{\label{sec:conclusion} Conclusions}
To understand how information propagates through social networks has many practical implications. Among the vast amount of work in this field, the IC model and the LT model are among the most popular choices of models. However, their characteristics, such as binary state variables and no feedback mechanism, while simplifying their analysis, neglect important mechanisms, e.g.,, a higher influence from people of higher activity and reinforcement within close-contact groups. Therefore, we extend both classic models to address these issues, and further propose the a general class of information propagation model, the GIP model, which unifies the mechanisms underlying the classic models. More importantly, the GIP model has features that each single model does not possess but may occur in real systems, as discussed in Sec.~\ref{sec:general_diffusion_model} (see also Appendix \ref{app:diff_coexistence}). We leave the investigation on real social networks, \imp{and also the incorporation of specific content of the information and potential interference or overlaps in the information} to future work. 

The general features of the GIP model necessarily lead to the IM problem with more general properties, such as the objective not being submodular. However, breaking the boundary of such restrictive properties is necessary for the IM problem to embrace a wider family of information propagation models in practice \cite{Gursoy_IMdet_2018,Li_IMsurvey_2018}. Therefore, we introduce MINLP to the IM problem, and provide derivative-free methods as general solution methods. Furthermore, we propose the CDS method particularly suited for the IM problem when the information propagation is governed by the GIP model, and numerically demonstrate its close-to-optimal performance in various scenarios through experiments. One can also consider fine tuning each step of the CDS method as possible extensions to further improve the performance.

In summary, in this paper, we have unified the mechanisms underlying the IC model and the LT model into a general class of information propagation model, and propose a general framework for the IM problem that is applicable to a broad range of functions describing the overall influence. As the two classic models are widely accepted for information propagation, we believe that the GIP model has the potential to explain more propagation phenomena on, but not restricted to, social networks. Meanwhile, the proposed IM framework provides a systematic approach to handle the case when the objective does not have desired properties (e.g.,~submodularity), which provides insights into solving the IM problem in more realistic scenarios.

\begin{acknowledgments}
Y.T.~is funded by the EPSRC Centre for Doctoral Training in Industrially Focused Mathematical Modelling (EP/L015803/1) in collaboration with Tesco PLC. R.L.~acknowledges support from the EPSRC Grants EP/V013068/1 and EP/V03474X/1. We also thank Sebastian Lautz, Alisdair Wallis, Karel Devriendt, and Lindon Roberts for useful discussions.
\end{acknowledgments}

\appendix
\section{Proofs}
\subsection{\label{app:proofs_model} Proofs for Sec.~\ref{sec:general_diffusion_model}}
\begin{proof}[Proof of Lemma \ref{lem:linear_equivalence-lower}]
	When a node $v_j$ has $\sum_{i}W_{ij}x_i(t-1) > 0$, $\exists v_i\in V$ s.t.~$x_i(t-1)W_{ij} > 0$. Then
	\begin{align*}
	    \sum_{i}W_{ij}x_i(t-1) \ge wx_*(t-1),
	\end{align*} 
	where $x_*(t-1) = \min\{x_i(t-1):x_i(t-1)>0\}$. Hence, we can show that statement \eqref{equ:special_weightsum} is true by proving
	\begin{align*}
	    wx_*(t-1) \ge l_{\min, 0}w^t \Leftrightarrow x_*(t-1) \ge l_{\min, 0}w^{t-1},\quad \forall t > 0,
	\end{align*}
	by induction. (i) When $t = 1$, $x_*(0) \ge l_{\min, 0}$. (ii) Suppose that $x_*(t-1) \ge l_{\min, 0}w^{t-1}$ is true $\forall t \le t'$. Then when $t = t'+1$, from the updating function of the GIP model, for each node $v_j$ with $x_j(t') > 0$,
	\begin{align*}
	    x_j(t') &= f_{j, t'}(\sum_iW_{ij}x_i(t'-1))\\
	    &\ge f_{j, t'}(wx_*(t'-1)) \ge f_{j, t'}(l_{\min, 0}w^{t'}) = l_{\min, 0}w^{t'}, 
	\end{align*}
	where the first inequality is obtained by $f_{j, t'}(\cdot)$ being a nondecreasing function, the second one is obtained together with the induction hypothesis, and the equality at the end is because $l_{j, t} \le l_{\min, 0}w^t \le h_{j, t}$, $\forall t > 0,\ v_j\in V$. 
\end{proof}

\begin{proof}[Proof of Theorem \ref{the:linear_equivalence}]
    We show that given such bounds in the GIP model, 
	\begin{align}
	    \mathbf{x}(t)^T = \mathbf{x}(0)^T\mathbf{W}^t,\quad \forall t > 0,
	\label{equ:extreme1_linear-induct}
	\end{align}
	by induction. By Lemma \ref{lem:linear_equivalence-lower}, there is no threshold effect from the current lower bounds, thus we will only check the saturation effect from the upper bounds in the following. 
	
	(i) When $t=1$, for each node $v_j$,
	\begin{align*}
	    \mathbf{x}(0)^T\mathbf{W}_{:, j} \le \mathbf{h}_0^T\mathbf{W}_{:, j} \le h_{j,1},
	\end{align*}
	since $x_i(0)\le h_{j,0}$ and $W_{ij}\ge 0$, $\forall i,j$. Hence,
	\begin{align*}
	    x_j(1) = f_{j,1}(\mathbf{x}(0)^T\mathbf{W}_{:, j}) = \mathbf{x}(0)^T\mathbf{W}_{:, j},
	\end{align*}
	and $\mathbf{x}(1)^T = \mathbf{x}(0)^T\mathbf{W}$. 
	
	(ii) Suppose $\mathbf{x}(t)^T = \mathbf{x}(0)^T\mathbf{W}^t$, $\forall t \le t'$. Then, for each node $v_j$,
	\begin{align*}
	    \mathbf{x}(t')^T\mathbf{W}_{:, j} = \mathbf{x}(0)^T\mathbf{W}_{:, j}^{t'+1} \le \mathbf{h}_0^T\mathbf{W}_{:, j}^{t'+1} \le h_{j, t'+1},
	\end{align*}
	where the equality is by the induction hypothesis, and the first inequality is again by $x_i(0)\le h_{j,0}$ and $W_{ij}\ge 0$, $\forall i,j$. Hence, 
	\begin{align*}
	    x_j(t'+1) = f_{j,t'+1}(\mathbf{x}(0)^T\mathbf{W}_{:, j}^{t'+1}) = \mathbf{x}(0)^T\mathbf{W}_{:, j}^{t'+1},
	\end{align*}
	and $\mathbf{x}(t'+1)^T = \mathbf{x}(0)^T\mathbf{W}^{t'+1}$.
	It is then straightforward to show that the dynamics characterized by (\ref{equ:extreme1_linear-induct}) has updating function (\ref{equ:linear_dynamics}) of the EIC model. 
\end{proof}

\begin{proof}[Proof of Theorem \ref{the:LT_equivalence}]
    We first note that if \eqref{equ:dyn_gen_LT-equ} is true, then for each node $v_j\in V$, 
    \begin{align*}
        \sum_{t=0}^{\infty}(1-\gamma)^tx_j(t) 
        &= \sum_{t=0}^{\infty}(1-\gamma)^t(\theta_l\alpha)^tx'_j(t)\\ 
        &= \sum_{t=0}^{\infty}(1-\gamma')^tx'_j(t), 
    \end{align*}
    thus the network has the same overall influence from the two models. 
    
    We then show that \eqref{equ:dyn_gen_LT-equ} is true by induction on the time step $t$. (i) At $t=0$, $x_j(0) = (\theta_l\alpha)^0x'_j(0)$, since $x_j(0) = x'_j(0),\ \forall v_j\in V$. (ii) Suppose $x_j(t) = (\theta_l\alpha)^tx'_j(t),\ \forall v_j\in V$, is true for all $t\le t'$, then for each node $v_j$ at $t = t'+1$, if we denote $y_j(t'+1) = \sum_iW_{ij}x_i(t')$,
    \begin{align}
	    x_j(t'+1) = 
    	\begin{cases}
    		0, &\quad y_j(t'+1) < l_{j,t'+1},\\
    		y_j(t'+1), &\quad l_{j,t'+1} \le y_j(t'+1) < h_{j,t'+1},\\
    		h_{j,t'+1}, &\quad y_j(t'+1) \ge h_{j,t'+1},
    	\end{cases}
    	\label{equ:app_LT-revis}
    \end{align}
    where $l_{j,t'+1} = (\theta_l\alpha)^{t'+1}$ and $h_{j,t'+1} = \theta_{h}\theta_l^{t'}\alpha^{t'+1}h_{j,0}$, by the GIP model. We now consider the state value $x'_j(t'+1)$ from the MLT model in the three different cases, and compare it with the state value $x_j(t'+1)$ in (\ref{equ:app_LT-revis}). (1) When 
    \begin{align*}
        &\sum_iW_{ij}x_j(t') < l_{j,t'+1} = (\theta_l\alpha)^{t'+1}\\
        \Leftrightarrow &\sum_iW_{ij}(\theta_l\alpha)^{t'}x'_i(t') <  (\theta_l\alpha)^{t'+1}\\
        \Leftrightarrow &\sum_iW_{ij}x'_j(t') < \theta_l\alpha = l'_{j},
    \end{align*}
   we have $(\theta_l\alpha)^{t'+1}x'_j(t'+1) = 0 = x_j(t'+1)$.\\ (2) When
    \begin{align*}
        &\sum_iW_{ij}x_j(t') \ge h_{j,t'+1} = \theta_h\theta_l^{t'}\alpha^{t'+1}h_{j,0}\\
        \Leftrightarrow &\sum_iW_{ij}(\theta_l\alpha)^{t'}x'_i(t') \ge  \theta_h\theta_l^{t'}\alpha^{t'+1}h_{j,0}\\
        \Leftrightarrow &\sum_iW_{ij}x'_j(t') \ge \theta_h\alpha h_{j,0} = h'_{j},
    \end{align*}
   we can then have $(\theta_l\alpha)^{t'+1}x'_j(t'+1) = (\theta_l\alpha)^{t'+1}m_j = (\theta_l\alpha)^{t'+1}(\theta_h h_{j,0})/\theta_l = h_{j,t'+1} = x_j(t'+1)$.\\
   (3) Finally, in the remaining case when 
    \begin{align*}
        &l_{j,t'+1} \le \sum_iW_{ij}x_j(t') < h_{j,t'+1}\\
        \Leftrightarrow\ &l'_{j} \le \sum_iW_{ij}x'_j(t') < h'_{j},
    \end{align*}
    the state value is in the linear regime where 
    \begin{align*}
        x'_j(t'+1) &= \frac{m_j - 1}{h'_j - l'_j}(\sum_iW_{ij}x'_j(t') - l'_j) + 1\\
        &= \frac{(\theta_h h_{j,0})/\theta_l - 1}{\theta_h\alpha h_{j,0} - \theta_l\alpha}(\sum_iW_{ij}x'_j(t') - \theta_l\alpha) + 1\\
        &= \frac{1}{\theta_l\alpha}\sum_iW_{ij}x'_j(t') = \frac{1}{\theta_l\alpha}\sum_iW_{ij}\frac{1}{(\theta_l\alpha)^{t'}}x_j(t')\\
        &= \frac{1}{(\theta_l\alpha)^{t'+1}}\sum_iW_{ij}x_j(t') = \frac{1}{(\theta_l\alpha)^{t'+1}}x_j(t'+1). 
    \end{align*}
    Hence, we have shown that $x_j(t'+1) = (\theta_l\alpha)^{t'+1}x'_j(t'+1), \forall v_j\in V$. 
\end{proof}

\begin{proof}[Proof of Claim \ref{cla:linear_difference}]
    For the SBM, $\mathbf{W} = \alpha\mathbf{A}$, where $\mathbf{A}$ is the (unweighted) adjacency matrix, and for each pair of nodes $v_i, v_j$, $A_{ij} \sim Bernoulli(p_{ij})$, with
	\begin{align*}
	    p_{ij} = p_{in}\delta(\sigma_i, \sigma_j) + p_{out}(1 - \delta(\sigma_i, \sigma_j)),
	\end{align*}
	where $\sigma_i\in\{1,2\}$ indicates the block membership of each node $v_i$, and $\delta(i,j)$ is the delta function where $\delta(i,j) = 1$ if and only if $i = j$ and $0$ otherwise. Further, we denote the linear part of the state vector by $\mathbf{y}(t+1) = \mathbf{W}^T\mathbf{x}(t)$, then for each node $v_j$ at each time step $t > 0$,
	\begin{align*}
	   \begin{cases}
	        x_j(t) &= f_{j,t}(y_j(t)),\\
	        y_j(t) &= \sum_{i}W_{ij}x_{i}(t-1) = \alpha\sum_{i}A_{ij}x_{i}(t-1).
	   \end{cases}
	\end{align*} 
	Then at $t = 1$ \footnote{\footnotesize{The expressions of initially activated nodes can be different from others in the same community, due to the common assumption of no self-edges. However, noting that $k\ll n$, we allow self-edges for illustrative purposes, thus ignore such differences.}}, 
	\begin{widetext}
	\begin{align*}
	    y_j(1) &= \alpha\sum_{i}A_{ij}x_{i}(0) = \alpha l_0\sum_{v_i\in\mathcal{A}_0}A_{ij} = \alpha l_0\left(\sum_{v_i\in\mathcal{A}_0\cap\mathcal{B}_{\sigma_j}}Bernoulli(p_{in}) + \sum_{v_i\in\mathcal{A}_0\backslash \mathcal{B}_{\sigma_j}}Bernoulli(p_{out})\right)\\
	    &= \alpha l_0\left[Bin(k_{\sigma_j}, p_{in}) + Bin(k-k_{\sigma_j}, p_{out})\right],
	\end{align*}
	\end{widetext}
	where $\mathcal{A}_0 = \{v_i: x_i(0) > 0\}$ is the (given) set of initially activated nodes, $k = \abs{\mathcal{A}_0}$, and $k_i = \abs{\mathcal{A}_0\cap\mathcal{B}_i}$, $i = 1,2$.
	
	Hence, for set (i), 
	\begin{align*}
	    y_j(1) = \alpha l_0\left[Bin(2, p_{in})\delta(\sigma_j, 1) + Bin(2, p_{out})\delta(\sigma_j, 2)\right], 
	\end{align*}
	while for set (ii),
	\begin{align*}
	    y_j(1) = \alpha l_0\left[Bin(1, p_{in}) + Bin(1, p_{out})\right]. 
	\end{align*}
	When $l_{j,1} \le l_1^*$, $x_j(1) = y_j(1)$, $\forall v_j\in V$. Then, for set (i), 
	\begin{align*}
	    \mathbb{E}\left[\sum_{j}(1-\gamma)x_j(1)\right] =  (1-\gamma)n_b\times\alpha l_0\times(2p_{in} + 2p_{out}),
	\end{align*}
	where $n_b$ is the size of each community, and for set (ii),
	\begin{align*}
	\mathbb{E}\left[\sum_{j}(1-\gamma)x_j(1)\right] =  2(1-\gamma)n_b\times\alpha l_0\times (p_{in} + p_{out}).
	\end{align*}
	Hence, the expected influence values at $t=1$ are the same from the two sets. 
	
	However, when $\alpha l_0 = l_1^* < l_{j,1} \le 2l_1^* = 2\alpha l_0$, for set (i), 
	\begin{align*}
	    &P\left[y_j(1) \ge l_{j,1}\right] = P\left[y_j(1) = 2\alpha l_0\right] \\
	    &= P\left[Bin(2, p_{in})\delta(\sigma_j, 1) + Bin(2, p_{out})\delta(\sigma_j, 2) = 2\right] \\
	    &= p_{in}^2\delta(\sigma_j, 1) + p_{out}^2\delta(\sigma_j, 2),
	\end{align*}
	thus, 
	\begin{align}
	    \begin{split}
	        &\mathbb{E}\left[\sum_{j}(1-\gamma)x_j(1)\right] \\
	        &= (1-\gamma)\sum_{j} \{0P\left[y_j(1) < l_{j,1}\right] + 2\alpha l_0P\left[y_j(1) = 2\alpha l_0\right]\} \\
	        &= (1-\gamma)n_b\times 2\alpha l_0\times (p_{in}^2 + p_{out}^2). 
	    \end{split}
	\label{equ:dyn_gen_lin_eg-expect-1}
	\end{align}
	While, for set (ii), 
	\begin{align*}
	    &P\left[y_j(1) \ge l_{j,1}\right] = P\left[y_j(1) = 2\alpha l_0\right]\\ 
	    &= P\left[Bin(1, p_{in}) + Bin(1, p_{out}) = 2\right] = p_{in}p_{out},
	\end{align*}
	thus, 
	\begin{align}
	    \begin{split}
	        &\mathbb{E}\left[\sum_{j}(1-\gamma)x_j(1)\right]\\ 
	        &= (1-\gamma)\sum_{j} \{0P\left[y_i(1) < l_{j,1}\right] + 2\alpha l_0P\left[y_j(1) = 2\alpha l_0\right]\}\\
	        &= 2(1-\gamma)n_b\times 2\alpha l_0\times p_{in}p_{out}.
	    \end{split}
	\label{equ:dyn_gen_lin_eg-expect-2}
	\end{align}
	Hence, the expectation from (i) as in (\ref{equ:dyn_gen_lin_eg-expect-1}) is larger than the one from (ii) as in (\ref{equ:dyn_gen_lin_eg-expect-2}) by condition (\ref{equ:extreme1_devSBM-cond}). 
\end{proof}

\begin{proof}[Proof of Claim \ref{cla:extreme2_devi-tree}]
    Since the underlying network has uniform weight $\alpha$, $\mathbf{W} = \alpha\mathbf{A}$. (1) When $\theta_h = \theta_l > 1$, $x_i(t') = (\theta_l\alpha)^{t'}l_0,\ \forall v_i\in \mathcal{A}_{t'}$. Then we have
	\begin{align*}
	    \sum_{i}W_{ij_*}x_i(t') = \alpha(\theta_l\alpha)^{t'}l_0 < \theta_l\alpha(\theta_l\alpha)^{t'}l_0 = l_{j_*, t'+1},
	\end{align*}
	thus node $v_{j_*}$ cannot have a positive state value at $t'+1$.
	
	(2) When $\theta_h > \theta_l > 1$, the highest possible value of node $j_0$ at time $t'$ is $\theta_h\theta_l^{t'-1}\alpha^{t'}l_0$. Hence, node $v_{j_*}$ can have positive state value at $t'+1$ if 
	\begin{align*}
	    W_{j_0j_*}\theta_h\theta_l^{t'-1}\alpha^{t'}l_0 &\ge l_{j_*, t'+1}\\
	    \Leftrightarrow \alpha\theta_h\theta_l^{t'-1}\alpha^{t'}l_0
	    &\ge (\theta_l\alpha)^{t'+1}l_0\\
	    \Leftrightarrow \theta_h &\ge \theta_l^2,
	\end{align*}
	which can be achieved given that $\theta_h$ is sufficiently large.
\end{proof}

\subsection{\label{app:proofs_IM} Proofs for Sec.~\ref{sec:influence_maximization}}
\begin{proof}[Proof of Theorem \ref{the:spe_contin-concave}]
    By Lemma \ref{lem:linear_equivalence-lower}, the lower bounds are effectively $0$. Therefore, in the GIP model, $f_{j,t}(x) = 0$ if and only if $x = 0$, $\forall v_j\in V,\ t>0$. Hence, $f_{j,t}(x)$ is equivalent to another bound function $\tilde{f}_{j,t}(x)$ associated with the same upper bound $\tilde{h}_{j,t} = h_{j,t}$ but a different lower bound $\tilde{l}_{j,t} = 0$. 
	
	We first note that if $f_{j,t}(x)$ is continuous for all $v_j\in V$ and $t > 0$, then by the properties of composite continuous functions, the objective function $s(\cdot)$ is also continuous. Now, we focus on $f_{j,t}(x)$, and show that it is continuous by proving the continuity of $\tilde{f}_{j,t}(x)$. (1) There are two parts of the function that are always continuous, when $\tilde{l}_{j,t} < x < \tilde{h}_{j,t}$ and when $x > \tilde{h}_{j,t}$. (2) We can show that the function is also continuous at the boundary points, where 
	\begin{align*}
	    \lim_{x\to \tilde{h}_{j,t}^-}\tilde{f}_{j,t}(x) = \tilde{h}_{j,t} = \tilde{f}_{i,t}(\tilde{h}_{j,t}) = \lim_{x\to \tilde{h}_{j,t}^+}\tilde{f}_{j,t}(x),
	\end{align*}
	and 
	\begin{align*}
	    \lim_{x\to \tilde{l}_{j,t}^+}\tilde{f}_{j,t}(x) = 0 = \tilde{f}_{j,t}(\tilde{l}_{j,t}).
	\end{align*}
	Hence, $\tilde{f}_{j,t}(x)$ is continuous for all $x \ge 0$. 
	
	Then for the concavity, we note that if $f_{j,t}(x)$ is concave for all $v_j\in V$ and $t > 0$, then since it is also nondecreasing (and the linear function $h(\mathbf{x}) = \mathbf{W}^T\mathbf{x}$ is also concave and nondecreasing), the objective function $s(\cdot)$ is concave by the properties of composite concave functions. Now, we consider specifically $f_{j,t}(x)$, and show that it is concave by the concavity of $\tilde{f}_{j,t}(x)$, $\forall x, y\ge 0$ and $\beta\in[0,1]$, 
	\begin{align}
	    \tilde{f}_{j,t}((1-\beta)x + \beta y) \ge (1-\beta)\tilde{f}_{j,t}(x) + \beta\tilde{f}_{j,t}(y).
	    \label{equ:concave}
	\end{align}
	(1) When $0 = \tilde{l}_{j,t} \le x,y < \tilde{h}_{j,t}$ or $x,y \ge \tilde{h}_{j,t}$, (\ref{equ:concave}) is true by the concavity of linear functions and constant functions, respectively. \\
	(2) When $\tilde{l}_{j,t} \le x < \tilde{h}_{j,t} \le y$, $\tilde{f}_{j,t}(x) = x < \tilde{h}_{j,t}$ and $\tilde{f}_{j,t}(y) = \tilde{h}_{j,t} \le y$. Then, if $(1-\beta)x + \beta y \ge \tilde{h}_{j,t}$, 
	\begin{align*}
	    \tilde{f}_{j,t}((1-\beta)x + \beta y) &= \tilde{h}_{j,t}\\
	    &= (1-\beta)\tilde{h}_{j,t} + \beta\tilde{f}_{j,t}(y)\\ &\ge (1-\beta)\tilde{f}_{j,t}(x) + \alpha\tilde{f}_{j,t}(y); 
	\end{align*}
	otherwise $(1-\beta)x + \beta y < \tilde{h}_{j,t}$, 
	\begin{align*}
    	\tilde{f}_t((1-\beta)x + \beta y) &= (1-\beta)x + \beta y \\
    	&= (1-\beta)\tilde{f}_{j,t}(x) + \beta y \\
    	&\ge (1-\beta)\tilde{f}_{j,t}(x) + \beta\tilde{f}_{j,t}(y).
	\end{align*}
	(3) When $\tilde{l}_{j,t} \le y < \tilde{h}_{j,t} \le x$, (\ref{equ:concave}) is true by exchanging $x,y$ in case (2). Hence, $\tilde{f}_{j,t}(x)$ is concave for all $x \ge 0$.
\end{proof}
\begin{remark}
    Furthermore, we can also show that the objective function is Lipschitz continuous by noting that $0\le f_{j,t}(x) \le x$, $\forall t>0, v_j\in V$ and potential value $x$. In this case, there are methods proven to have global convergence, for example the new derivative-free line-search type algorithms \cite{Giovannelli_DFL_2021}. Meanwhile, the condition of the lower bounds could be looser in practice, since the propagation does not necessarily go through the edge(s) of the smallest weight in every time step. Therefore, we can have a larger region of the parameters where the objective function is (Lipschitz) continuous and concave.
\end{remark}

\begin{proof}[Proof of Theorem \ref{the:lineard_exact_sol}]
    By Theorem \ref{the:linear_equivalence}, the GIP model reaches the extreme of the EIC model. Then as in (\ref{equ:dyn_lineard_obj}), the objective function $s(\cdot)$ is linear in $\mathbf{x}$,
	\begin{align*}
	    s(\mathbf{x}) = \mathbf{c}^T\mathbf{x},
	\end{align*}
	where $\mathbf{c} = ((\mathbf{I} - (1-\gamma)\mathbf{W})^{-1} - \mathbf{I})\mathbf{1}$. For illustrative purposes, we split the proof into two parts. (1) We first analyze the MINLP (\ref{equ:opt_dyn_gen}) solely w.r.t.~$\mathbf{x}$ while fixing the integer variables $\mathbf{z}$. The problem can then be decomposed into $n$ sub-problems, where for each $v_j\in V$,
	\begin{align}
    	\begin{split}
        	\max_{x_j} \quad &\tilde{s}_j(x_j) \coloneqq c_jx_j\\
        	s.t. \quad & x_j \le h_{j,0}z_j,\\
        	& x_j \ge l_{j,0}z_j,\\
        	&x_j\in \mathbb{R}.
    	\end{split}
    	\label{equ:app_opt_dyn_gen-sub}
	\end{align}
	Because $c_j, z_j\ge 0,\ \forall v_j\in V$, we can show that the optimal solution to each sub-problem (\ref{equ:app_opt_dyn_gen-sub}) is $x^*_j = h_{j,0}z_i$, and the optimal value is $\tilde{s}^*_j = c_jh_{j,0}z_j$. (2) Then we consider the MINLP (\ref{equ:opt_dyn_gen}) w.r.t.~$\mathbf{z}$ when $\mathbf{x}$ is at its optimal value, where
	\begin{align*}
    	\begin{split}
    	\max_{\mathbf{z}} \quad &\sum_{j}\tilde{s}^*_j = \sum_jc_jh_{j,0}z_j\\
    	s.t. \quad &\sum_j z_j \le k,\\
    	&z_j\in \{0,1\}, \forall j.
    	\end{split}
	\end{align*}
	We can show that the optimal solution is to set $z_j = 1$ if node $j$ is ranked among the top $k$ according to its coefficient in the objective, $c_ih_{j,0}$. This gives the solution (\ref{equ:lineard_exact_sol}), and the uniqueness of the solution depends on the uniqueness of the top $k$ nodes. 
\end{proof}

\begin{proposition}
    In the MINLP \eqref{equ:opt_dyn_gen} with the GIP model governing the information propagation process and Eq.~\eqref{equ:dyn_gen_influ-j} as the function for individual influence, the objective function $s(\mathbf{x})$ is nondecreasing in $\mathbf{x}$. 
\end{proposition}
\begin{proof}
    Here, we denote the vector $\mathbf{x}$ by $\mathbf{x}(0)$ to emphasize the it corresponds to the state values at $t=0$. We first show that $\mathbf{x}(t)$ is nondecreasing in $\mathbf{x}(0)$, $\forall t > 0$, by induction. (i) $\mathbf{x}(1)$ is nondecreasing in $\mathbf{x}(0)$, since each $f_{j,1}(\cdot)$ and each linear function $h_j(\mathbf{x}) = \sum_iW_{ij}x_i$ in (\ref{equ:dyn_gen_revis}) are nondecreasing. (ii) Suppose $\mathbf{x}(t)$ is nondecreasing in $\mathbf{x}(0)$, $\forall t\le t'$. Then, since $\mathbf{x}(t'+1)$ is nondecreasing in $\mathbf{x}(t')$ by the same logic as (i) and $\mathbf{x}(t')$ is nondecreasing in $\mathbf{x}(0)$ by the induction hypothesis, we have $\mathbf{x}(t'+1)$ is nondecreasing in $\mathbf{x}(0)$. Hence, $s(\mathbf{x}(0)) = \sum_{j}\sum_{t=0}^{\infty}(1-\gamma)^tx_j(t)$ is nondecreasing in $\mathbf{x}(0)$. 
\end{proof}

\section{\label{app:diff_further_features} Further features of the GIP model}
In this section, we discuss other interesting features of the GIP model we have developed in Sec.~\ref{sec:general_diffusion_model}. We first show that with the GIP model, both the locally EIC-like and the locally ELT-like propagation can coexist in a single network in Sec.~\ref{app:diff_coexistence}, and then numerically analyze this feature in Sec.~\ref{app:experiment_coexist}. We conclude with a discussion of the derivative information in Sec.~\ref{app:diff_derivative}, which plays an important role in the IM problem in Sec.~\ref{sec:influence_maximization}.

\subsection{\label{app:diff_coexistence} Coexistence of regimes}
As discussed in Sec.~\ref{sec:general_diffusion_model}, an important feature of the GIP model is that it can be equivalent to the EIC model at one end and the ELT model at the other. Here, we show that both types of propagation can coexist in a single network given that the underlying propagation process follows the GIP model, which further illustrates the generality of the GIP model. 

Specifically, we construct a network as in Fig.~\ref{fig:illustr_coexist}, where we connect a tree, composed by nodes $\cup_{i=1, i\ne 7}^{11}\{v_i\}\cup\{v_{12}, v_{13}, v_{15}, v_{16}, v_{17}, v_{21}\}$, and part of a regular lattice, composed of nodes $\{v_6, v_7, v_8, v_{13}, v_{14}, v_{15}, v_{18}, v_{19}, v_{20}\}$, with each edge being bidirectional. Suppose we activate the nodes in red initially. Then (1) if the underlying propagation is perfectly linear, every nodes will have positive state values as soon as possible, while (2) if the propagation follows the ELT model (requiring more than one active neighbors), some nodes in the tree substructure, e.g.,~leaf nodes $v_3, v_{12}, v_{16}, v_{17}, v_{21}$, will never have positive state values. However, with the GIP model, we can have both characteristics in the propagation, where (1) nodes such as $v_2$ and $v_3$ are influenced immediately after they have active neighbors, while (2) nodes such as $v_{12}$ and $v_{17}$ do not take positive state values when there is certain amount of influence in their neighborhood for the first time; see Table \ref{tab:illustr_coexist_state}. Here,  we apply the GIP model with threshold-type bounds (\ref{equ:dyn_gen_thres-bBt}) satisfying (\ref{equ:dyn_gen_thres-uni}) and $l_{j,0} = h_{j,0} = 1,\ \forall v_j\in V$, and we choose $\theta_h = 8$ to be slightly larger than $\theta_l = 2$. 

\begin{figure}
	\centering
	\includegraphics[width=.35\textwidth]{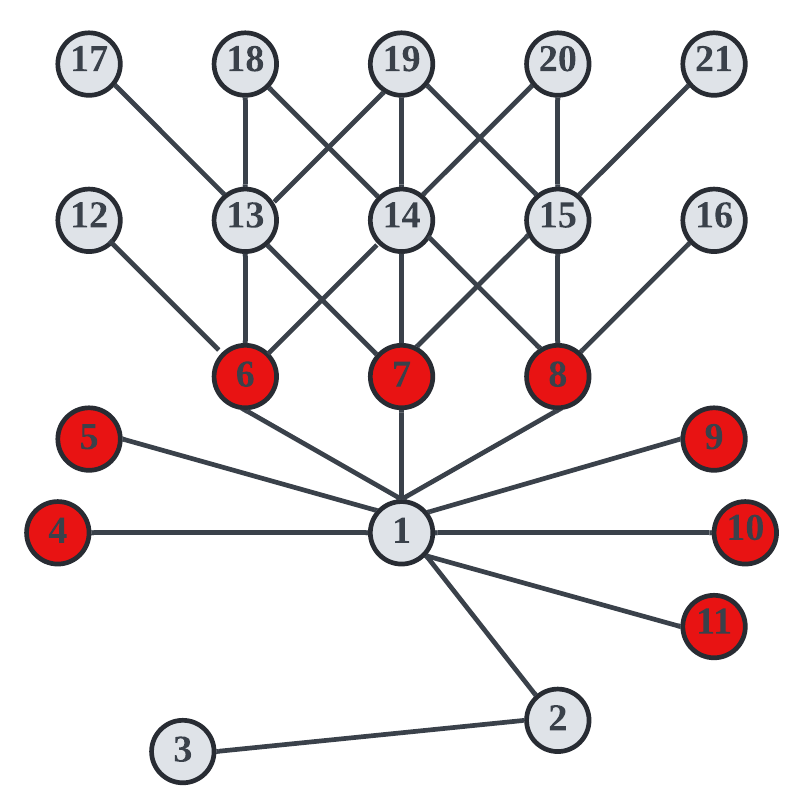}
	\caption{Example network with both tree substructure and regular-lattice-like substructure.}
	\label{fig:illustr_coexist}
\end{figure}
\begin{table}
	\centering
	\caption{The state values of representative nodes in Fig.~\ref{fig:illustr_coexist}, where the underlying propagation follows the GIP model with the threshold-type bounds (\ref{equ:dyn_gen_thres-bBt}) satisfying (\ref{equ:dyn_gen_thres-uni}), $l_{j,0} = h_{j,0} = 1,\ \forall v_j\in V$, and $\theta_l = 2, \theta_h = 8$, and the initially activated nodes are in red.}
	\label{tab:illustr_coexist_state} 
	\begin{tabular}{c|ccccccc}
		Node &  1 & 2 & 3 & 12 & 13 & 17 & 18 \\
		\hline
		$t = 1$ & $8\alpha$   & $0$            & $0$                 & $0$             & $2\alpha$  & $0$ & $0$ \\
		$t = 2$ & $0$             & $8\alpha^2$ & $0$                 & $0$             & $0$            & $0$ & $4\alpha^2$ \\
		$t = 3$ & $32\alpha^3$ & $0$            & $8\alpha^3$ & $12\alpha^3$ & $32\alpha^3$ & $0$ & $0$
	\end{tabular}
\end{table}

Although enlarging the upper bounds can make the propagation governed by the GIP model perform as the EIC model (i.e.,~linear dynamics) in some tree structures, another feature of the GIP model is that this can only happen within limited number of time steps. Specifically with the threshold-type bounds (\ref{equ:dyn_gen_thres-bBt}) satisfying (\ref{equ:dyn_gen_thres-uni}) and $l_{j,0} = h_{j,0} = 1,\ \forall v_j\in V$, suppose a node $v_{j_0}$ hits the upper bounds at $t=t'>0$, $x_{j_0}(t') = \theta_h\theta_l^{t'-1}\alpha^{t'}$, and it is the only source of influence for the nodes within certain distance $r_0$. Then nodes of distance $r \le r_0$ can have positive state values if 
\begin{align*}
    \alpha^r(\theta_h\theta_l^{t'-1}\alpha^{t'}) &\ge l_{j,t'+r} = (\theta_l\alpha)^{t'+r}\\
    \Leftrightarrow \theta_h &\ge \theta_l^{r+1}. 
\end{align*}
Hence, for the network in Fig.~\ref{fig:illustr_coexist}, if there are extra nodes connecting with node $v_3$ only or further tree substructures with node $v_3$ as the root node, they will not have positive state values. In real networks, we can consider the case where a highly active person is very likely to influence their friends and further friends, but reinforcement among friends will be required at some point to convince people even further.

\subsection{\label{app:experiment_coexist}Numerical experiments}
We now investigate the feature discussed in Sec.~\ref{app:diff_coexistence} numerically. Specifically, we construct a network by connecting two very different structures, an Erd\H{o}s R\'enyi (ER) random graph and a regular lattice, which we refer to as a \textit{composite network}, and explore the difference between the propagation processes from the GIP model at the ELT extreme (i.e.,~the ELT model) and that with other parameters, or in a general case. As in Sec.~\ref{sec:experiments_results}, $l_{j,0} = h_{j,0} = 1,\ \forall v_j\in V$, $\gamma = 0$, and we apply exclusively the threshold-type bounds in (\ref{equ:dyn_gen_thres-bBt}) with condition (\ref{equ:dyn_gen_thres-uni}). 

The composite network, \imp{with each edge being bidirectional}, is composed of the following: (i) a regular lattice of size $n_o$ with mean degree $d_o$, (ii) an ER random graph of the same size $n_o$ and with connecting probability $p_{er} = d_o/n_o$, and (iii) edges randomly placed between the two parts with a small probability $p_o$. Here, we choose $d_o = 4$, $n_o = 25$, $p_o = 0.01$, thus the network is of the same size as the SBM in Sec.~\ref{sec:experiments_accuracy} with $n = 50$. We label the nodes in the regular lattice as $0$ to $24$ and in the random part as $25$ to $49$; see Fig.~\ref{fig:Diff_smallworldrand_net} for one realization.
\begin{figure}
	\centering
	\includegraphics[width = .46\textwidth]{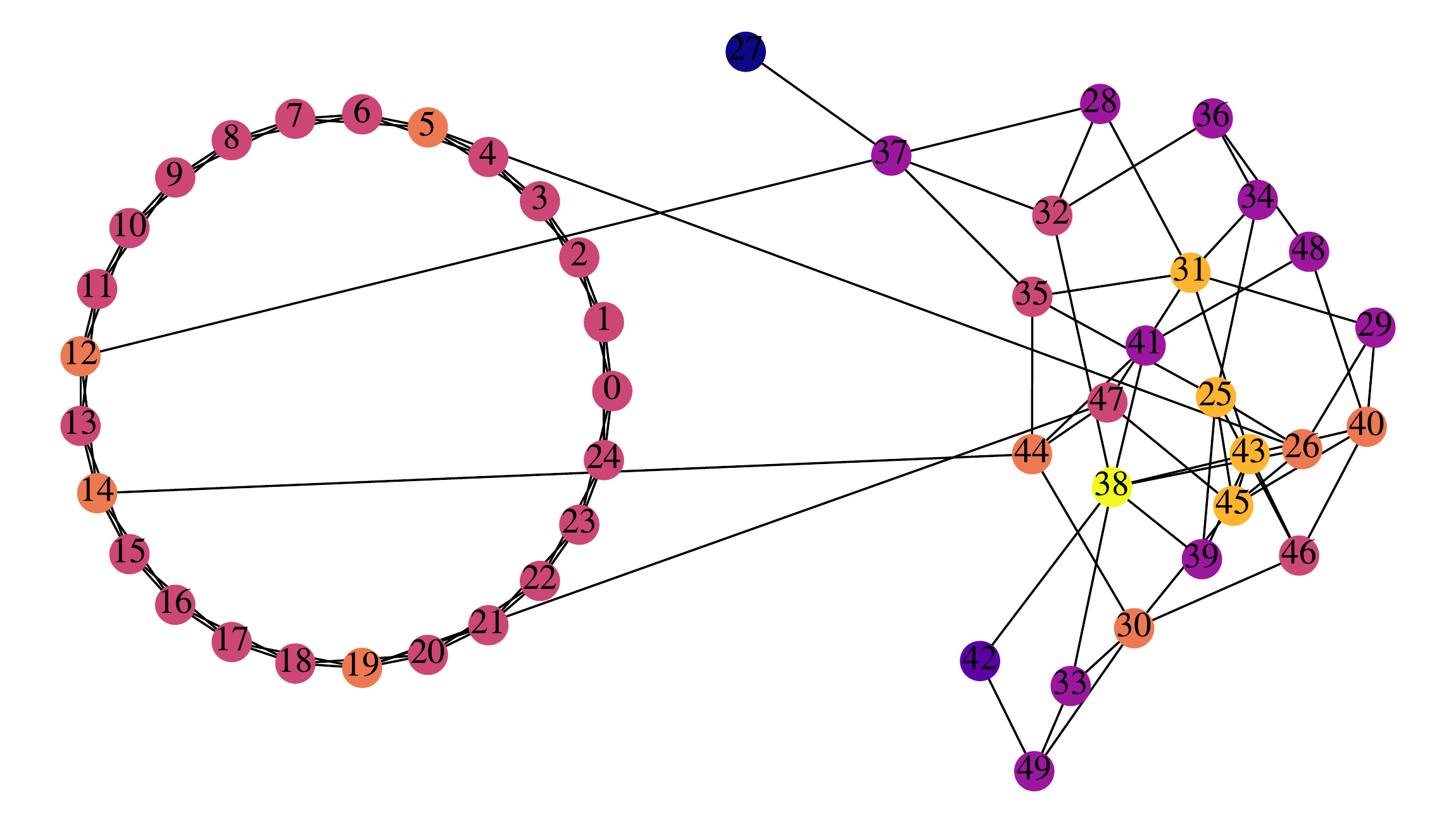}
	\caption{One realization of the composite graph, where nodes $0$ to $24$ are from a regular lattice and nodes $25$ to $49$ are from a ER random graph and the lighter the color of a node is, the higher the degree of the node is.}
	\label{fig:Diff_smallworldrand_net}
\end{figure}
We assign a uniform weight $\alpha=0.1$, and set the following parameters for the two cases of the GIP model: (i) $\theta_l = \theta_h = 2$ for the extreme of the ELT model, and (ii) $\theta_l = 2, \theta_h = 16$ for the general case for comparison. To emphasize the differences in the propagation spread, we consider the \textit{number} of nodes that have positive influence up to time step $t$, 
\begin{align*}
	n_a(t) = \sum_{j}I_{\{\sum_{t'=0}^{t}(1-\gamma)^{t'}x_j(t') > 0\}},
\end{align*}
where $I_{\{\cdot\}}$ is the indicator function, and $\lim_{t\to\infty}s_j(t) - x_j(0)$ is the overall influence on each node as in Eq.~\eqref{equ:dyn_gen_influ-j} with $s_j(t) \coloneqq \sum_{t'=0}^{t}(1-\gamma)^{t'}x_j(t')$.

\begin{figure*}
	\centering
	\begin{tabular}{cc}
		\includegraphics[width=.35\textwidth]{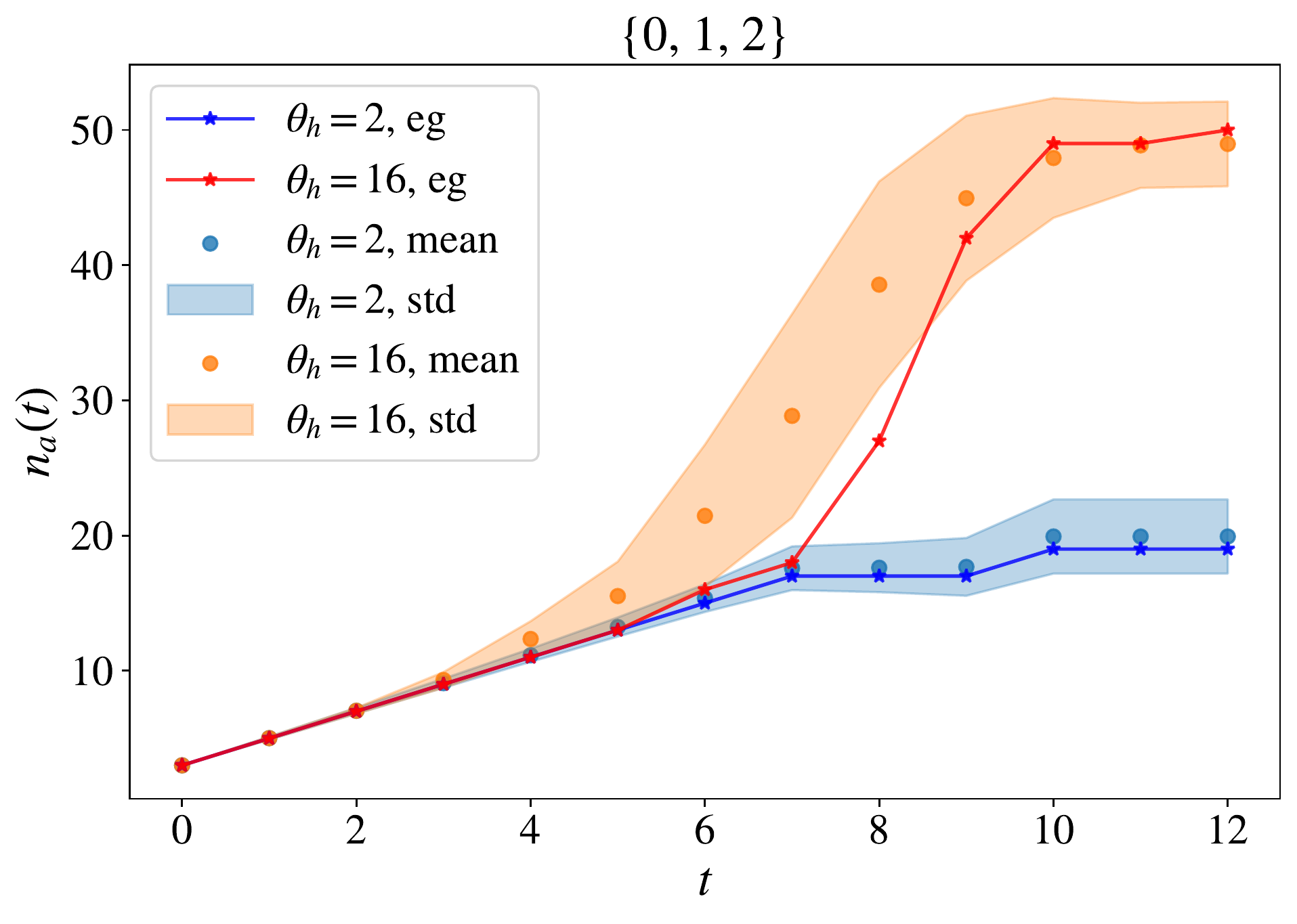} & \includegraphics[width=.35\textwidth]{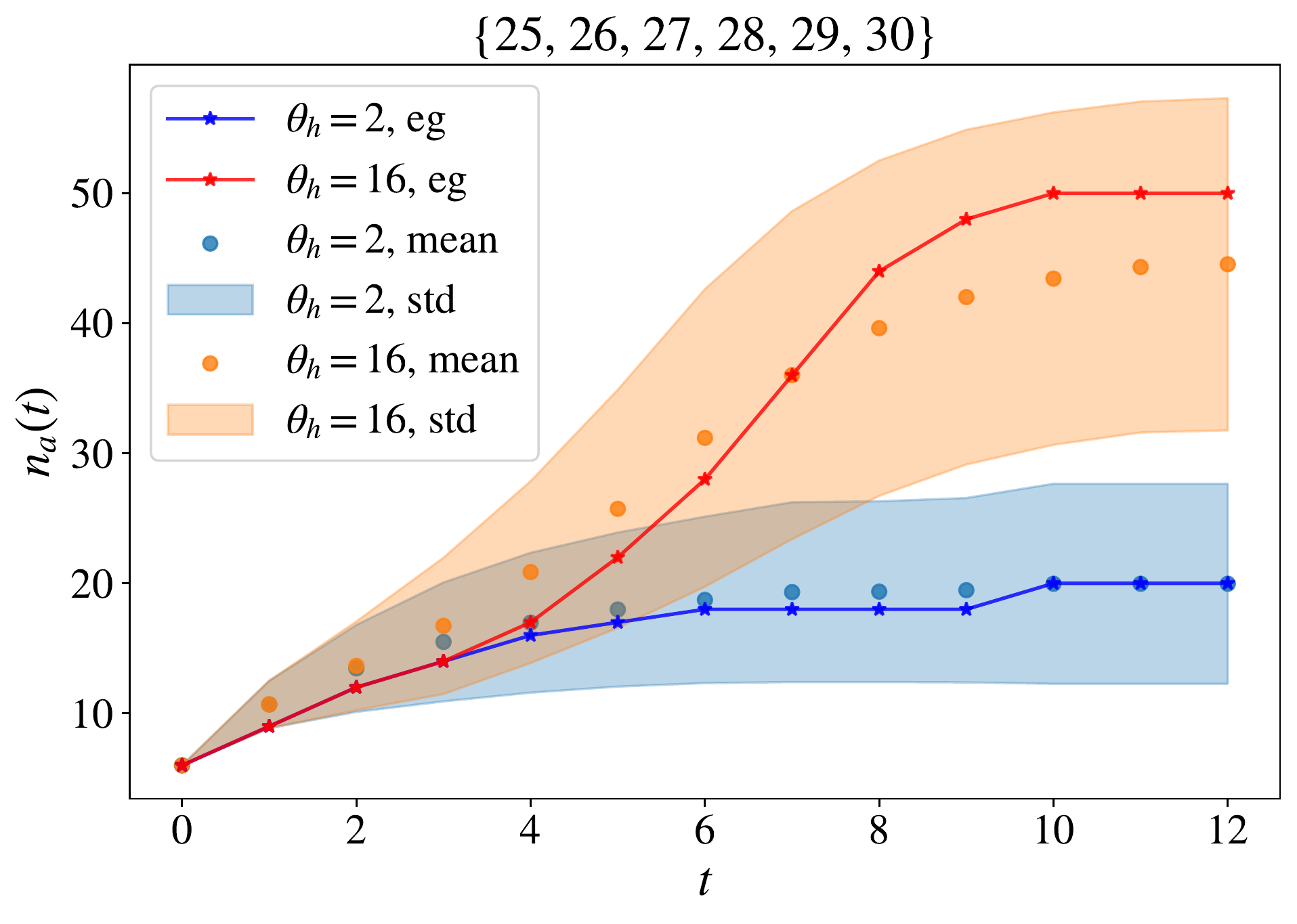}
	\end{tabular}
	\caption{The number of nodes of positive influence up to increasing time steps ($x$-axis), from the GIP model at the extreme of ELT model ($\theta_l = \theta_h = 2$) and in the general case ($\theta_l = 2, \theta_h = 16$), with the initially activated nodes in the regular lattice (left) and the ER random part (right), where ``eg" corresponds to the results on the composite network in Fig.~\ref{fig:Diff_smallworldrand_net} and others are from the results on $1000$ samples of the random composite network.}
	\label{fig:Diff_SmallWorldRand_As_ns1000}
\end{figure*}	
We observe that, particularly clearly in the composite network in Fig.~\ref{fig:Diff_smallworldrand_net}, the set from either part of the network can eventually influence the whole network when the GIP model is in the general case, while only certain parts of the network can be reached when the GIP model is at the extreme of the ELT model; see Fig.~\ref{fig:Diff_SmallWorldRand_As_ns1000}. The two propagation processes initially proceed similarly in terms of the number of nodes with positive influence, while after this ``initial preparation" phase, the general case gradually reaches more nodes and finally the whole network. The results are overall consistent with those obtained from $1000$ samples from the random composite network. Since the randomness is mostly from the ER model part, the results have more variance when the initially activated nodes are from this part. 

\subsection{\label{app:diff_derivative} Derivative and backpropagation}
Understanding how the overall influence change w.r.t.~the initial state values is important in many applications, e.g.,~the influence maximization problem we consider in Sec.~\ref{sec:influence_maximization}. With a functional form for the GIP model at each time step, we can then analyze the derivative information though \textit{backpropagation} or \textit{chain rule} given the function encoding the overall influence is differentiable, for this purpose.

Now, we slice the overall influence along the dimension of time, and consider the influence of all nodes at each time step $t$, 
\begin{align}
	s_t(\mathbf{x}(0)) = \sum_{j}(1-\gamma)^tx_j(t).
	\label{equ:dyn_gen_sumt}
\end{align}
Following the GIP model, $s_t(\mathbf{x}(0))$ can be considered as the output of a \textit{neural network}, with $t$ hidden layers, $\mathbf{x}(0)$ as the input layer, $\mathbf{W}^T$ as the weight matrix and $\{f_{j, t'}\}$ as the activation functions for nodes $\{v_j\}$ in the layers corresponding to $t' = 1,2,\dots, t$. The output layer only consists of one node, and is computed by summing over the elements in the previous layer corresponding to $\mathbf{x}(t) = (x_j(t))$.

We first note that $s_t$ is not always differentiable, and even discontinuous in the general form, because each bound function $f_{j,r}$ can have jump discontinuity at $l_{j,r}$ and be nondifferentiable at $h_{j,r}$ at each time step $r\le t$. However, $f_{j,r}$ is always semidifferentiable, specifically right-differentiable, with the right derivative, 
\begin{align}
	\partial_+f_{j,r}(x) = 
	\begin{cases}
		1, &\quad l_{j,r} \le x < h_{j,r},\\
		0, &\quad x < l_{j,r},\ x \ge h_{j,r},
	\end{cases}
	\label{equ:dyn_gen_fderiv}
\end{align}
thus so is the function $s_t$, by the chain rule of semidifferentiability and $W_{ij} > 0$, $\forall i,j$. Therefore, we can obtain the right derivative of $s_t$ with respect to $\mathbf{x}(0)$, $\partial_+\mathbf{s}_t = \{\partial_+s_t[x_j(0)]\}$, where if we denote $\mathbf{y}(r) = \mathbf{W}^T\mathbf{x}(r-1)$, and $\partial_+\mathbf{f}_{r} = \{\partial_+f_{j,r}[y_j(r)]\},\ \forall r\le t$, we have 
\begin{align*}
	\frac{1}{(1-\gamma)^t}\partial_+\mathbf{s}_t^T = \partial_+\mathbf{f}_{t}^T\frac{\partial \mathbf{y}(t)}{\partial\mathbf{x}(t-1)}\prod_{r=1}^{t-1} \mathbf{Diag}(\partial_+\mathbf{f}_{r})\frac{\partial\mathbf{y}(r)}{\partial\mathbf{x}(r-1)},
\end{align*}
where $\partial$ is the (bidirectional) derivative, $\partial\mathbf{y}/\partial\mathbf{x} = (\partial y_i/\partial x_j)$ with $\mathbf{x}, \mathbf{y}\in \mathbb{R}^n$, and $\mathbf{Diag}(\cdot)$ is the corresponding diagonal matrix. Here, $\partial\mathbf{y}(t)/\partial\mathbf{x}(t-1) = \mathbf{W}^T$, $\forall t > 0$, thus the right derivative can be reduced to
\begin{align}
	\frac{1}{(1-\gamma)^t}\partial_+\mathbf{s}_t^T = \partial_+\mathbf{f}_{t}^T\mathbf{W}^T\prod_{r=1}^{t-1}\mathbf{Diag}(\partial_+\mathbf{f}_{r})\mathbf{W}^T.
	\label{equ:dyn_gen_deriv}
\end{align}
From Eq.~(\ref{equ:dyn_gen_fderiv}), $\partial_+f_{j,r}(\cdot) \in \{0, 1\},\ \forall v_j\in V, r>0$, hence the overall right derivative in (\ref{equ:dyn_gen_deriv}) is a restricted version of the (bidirectional) derivative in the case of the corresponding linear dynamics where $\partial f_{j,r}(\cdot) = 1,\ \forall v_j\in V, r>0$, and 
\begin{align}
	\frac{1}{(1-\gamma)^t}\partial \mathbf{s}_t^T = \mathbf{1}^T\left(\mathbf{W}^T\right)^t.
	\label{equ:dyn_gen_deriv-lin}
\end{align}

Therefore, the right derivative can be useful if the propagation is dominated by the linear part, but will almost always be $0$ if the lower bounds are close to the upper bounds. Suppose at a particular time step $t'$, $l_{j, t'} = h_{j, t'} = l_{t'},\ \forall v_j\in V$, then $f_{j, t'}(x) = l_{t'} H(x-l_{t'})$, where $H(\cdot)$ is the Heaviside step function, and
\begin{align}
\frac{df_{j, t'}(x)}{dx} = l_{t'}\frac{d H(x-l_{t'})}{dx} = l_{t'} \delta(x-l_{t'}),
\end{align} 
where $\delta(x)$ is the \textit{Dirac delta function} with $\delta(x) = +\infty$ if $x = 0$ and $0$ otherwise. Accordingly, the derivative $\partial s_t(x_j(0))$ for each node $v_j$ at any time steps $t \ge t'$ can only be $0$ or $+\infty$, which is not very informative. We conclude here that the derivative information generally has limited use in understanding the change of the overall influence, and accordingly in the IM problem.

\section{Further features of the influence maximization} 
In this section, we discuss more details of the IM problem. We first explore the constraint of limited budget size in Sec.~\ref{app:saturation}. Then we give more details of the general MADS approach in Sec.~\ref{app:MADS} and the experiments exploring the performance of the proposed CDS method in Sec.~\ref{app:CDS-performance}. Finally, we discuss the time complexity of the proposed method in Sec.~\ref{sec:appendix_CDS_timecomplex}. As in Sec.~\ref{sec:experiments_results}, $l_{j,0} = h_{j,0} = 1,\ \forall v_j\in V$, $\gamma = 0$, and we apply exclusively the threshold-type bounds in (\ref{equ:dyn_gen_thres-bBt}) with condition (\ref{equ:dyn_gen_thres-uni}).  

\subsection{\label{app:saturation} Budget size}
We start from exploring the dependence of the optimal objective value on the budget size $k$, through small networks, because it is practically hard for global algorithms like the brute-force to obtain an optimal solution in large networks. 

Specifically, we consider a network of $n=20$ nodes, generated from a two-block SBM with each edge being bidirectional and of connecting probability $p_1 = 0.5$ in one community, $p_2 = 0.25$ in the other, and $p_{12} = 0.05$ between the two communities. The probabilities in the two communities are set to be different in order to separate the nodes in each community. We again assign uniform weight $\alpha=0.1$. We compare optimal objective values $s^*$ obtained subject to different budget sizes through the \textit{relative optimal objective value}, $s^*/s^*_{max}$, where $s^*_{max}$ is the optimal value from $k=n=20$, i.e.,~the maximum objective value w.r.t.~all possible $k$.

\begin{figure*}
	\centering
	\begin{tabular}{cc}
		\includegraphics[width=.4\textwidth]{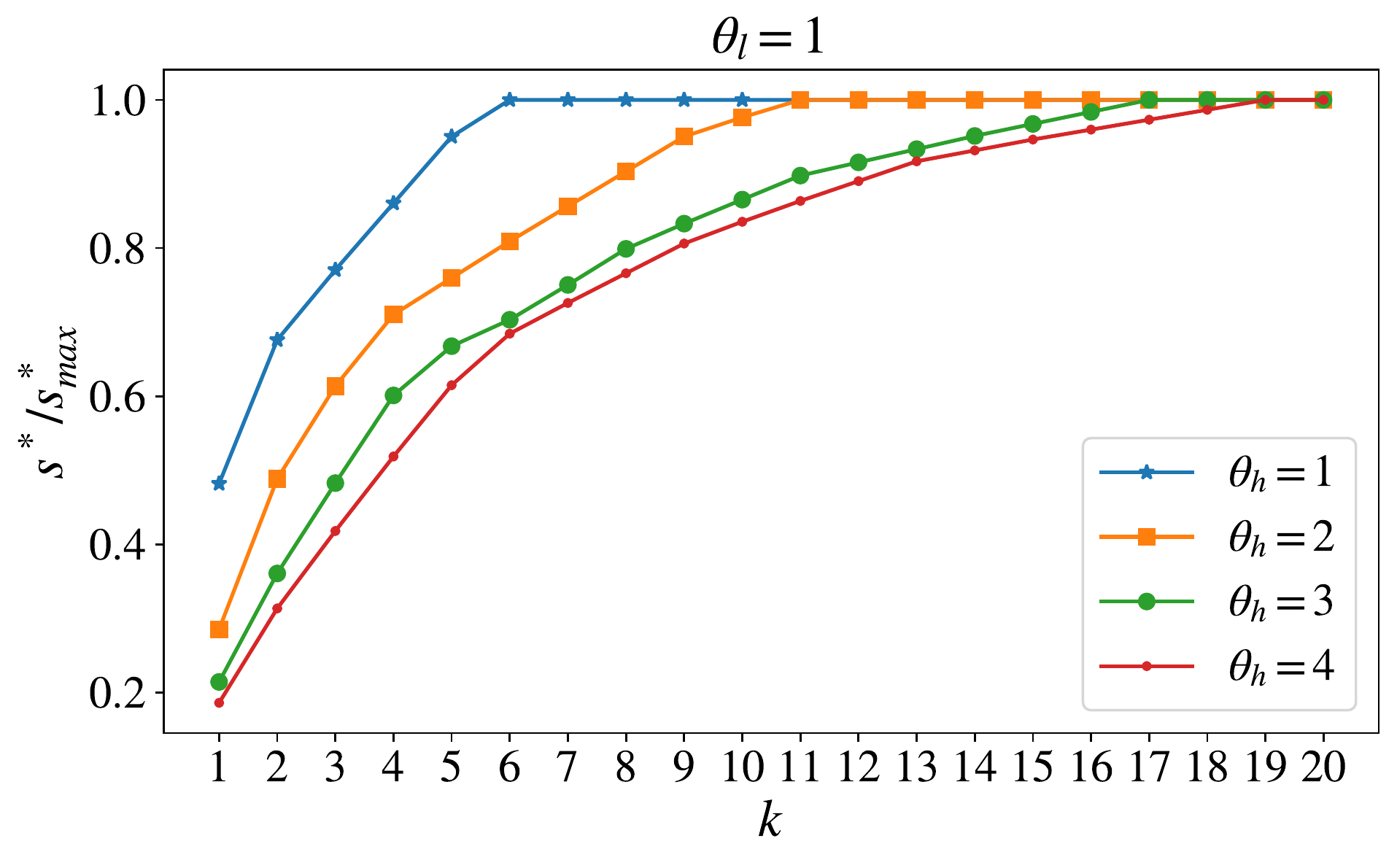} & 
		\includegraphics[width=.4\textwidth]{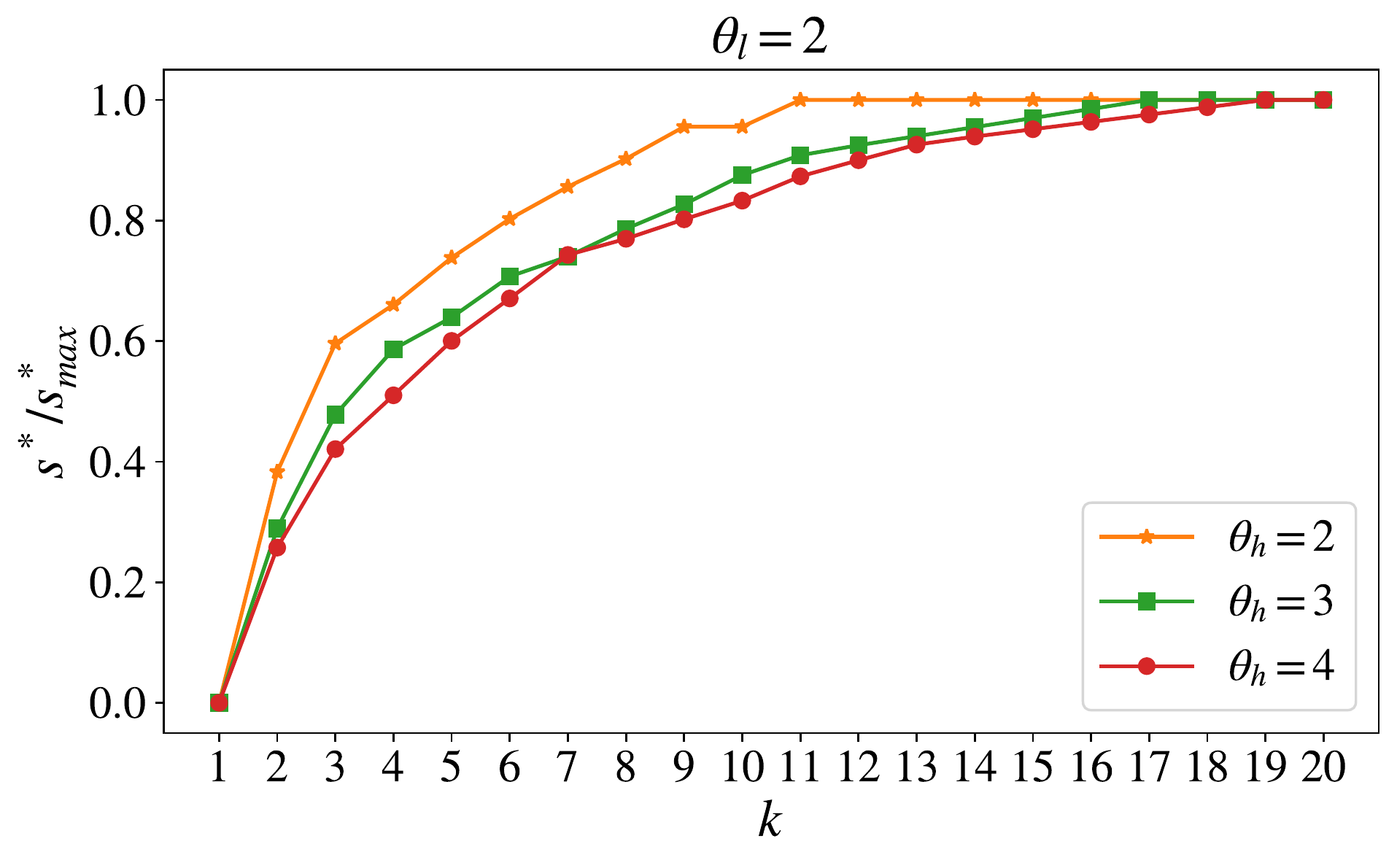}
	\end{tabular}
	\caption{The change of relative optimal objective value $s^*/s^*_{max}$ with respect to the budget size $k$ ($x$-axis) on the SBM, when the GIP model has varying upper bound threshold $\theta_h$ while the lower bound threshold $\theta_l = 1$ (left) and $\theta_l = 2$ (right).}
	\label{fig:IM_saturate_ab}
\end{figure*}
\begin{figure*}
	\centering
	\begin{tabular}{cc}
		\includegraphics[width=.4\textwidth]{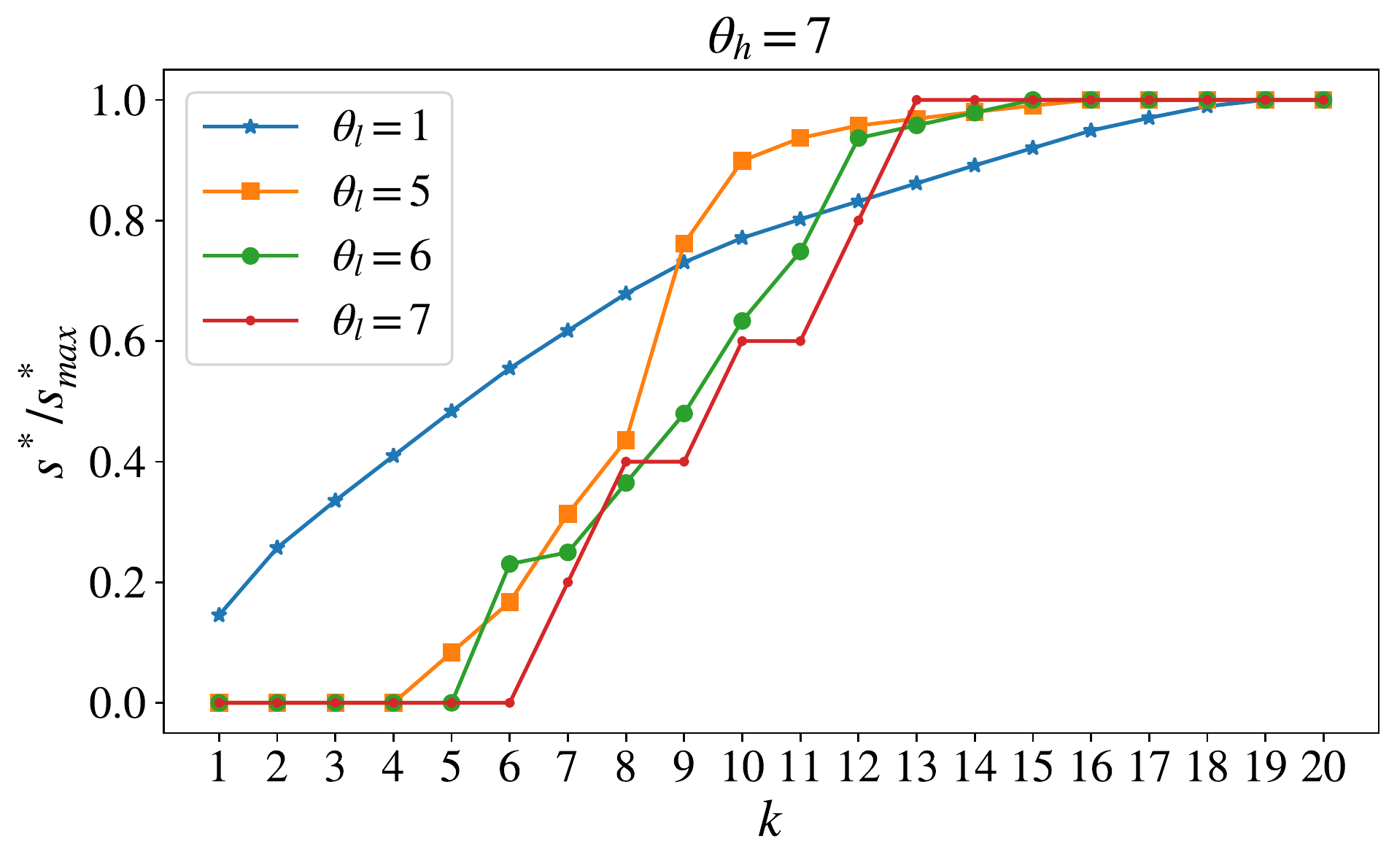} & 
		\includegraphics[width=.4\textwidth]{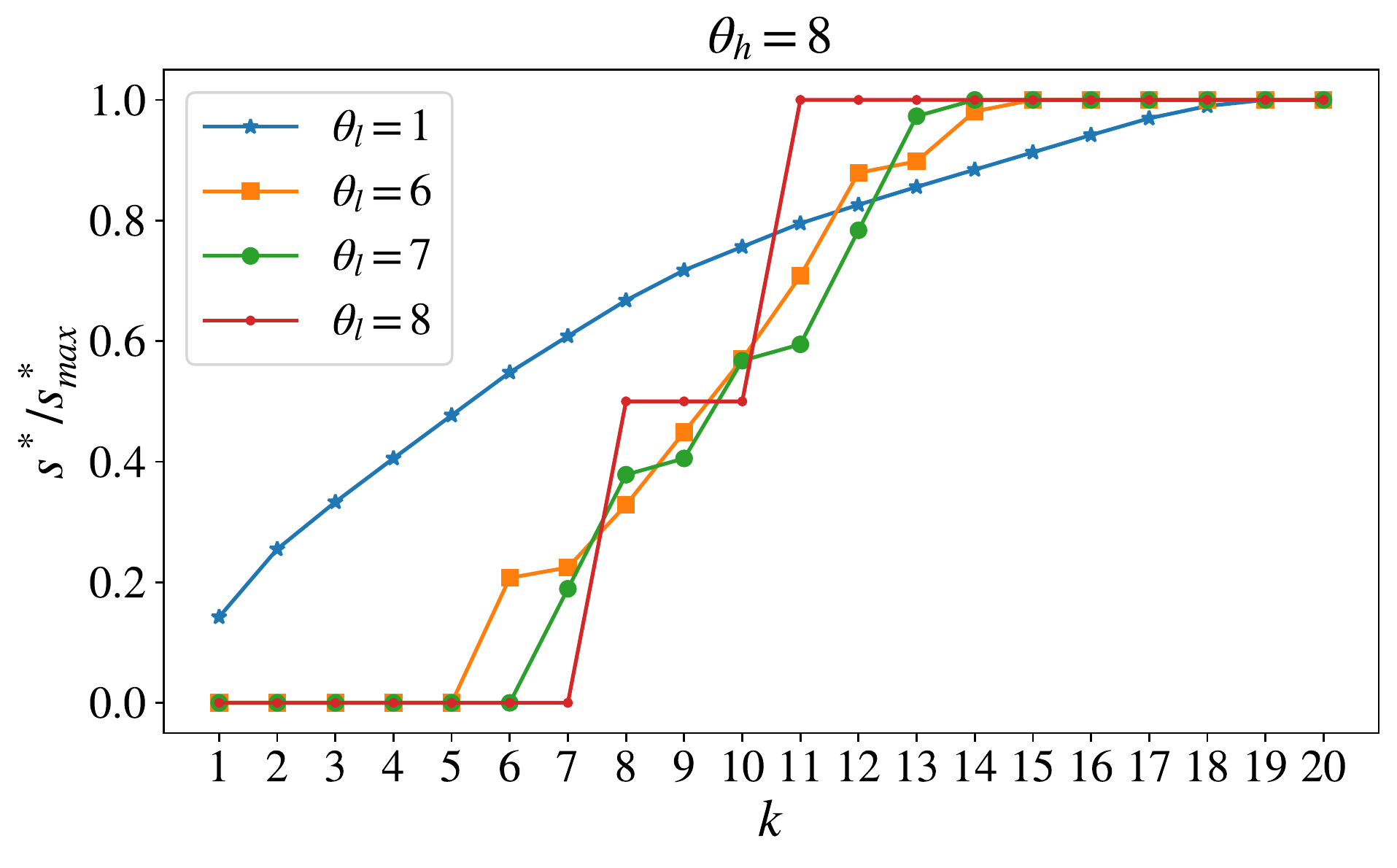}
	\end{tabular}
	\caption{The change of relative optimal objective value $s^*/s^*_{max}$ with respect to the budget size $k$ ($x$-axis) on the SBM, when the GIP model has varying lower bound threshold $\theta_l$ while the upper bound threshold $\theta_h = 7$ (left) and $\theta_h = 8$ (right).}
	\label{fig:IM_saturate_aB}
\end{figure*}
We observe that the optimal objective value reaches its maximum level at a smaller budget size $k$, as the upper bounds decrease; see Fig.~\ref{fig:IM_saturate_ab}. When $\theta_l = \theta_h = 2$, only activating $k = 11$ nodes initially can achieve the maximum level of influence on the network, and for $\theta_l = \theta_h = 1$, only $k = 6$ nodes are needed. This property of saturation also illustrates the rationality for the IM problem, where one aims to influence a large portion of the network from a small set of initially activated nodes: it is not only because of limited resources, but also because activating more nodes does not necessarily benefit the objective substantially.  

Interestingly, when enlarging the lower bounds, the early saturation characteristics is not significant, but the step effect is increasingly explicit; see Fig.~\ref{fig:IM_saturate_aB}. When $\theta_l = \theta_h$, the optimal objective function changes as a step function w.r.t.~the budget size $k$, and when $\theta_l=1$, it is closer to being linear. When $\theta_l$ lies in between these two extreme cases, the optimal objective function interpolates these two shapes, with both linear-like increase and step effect, consistently with the GIP model capturing features both from the EIC model and from the ELT model.

\subsection{\label{app:MADS} MADS method}
We now discuss in more detail the general solution method for the MINLP \eqref{equ:opt_dyn_gen}: the MADS method. The CDS method we have proposed in Sec.~\ref{sec:IM_customized_direct_search} is largely based on it.  

In the MADS for mixed variables (MV), each vector $\mathbf{y}$ is partitioned into its continuous and discrete components, $\mathbf{y}^c$ and $\mathbf{y}^d$, respectively. For the MINLP \eqref{equ:opt_dyn_gen}, $\mathbf{y}^c = \mathbf{x}$ and $\mathbf{y}^d = \mathbf{z}$. We denote the maximum dimensions of the continuous and discrete variables by $n^c$ and $n^d$, respectively, thus $\mathbf{y}^c\in\Omega^c\subseteq \mathbb{R}^{n^c}$ and $\mathbf{y}^d\in\Omega^d\subseteq \mathbb{Z}^{n^d}$. The general problem under consideration is a minimisation problem,
\begin{align*}
	\min_{\mathbf{y}\in\Omega} f(\mathbf{y}),
\end{align*} 
where $f: \Omega\to \mathbb{R}\cup \{\infty\}$, and the domain, or the feasible region, is the union of continuous domain across possible discrete variable values,
\begin{align*}
    \Omega = \bigcup_{\mathbf{y}^d\in\Omega^d}(\Omega^c(\mathbf{y}^d)\times \{\mathbf{y}^d\}),
\end{align*}
where $\Omega^c(\mathbf{y}^d)$ indicates that the continuous domain can change with different discrete variable values, and $\Omega=\Omega^c$ if $n^d = 0$. Hence, we need to modify the objective of the MINLP \eqref{equ:opt_dyn_gen} as $\min_{\mathbf{x},\mathbf{z}}-s(\mathbf{x})$ when applying the algorithm, i.e.,~to set $f = -s$. The constraints are incorporated in the domain, and are treated by the extreme barrier approach $f_\Omega$, where $f_\Omega(\mathbf{y}) = f(\mathbf{y})$ if $\mathbf{y}\in\Omega$ and $\infty$ otherwise. 


The MADS algorithm is characterized by an optional \texttt{search} step, a local \texttt{poll} step, and an \texttt{extended poll} step, where the objective $f_\Omega$ is evaluated at specific points defined on an underlying mesh $M_r$ at each iteration $r$. The goal of each iteration is to find a feasible improved mesh point from the current iterate, $\mathbf{y}\in M_r$ s.t.~$f_\Omega(\mathbf{y}) < f_\Omega(\mathbf{y}^{(r)})$, and the algorithm will output the point it converges to. The mesh $M_r$ at each iteration $r$ is a central concept in this method, which is formed as the direct product of $\Omega^d$ with the union of a finite number of lattices in $\Omega^c$, 
\begin{align}
    M_r = \bigcup_{q=1}^{q_{max}}M_r^q\times \Omega^d,
    \label{equ:MADS_mesh}
\end{align}    
where $q = 1, ... q_{max}$ indicates each combination of discrete variable values, and the lattice $M_r^q$ is defined through previously evaluated points, the positive spanning directions and the mesh size parameter $\Delta_r^m$ which dictates its coarseness. In the \texttt{poll} step, the method evaluates the discrete neighborhood $\mathcal{N}(\mathbf{y}^{(r)})$, and the points whose continuous parts are close to the current iterate along certain directions, $P_r(\mathbf{y}^{(r)})$, controlled by $\Delta_r^m$ and the poll size parameter $\Delta_r^p$. The \texttt{extended poll} step is triggered when the \texttt{poll} step fails to find an improved point. It consists of a finite sequence of \texttt{poll} steps performed around the points in $\mathcal{N}(\mathbf{y}^{(r)})$ whose objective values are sufficiently close to the incumbent value, $f_\Omega(\mathbf{y}^{(r)}) \le f_\Omega(\mathbf{y}) \le f_\Omega(\mathbf{y}^{(r)}) + \xi_r$, for some user-defined tolerance $\xi_r \ge \xi$ [e.g.,,~$\xi_r = \max\{\xi, 0.05\abs{f(\mathbf{y}^{(r)})}\}$], and we denote such set of nodes by $\mathcal{N}_r^{\xi_r}(\mathbf{y}^{(r)})$. We summarize the main ideas of the MADS in Algorithm \ref{alg:mv-mads}. 

\subsection{\label{app:CDS-performance} The CDS method}
Here, we provide the further details from analyzing the performance of the CDS method on both real and synthetic networks of various structures, to complement the results in Sec.~\ref{sec:experiments_accuracy}. 

\paragraph{Karate club network.}
The karate club network is a social network of a university karate club \cite{Zachary_karate_1977}. It captures $n=34$ members of the club, and has $|E| = 78$ edges indicating pairs of members who interact outside the club, thus they are bidirectional. This real network is comprehensively used in network analysis for various purposes, while in a manageable small scale, and thus is suitable here.

When $k=3$, the CDS method can successfully find an optimal solution in all different choices of the upper and lower bounds; see Fig.~\ref{fig:IM_karate_3_acc-rank}. 
Furthermore, the time consumed by the CDS method is always less than $5\%$ of the brute force's, around $0.05$s versus $1$s, thus the CDS method is also much more efficient.
\begin{figure*}
	\centering
	\begin{tabular}{cc}
		\includegraphics[width=.48\textwidth]{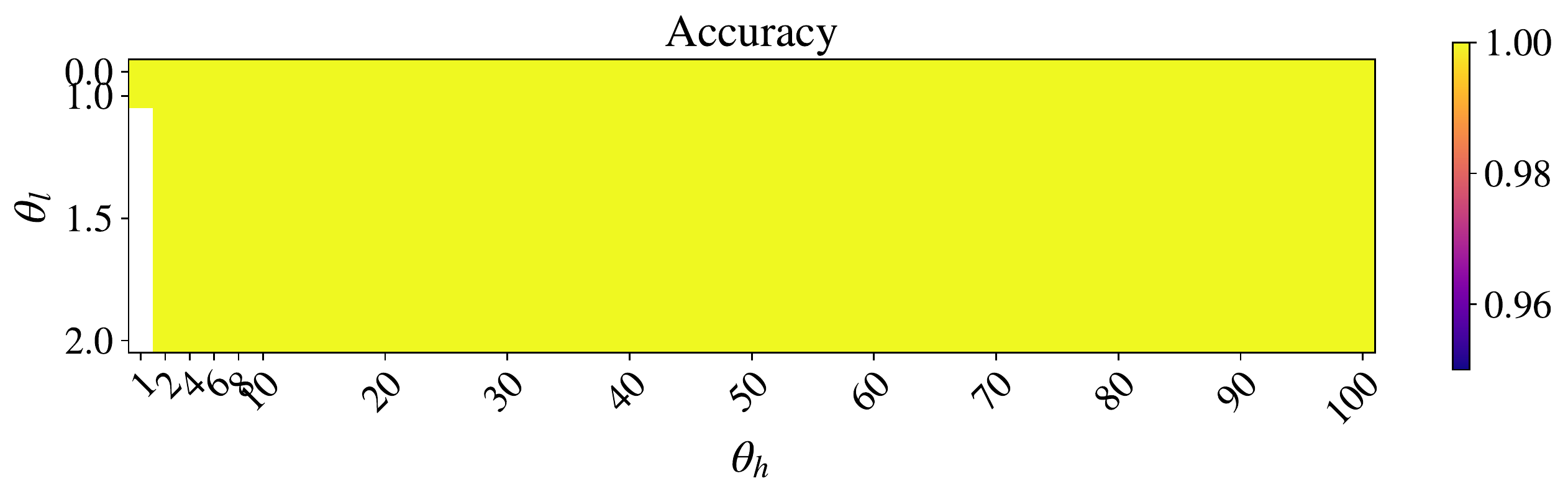} & 
		\includegraphics[width=.48\textwidth]{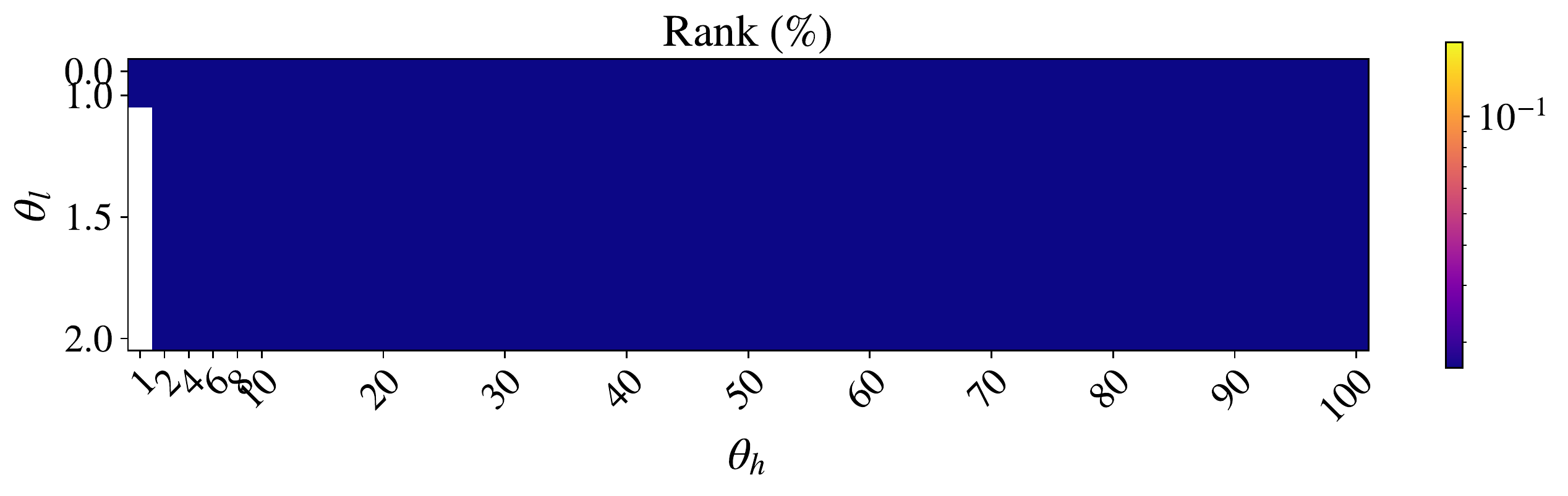}
	\end{tabular}
	\caption{Performance of the CDS method on the karate club network in terms of the output's accuracy (left) and rank (\%, right), subject to changing upper ($x$-axis) and lower ($y$-axis) bound thresholds of the GIP model, $\theta_h$ and $\theta_l$, respectively, when $k=3$.}
	\label{fig:IM_karate_3_acc-rank}
\end{figure*}

We then change the budget size $k$ in two different cases of bounds: (i) $\theta_l = 2,\ \theta_h = 16$, and (ii) $\theta_l = 2 = \theta_h$, corresponding to the extreme of the ELT model. We observe that the CDS method can always find a global optimal solution in case (i), while the performance drops slightly as $k$ increases but is still close to optimal in case (ii); see Fig.~\ref{fig:IM_karate_aus_Ks_acc}. Moreover, the time consumption of the CDS method increases approximately linearly as $k$ rises, while for the brute force, it changes exponentially, which makes it practically hard to obtain a global optimum under larger budget sizes.
\begin{figure*}
	\centering
	\begin{tabular}{cc}
		\includegraphics[width=.35\textwidth]{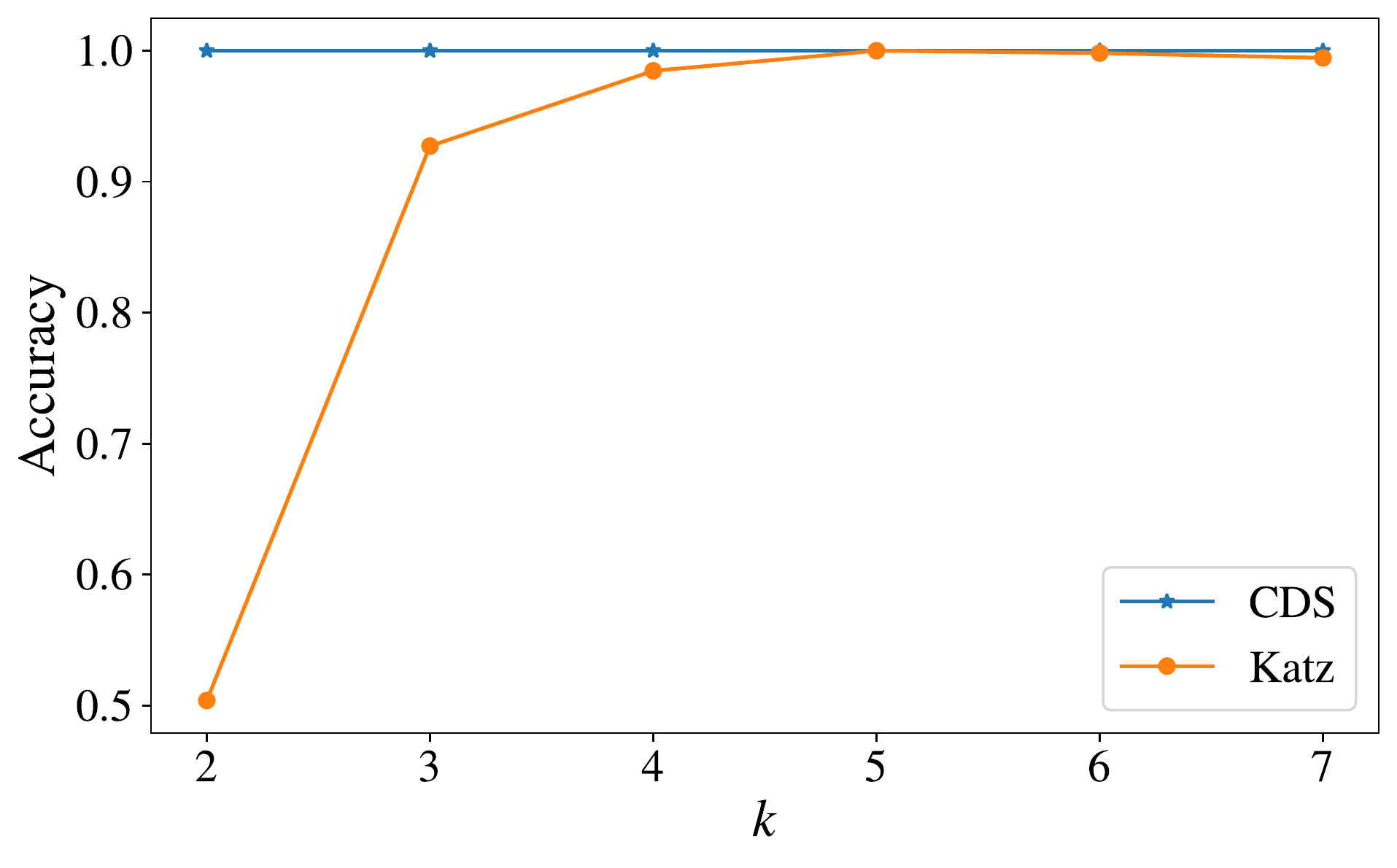} & 
		\includegraphics[width=.35\textwidth]{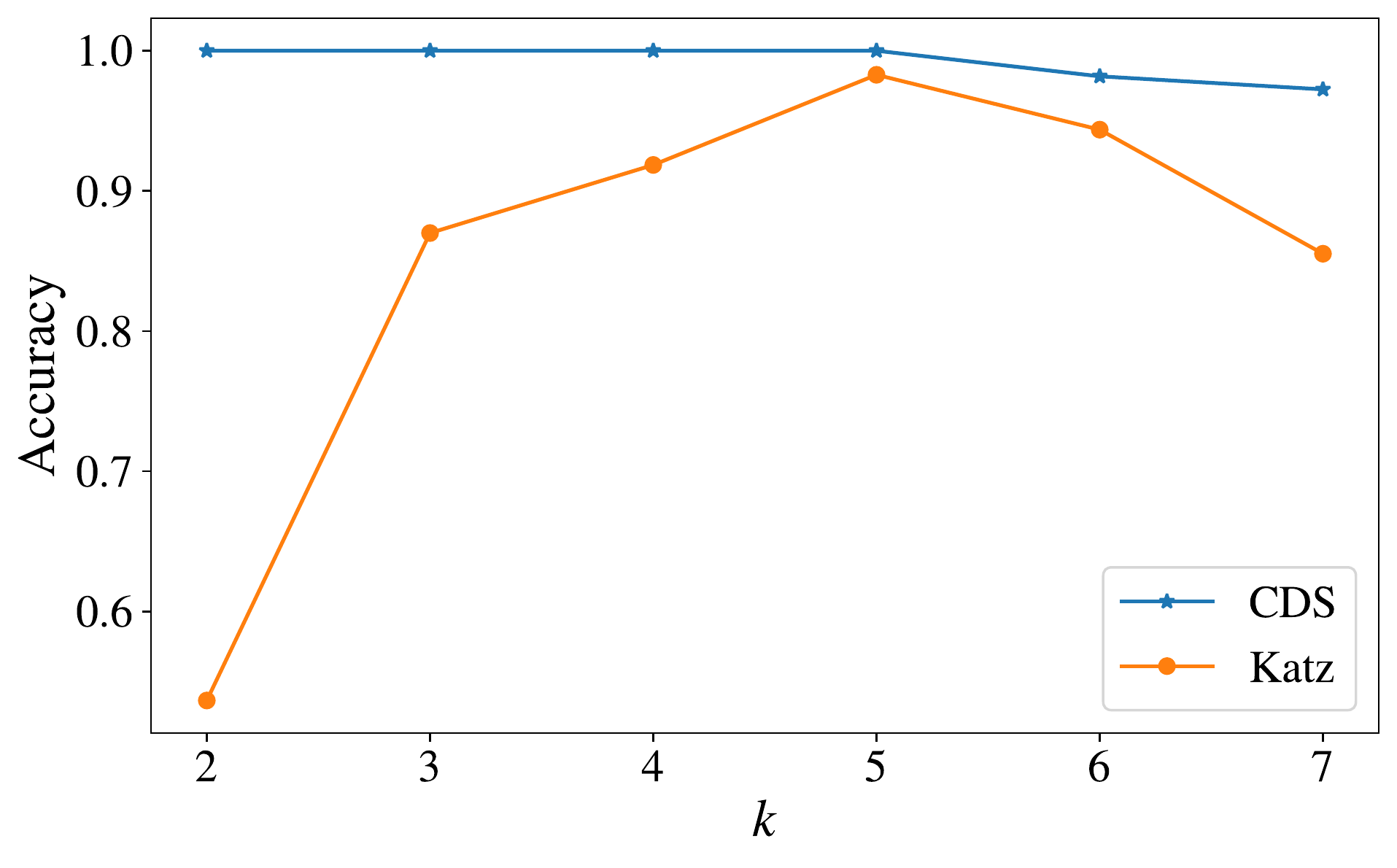}
	\end{tabular}
	\caption{Performance of the CDS method on the karate club network in terms of the output's accuracy, subject to changing budget size $k$ ($x$-axis) when $\theta_l = 2, \theta_h = 16$ (left) and $\theta_l = \theta_h = 2$ (right).}
	\label{fig:IM_karate_aus_Ks_acc}
\end{figure*}

\begin{algorithm}[H]
	\caption{Mesh adaptive direct search for mixed variables (MADS-MV).}
	\label{alg:mv-mads}
	\begin{algorithmic}[1] 
		\State{initialization: Set $\xi > 0$ and $\xi_0 \ge \xi$. Let $\mathbf{y}^{(0)}\in \Omega$ such that $f_{\Omega}(\mathbf{y}^{(0)}) < \infty$, set $\Delta_0^p \ge \Delta_0^m > 0$. Set iteration $r = 0$.} 
		\State{\texttt{SEARCH} step (optional): Evaluate $f_{\Omega}$ on a finite set of trial points on the mesh $M_r$ (\ref{equ:MADS_mesh}). If an improved mesh point is found, the \texttt{SEARCH} step may terminate, skip the next \texttt{POLL} step and go directly to step 5.} 
		\State{\texttt{POLL} step: Evaluate $f_\Omega$ on the set $P_r(\mathbf{y}^{(r)})\cup \mathcal{N}(\mathbf{y}^{(r)}) \subset M_r$ (i.e.,~close to the current iterate), until an improved mesh point is found, or until all points have been exhausted. If an improved mesh point is found, go to step 5.} 
		\State{\texttt{EXTENDED POLL} step: Perform a finite sequence of \texttt{poll} step starting from each point $\mathbf{y}\in\mathcal{N}_r^{\xi_r}(\mathbf{y}^{(r)}) \subseteq\mathcal{N}(\mathbf{y}^{(r)})$ with $f_\Omega(\mathbf{y}^{(r)}) \le f_\Omega(\mathbf{y}) \le f_\Omega(\mathbf{y}^{(r)}) + \xi_r$, until an improved mesh point is found or until all points have been exhausted.}
		\State{Parameter update: Coarsen $\Delta_{r+1}^m$ and $\Delta_{r+1}^p$ when an improved mesh point is found and refine them otherwise. Update $\xi_r\ge\xi$, increment $r\leftarrow r+1$, and go to step 2.}  
	\end{algorithmic}
\end{algorithm}

\paragraph{Composite network.}
The composite network here particularly refers to the one in Fig.~\ref{fig:Diff_smallworldrand_net}, constructed by connecting a regular lattice of size $n_o = 25$ and mean degree $d_o = 4$, and a ER random graph of the same size $n_o$ and probability $p_{er} = d_o/n$, with edges randomly placed between the two parts with a small probability $p_o = 0.01$. As discussed in Appendix \ref{app:experiment_coexist}, there is noticeable difference in the propagation on this network when the GIP model is at the extreme of the ELT model versus in a general case. It is then interesting to explore the performance of the CDS method for the corresponding IM task, on this particular structure.

When $k=4$, the CDS method can again find an optimal or close-to-optimal solution with different choices of the upper and lower bound thresholds; see Fig.~\ref{fig:IM_smallworld_4_acc-com}. There are only $4$ cases where the CDS method cannot output a globally optimal solution, and the worst-case scenario occurs when $\theta_l = 1.4$ and $\theta_h = 2$, where the accuracy is about $0.85$ and the rank is slightly below $0.009\%$ in overall $230\,300$ possibly initially activated sets, i.e.,~the CDS method can still output a top $20$ set in this case. 
\begin{figure*}
	\centering
	\begin{tabular}{cc}
		\includegraphics[width=.48\textwidth]{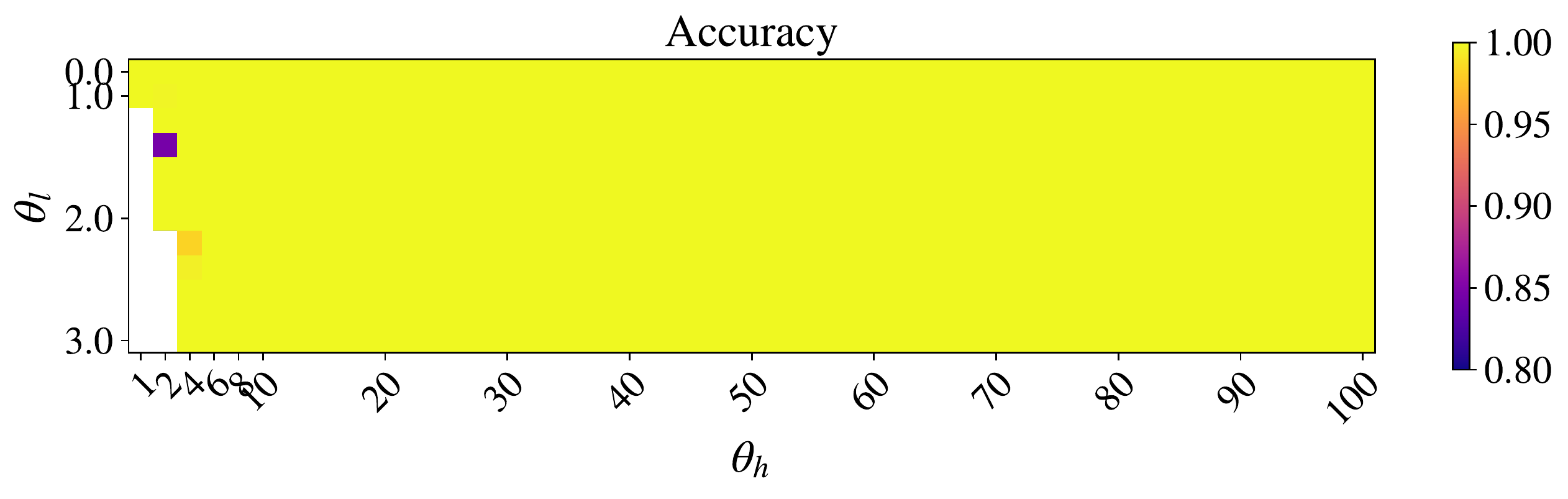} & 
		\includegraphics[width=.48\textwidth]{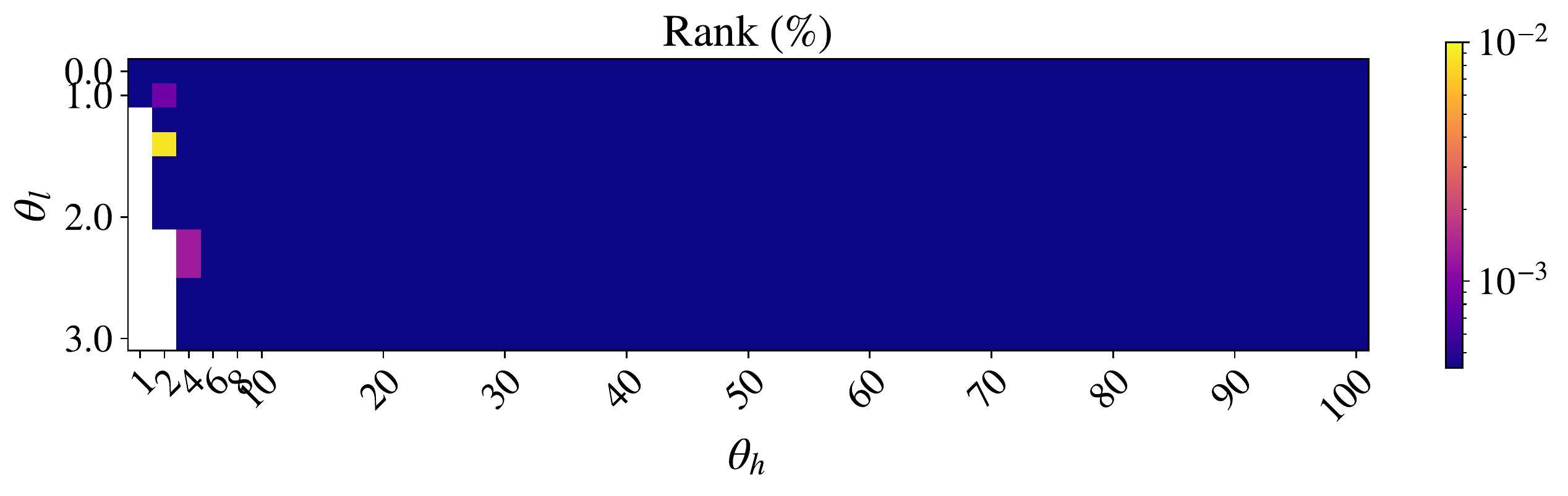}\\
		\includegraphics[width=.48\textwidth]{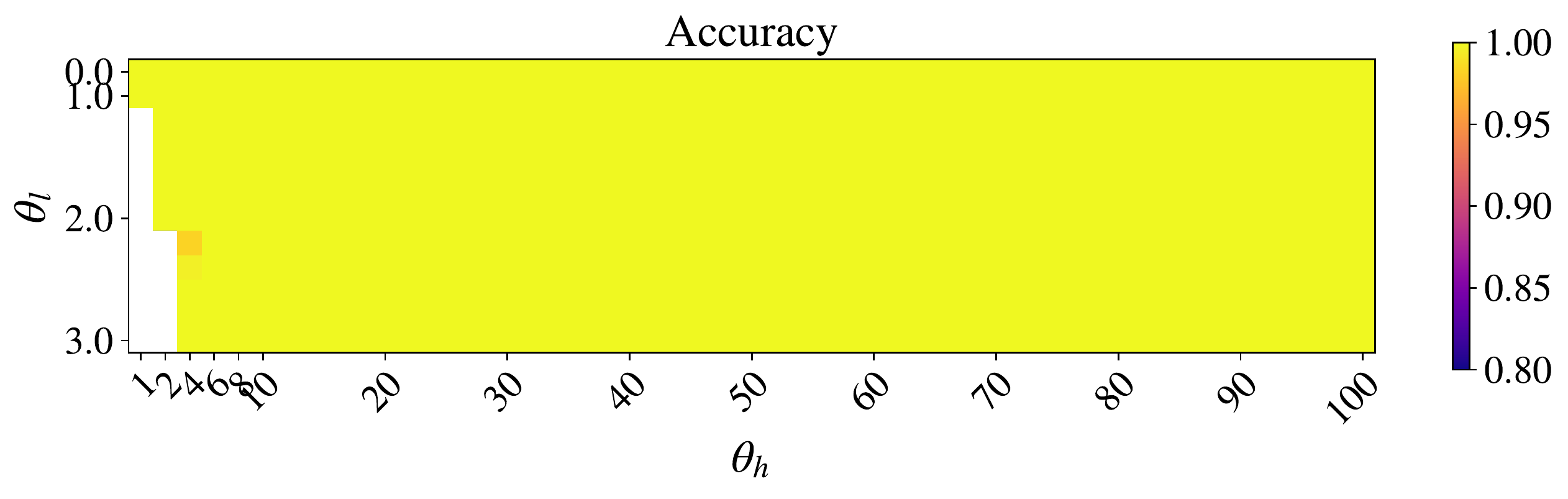} & 
		\includegraphics[width=.48\textwidth]{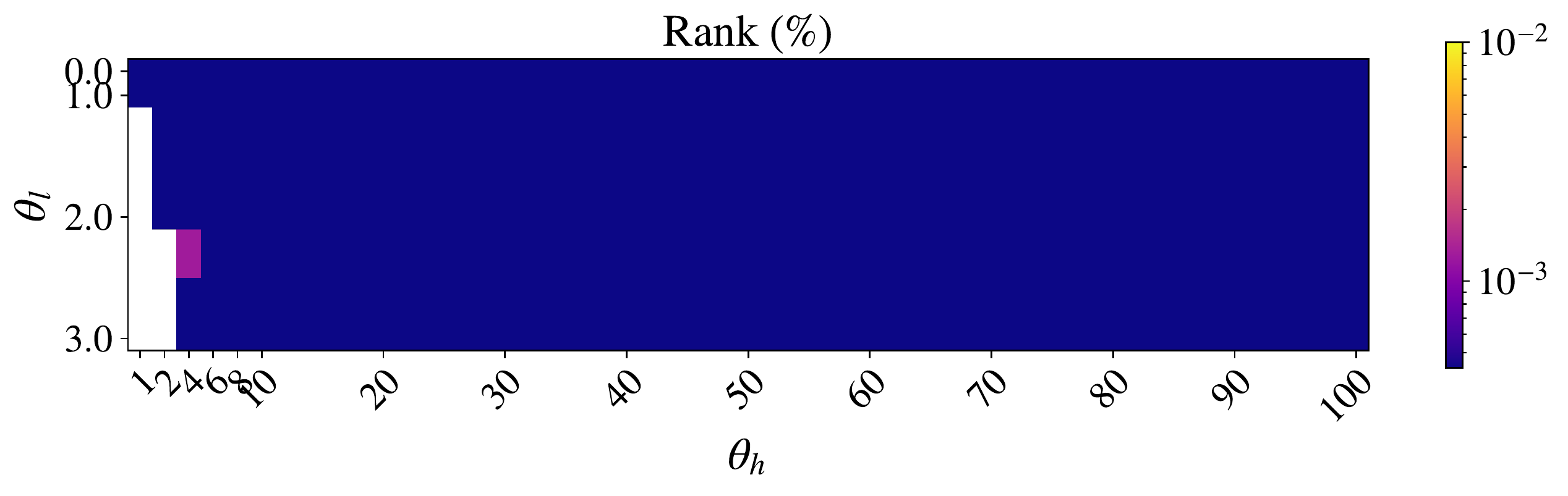}
	\end{tabular}
	\caption{Performance of the CDS method (top) and that with community restart (bottom) on the composite network in terms of the output's accuracy (left) and rank ($\%$, right), subject to changing upper ($x$-axis) and lower ($y$-axis) bound thresholds of the GIP model, $\theta_h$ and $\theta_l$, respectively, when $k=4$.}
	\label{fig:IM_smallworld_4_acc-com}
\end{figure*}

Since there are certain cases where a global optimum cannot be reached by the (plain) CDS method, we explore one improvement strategy here: to restart the search process, i.e.,~steps 2, 3, and 4 in Algorithm \ref{alg:mv-mads_custom}, from other unexplored points. We achieve this by its noticeable community structure. Specifically, we propose the following \textit{community restart strategy}: (i) construct a set containing different splits of $k$ into the two communities, $\mathcal{S}\coloneqq \{(k_1, k_2): k_1 + k_2 = k, k_1, k_2\in\mathbb{N}\}$; (ii) construct the set of initial points corresponding to activating $k_1$ and $k_2$ nodes of the highest values of $h_{j,0}c_j$ where $c_j$ is the Katz centrality of node $v_j$, in communities 1 and 2, respectively, $\forall (k_1, k_2)\in \mathcal{S}$; (iii) restart the search process from each point in (ii) if it is has not been explored yet. We observe that the community restart strategy can assist the CDS method to find a global optimum when the lower bound threshold $\theta_l < 2$; see Fig.~\ref{fig:IM_SBM_4_acc-com}. Now, the worse-case scenario has accuracy $0.98$ and rank $0.001\%$, i.e.,~it can now output a top $2$ set. We refer the reader to \cite{Schoenebeck_GlobalLocal_2019} for more theoretical results on the interplay between the community structure and complex contagions.

\begin{figure*}
	\centering
	\begin{tabular}{cc}
		\includegraphics[width=.35\textwidth]{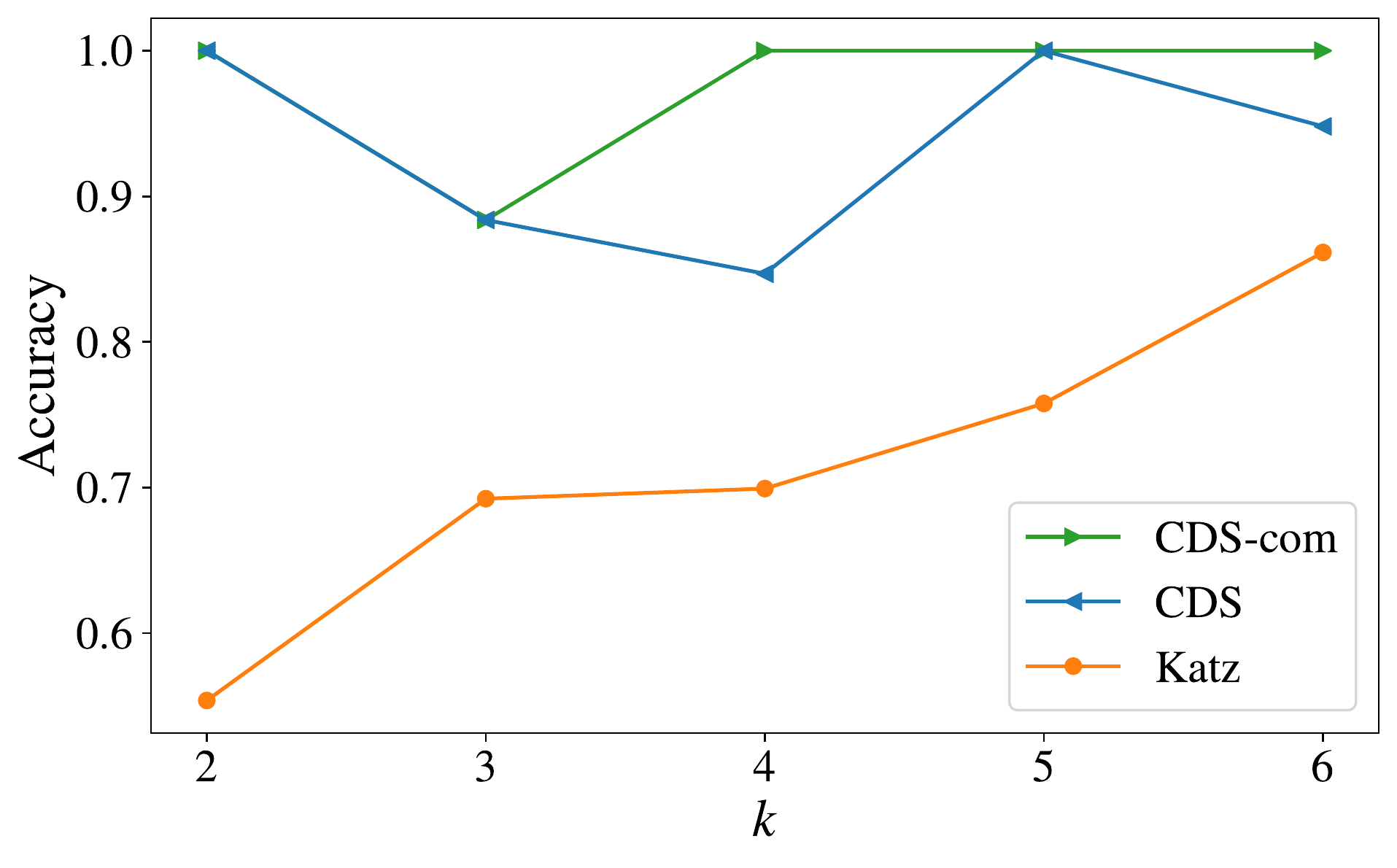} & 
		\includegraphics[width=.35\textwidth]{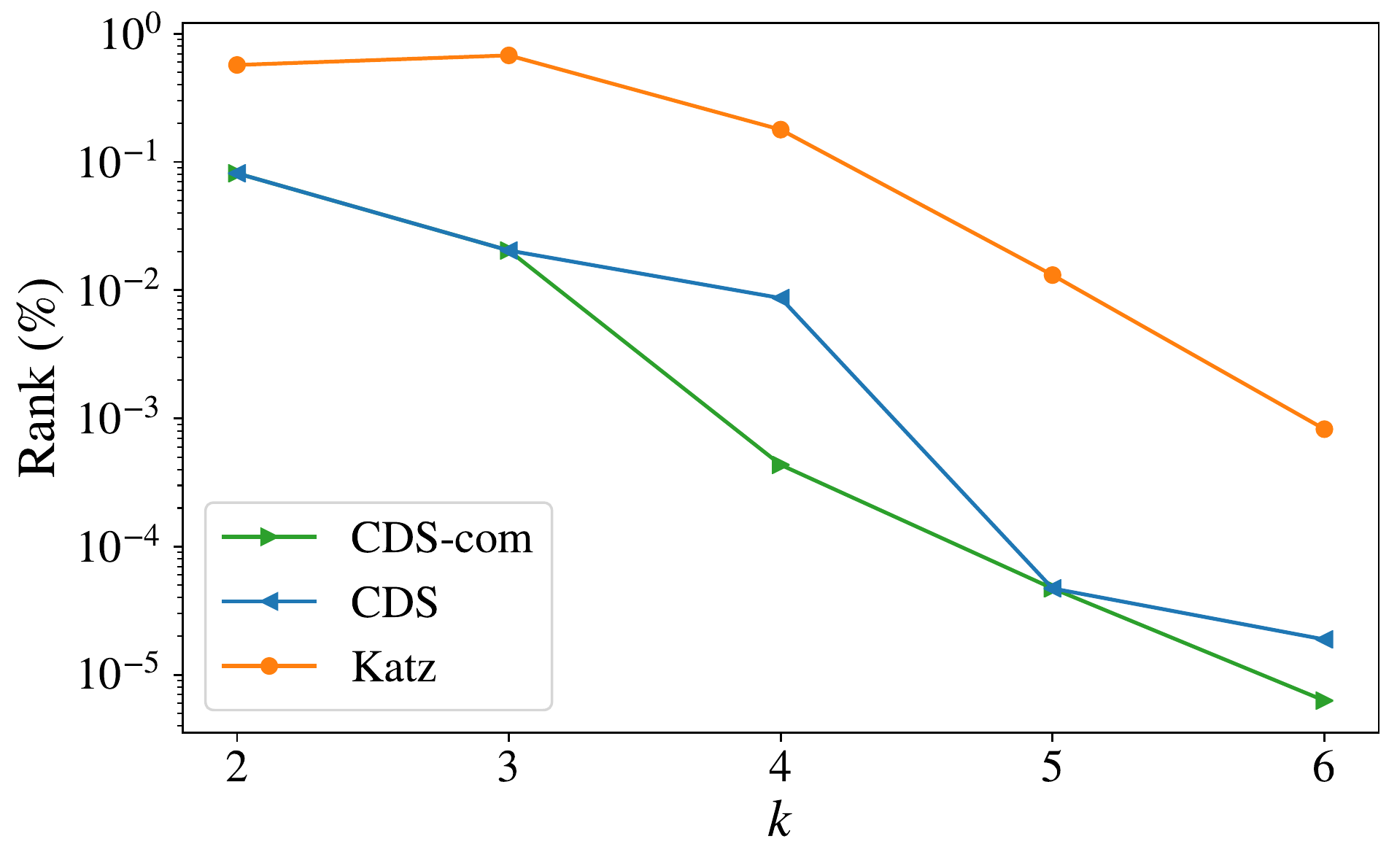}
	\end{tabular}
	\caption{Performance of the CDS method (``CDS") and that with community restart (``CDS-com") on the composite network in terms of the output's accuracy (left) and rank (right), subject to changing budget size $k$ (x-axis) when $\theta_l = 1.4, \theta_h = 2$.}
	\label{fig:IM_smallworld_Ks}
\end{figure*}
We then explore the performance of the CDS method, together with the community restart strategy, when varying the budget size $k$ and $\theta_l = 1.4, \theta_h = 2$, the pair with worst-case performance of the (plain) CDS method when $k = 4$. The performance of the CDS method is generally good, with accuracy greater than $0.8$ and consistently low rank, and the performance can be further improved by the community restart strategy, particularly when $k$ is large; see Fig.~\ref{fig:IM_smallworld_Ks}. The results also verifies the valuable information in the community structure.

\subsection{\label{sec:appendix_CDS_timecomplex} Time complexity of the CDS method}
Finally, we discuss the time complexity of the CDS method that we have proposed in Sec.~\ref{sec:IM_customized_direct_search}, through decomposing it into two parts: (I) the evaluation of the objective function, and (II) the number of evaluations required until convergence. From Theorem \ref{the:evaluation_complexity}, we know that given a network, the worst-case scenario of (I) occurs when the GIP model reaches the extreme of the EIC model, since it corresponds to an achievable upper bound for $t_\epsilon$. Hence, the CDS method circumvents the worst-case complexity by starting from the solution associated with the linear dynamics. For (II), we further decompose it into the product of (i) the size of discrete neighborhood that is feasible, and roughly, (iia) the number of steps towards convergence, or more precisely, (iib) the number of evaluations divided by the neighborhood size. (i) is $k(n - k)$, by Eqs.~(\ref{equ:opt_spec_mesh}) and (\ref{equ:opt_spec_neigh}), and is equivalently $O(n)$ since $k$ is normally assumed to be $O(1)$. The only remaining part for a full description of the time complexity is (iia) or (iib), both are related to the rate of convergence which is difficult to provide a theoretical guarantee given the general properties of the objective function of the MINLP \eqref{equ:opt_dyn_gen}. Instead, we conjecture that the complexity of (iib) does not increase significantly when varying the network size $n$, i.e.,~approximately $O(1)$, given the same mean degree of nodes, and verify it empirically through experiments on ER random graphs \imp{with each edge being bidirectional}. 

Specifically, we construct ER random graphs of the same expected node degree, $\langle d\rangle = (n-1)p$ where $n$ is the network size and $p$ is the connecting probability in ER random graphs, but of increasing size $n$. We then apply the CDS method to the IM problem on this series of networks to explore the dependence of its time complexity on the network size $n$. The experiments are performed with different combinations of parameters in order to take into account their effects, including the upper bounds, the lower bounds, and the budget size $k$. Here, we assign a uniform edge weight $\alpha = 0.1$ to all the networks, set $l_{j,0} = h_{j,0} = 1,\ \forall v_j\in V$, $\gamma = 0$, and apply exclusively the threshold-type bounds in (\ref{equ:dyn_gen_thres-bBt}) with condition (\ref{equ:dyn_gen_thres-uni}), as in Sec.~\ref{sec:experiments_results}.

\begin{figure*}
	\centering 
	\begin{tabular}{cc}
		\includegraphics[width=.35\textwidth]{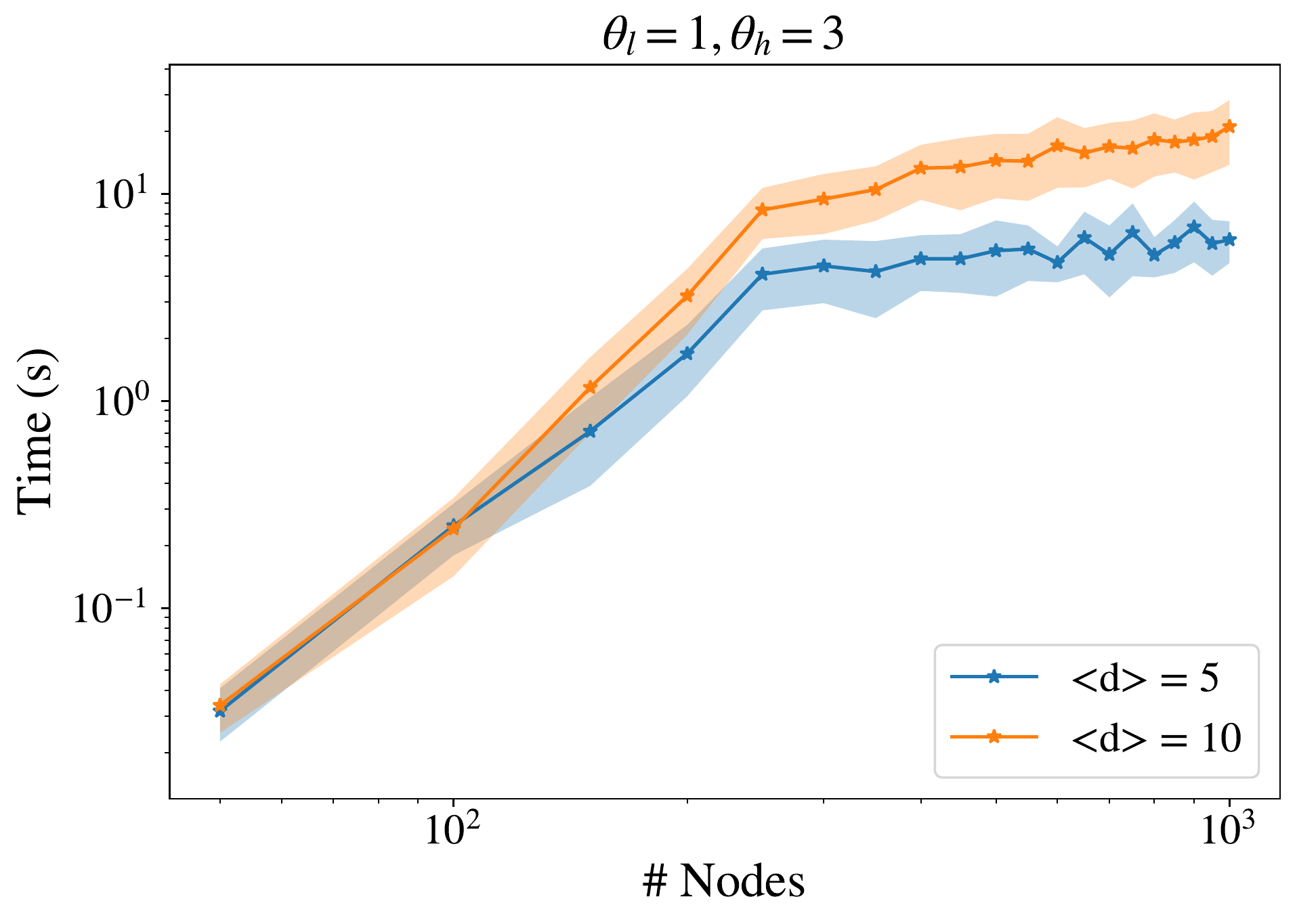} & 
		\includegraphics[width=.35\textwidth]{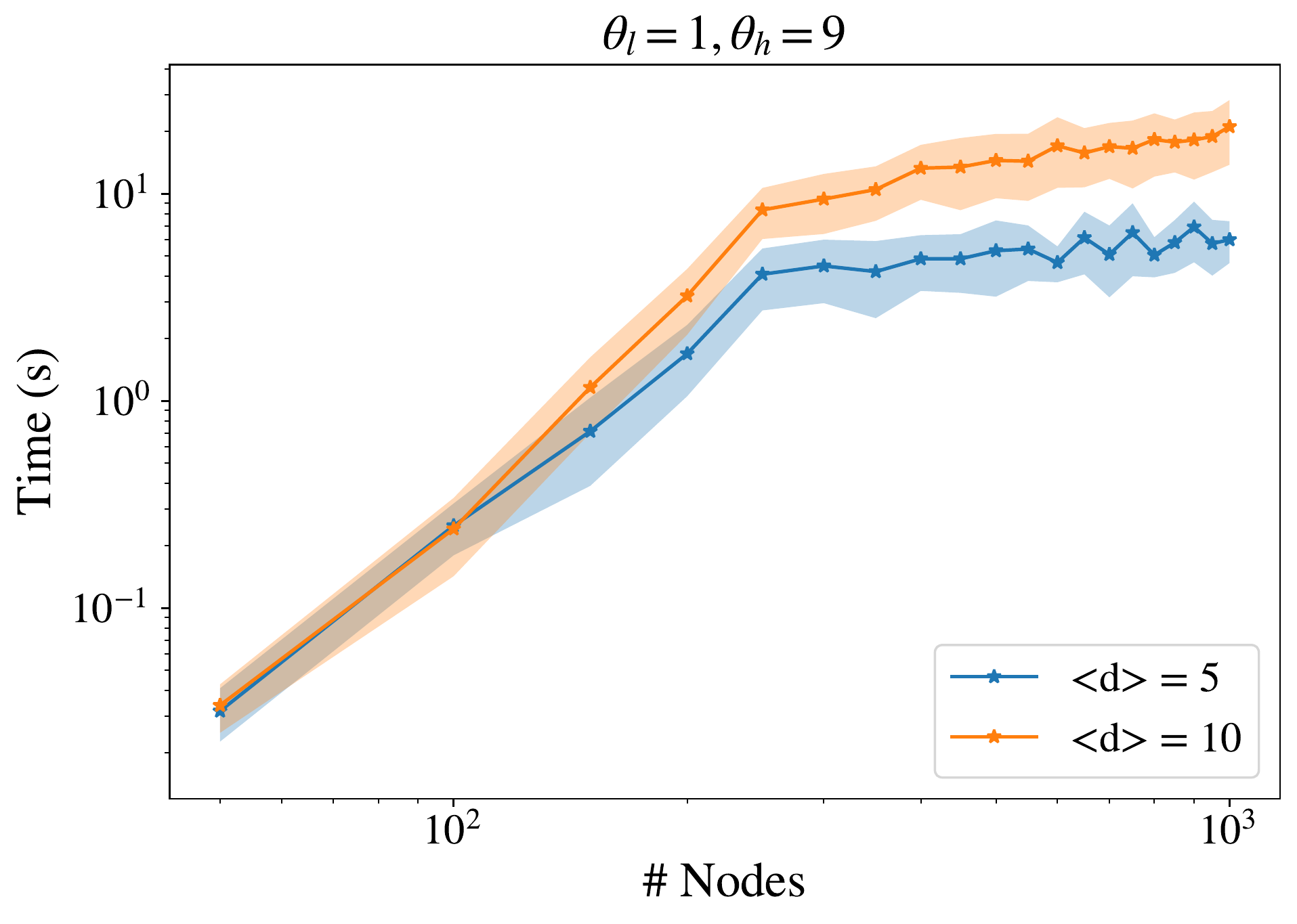} \\
		\includegraphics[width=.35\textwidth]{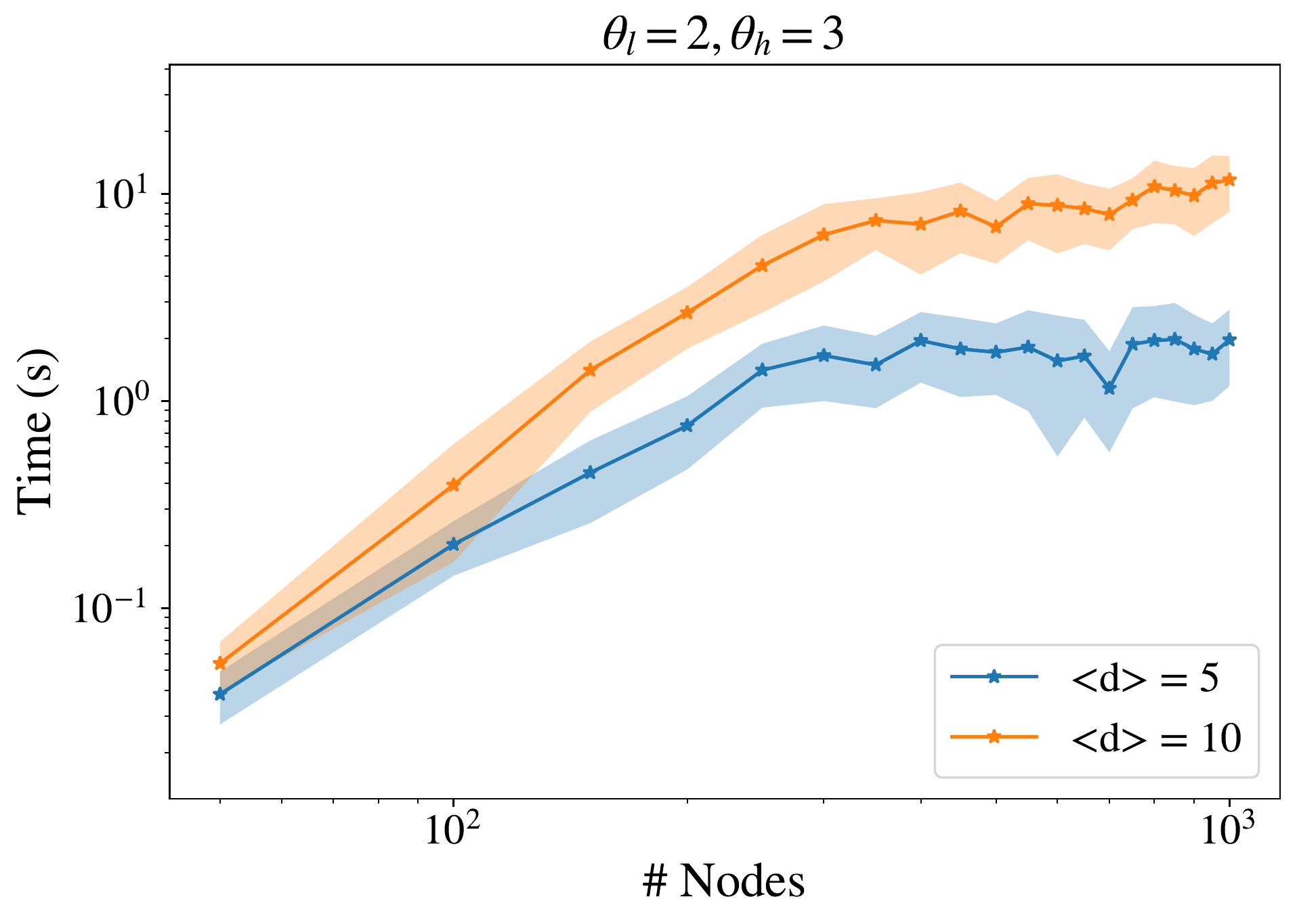} & 
		\includegraphics[width=.35\textwidth]{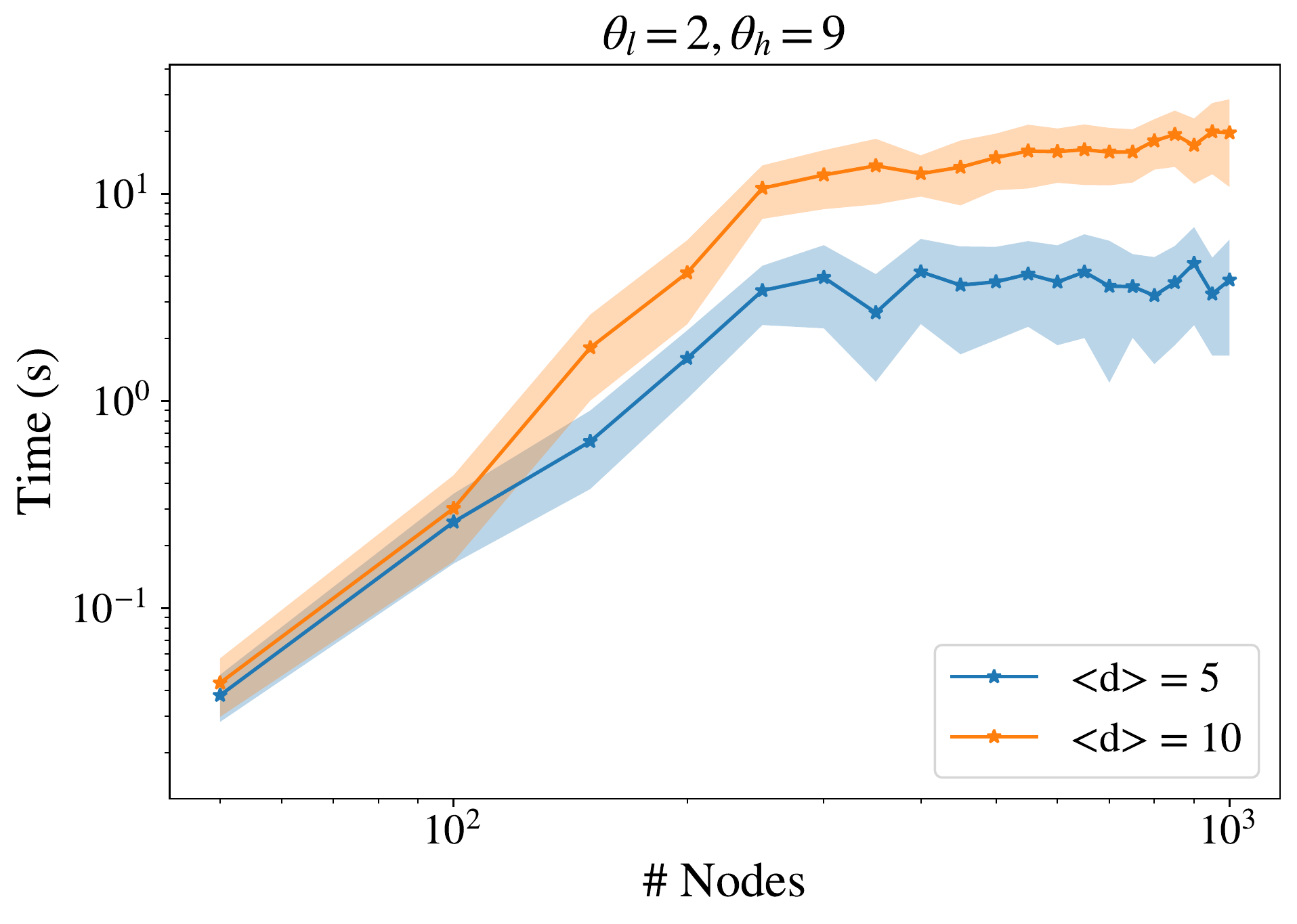}
	\end{tabular}
	\caption{Time consumption of the CDS method on ER random graphs of increasing sizes ($x$-axis) when $k=3$, where we consider different combinations of parameters ($\theta_l, \theta_h$) while maintaining the same mean degree $5$ (blue) and $10$ (orange), and the results are obtained from $50$ samples of the ER graphs of each size $n$, with dots showing the mean value and shades indicating the standard deviation.}
	\label{fig:app_timecomplex_K3}
\end{figure*}
\begin{figure*}
	\centering 
	\begin{tabular}{cc}
		\includegraphics[width=.35\textwidth]{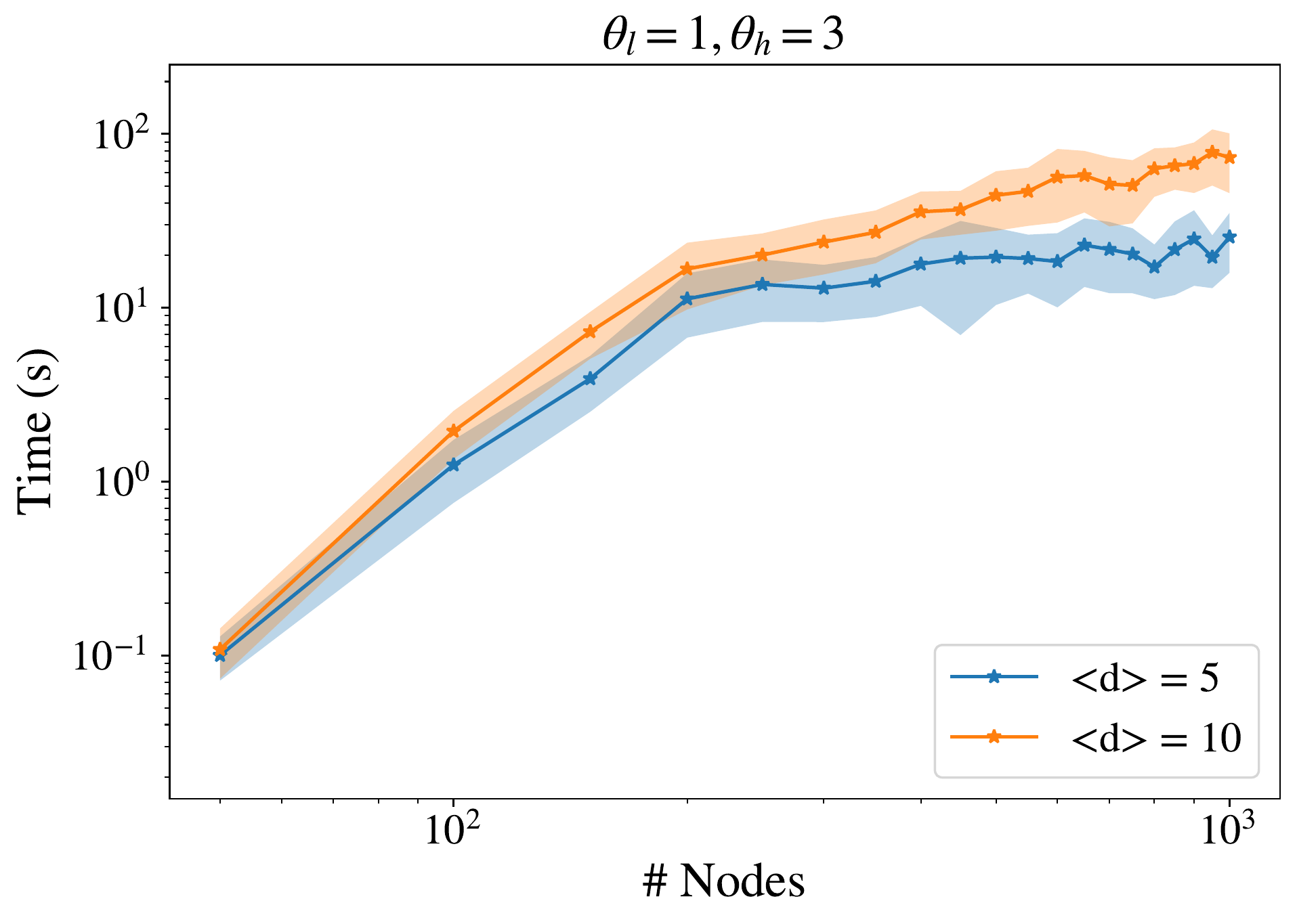} & 
		\includegraphics[width=.35\textwidth]{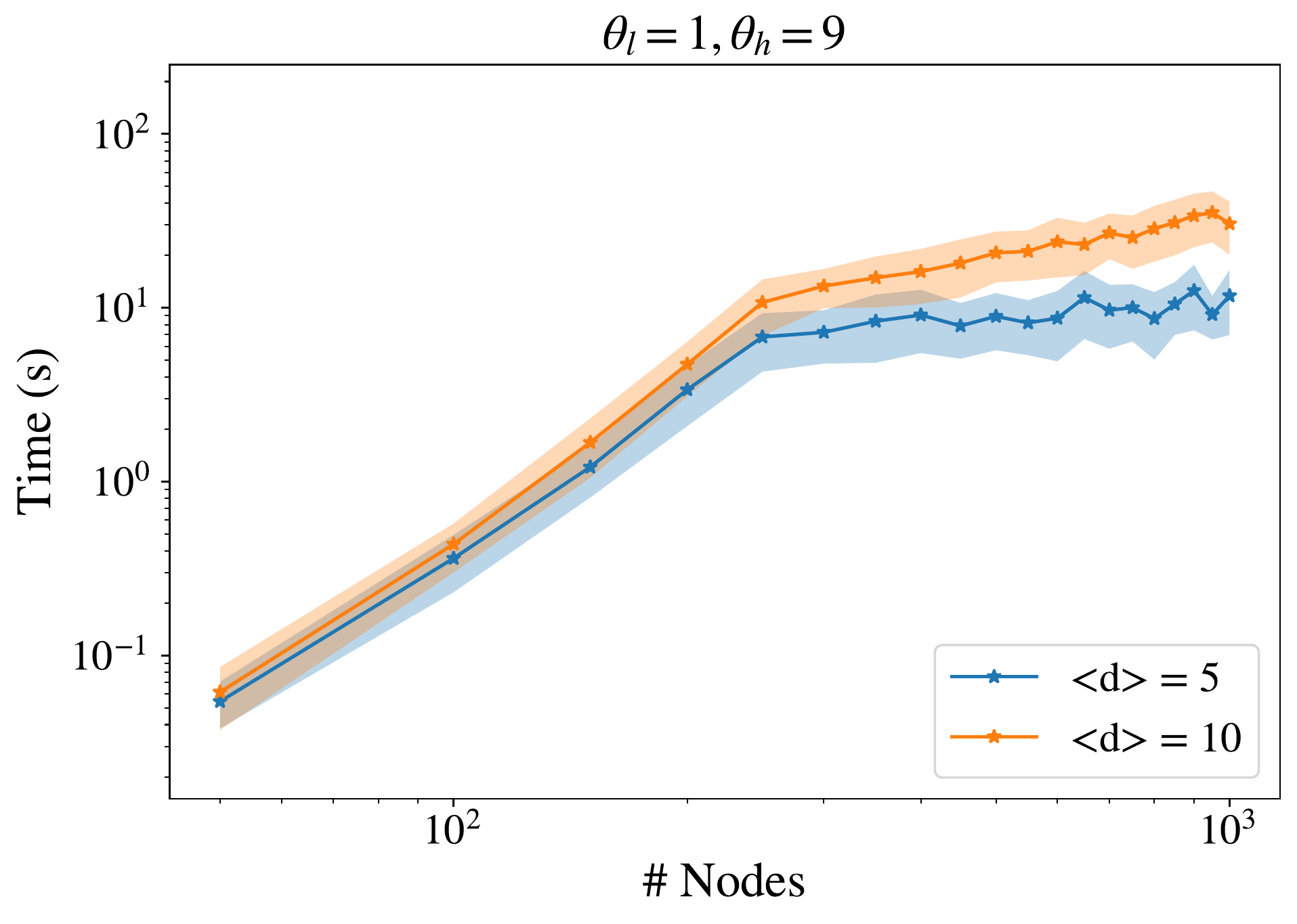} \\
		\includegraphics[width=.35\textwidth]{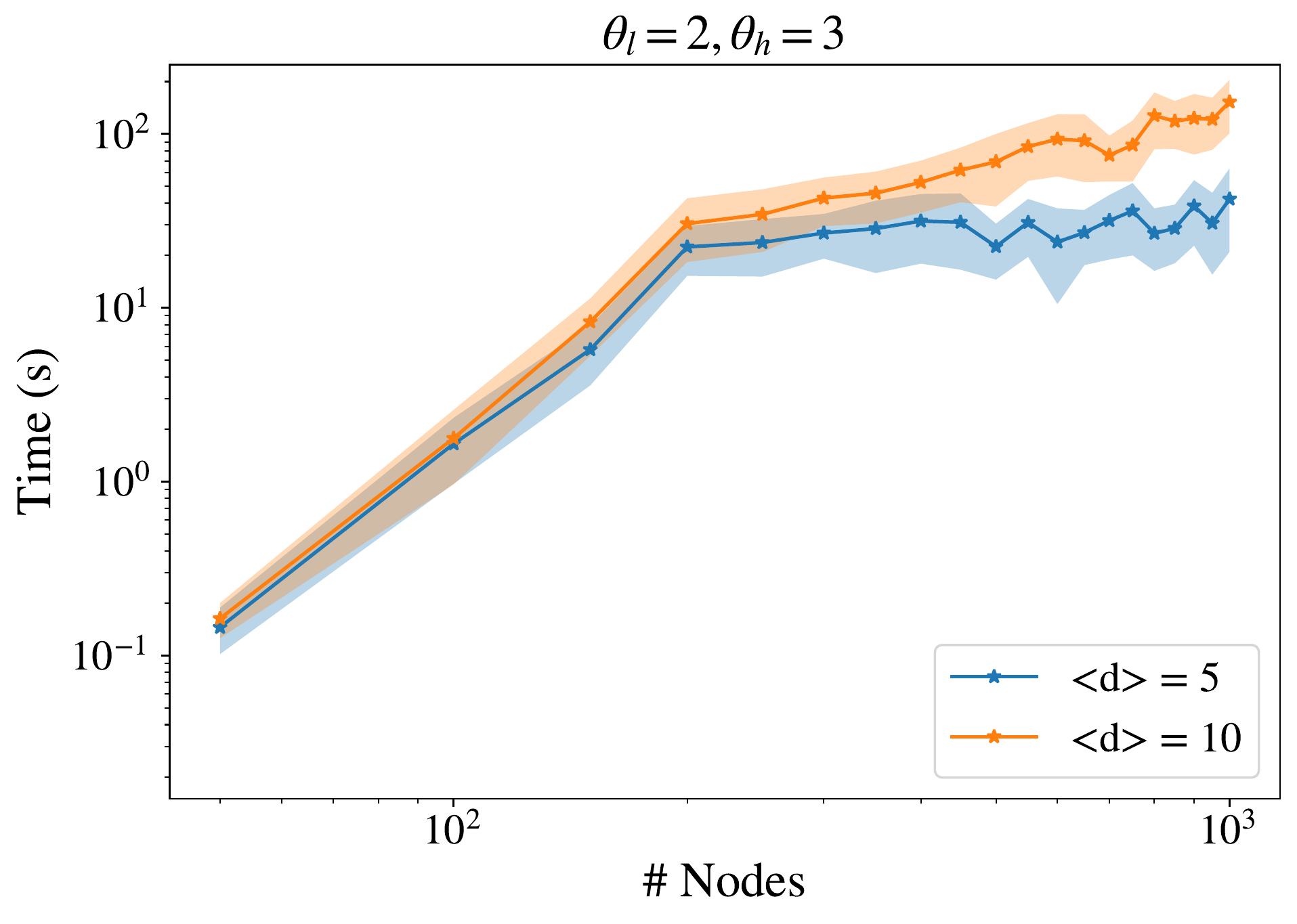} & 
		\includegraphics[width=.35\textwidth]{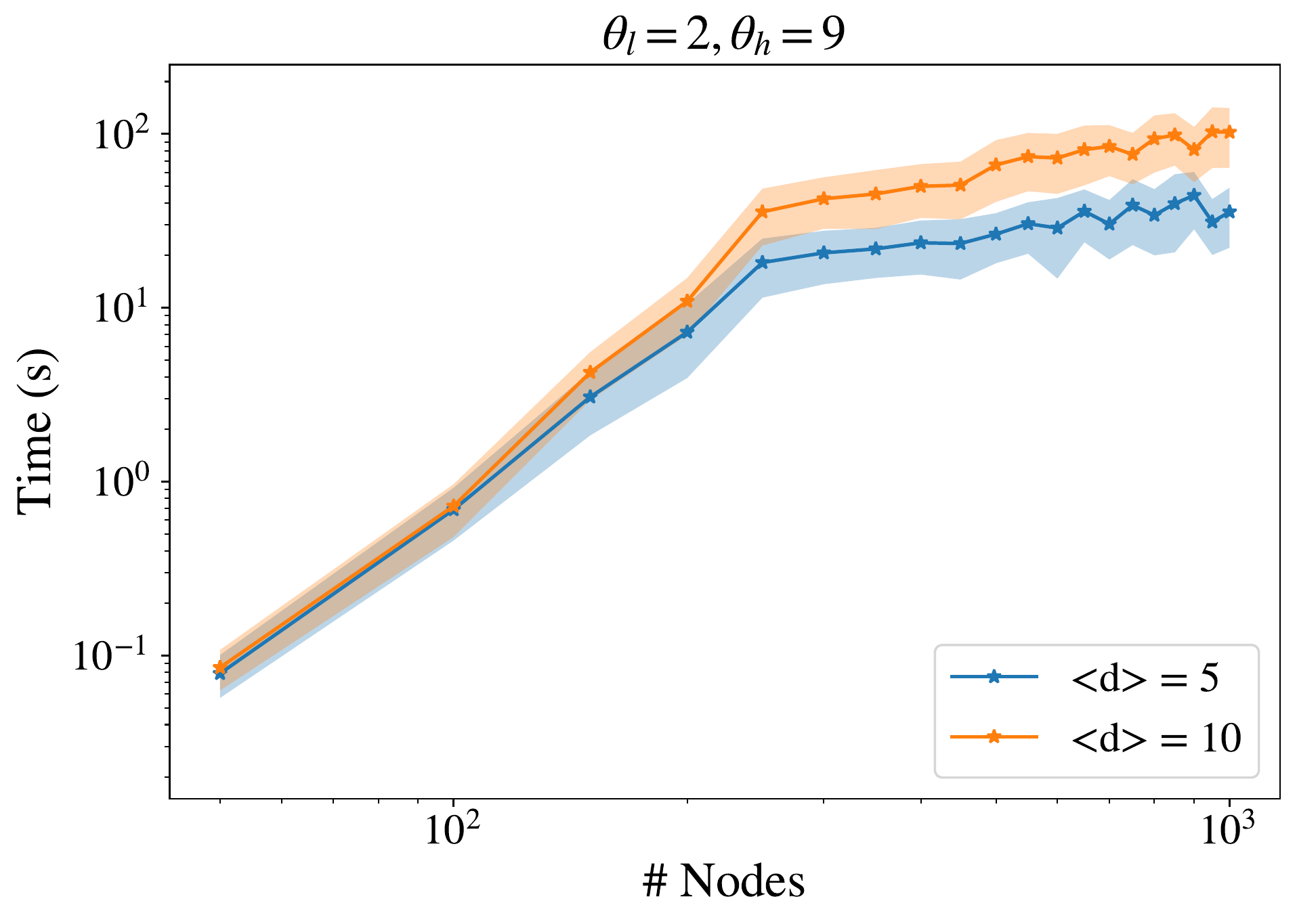} \\
		\includegraphics[width=.35\textwidth]{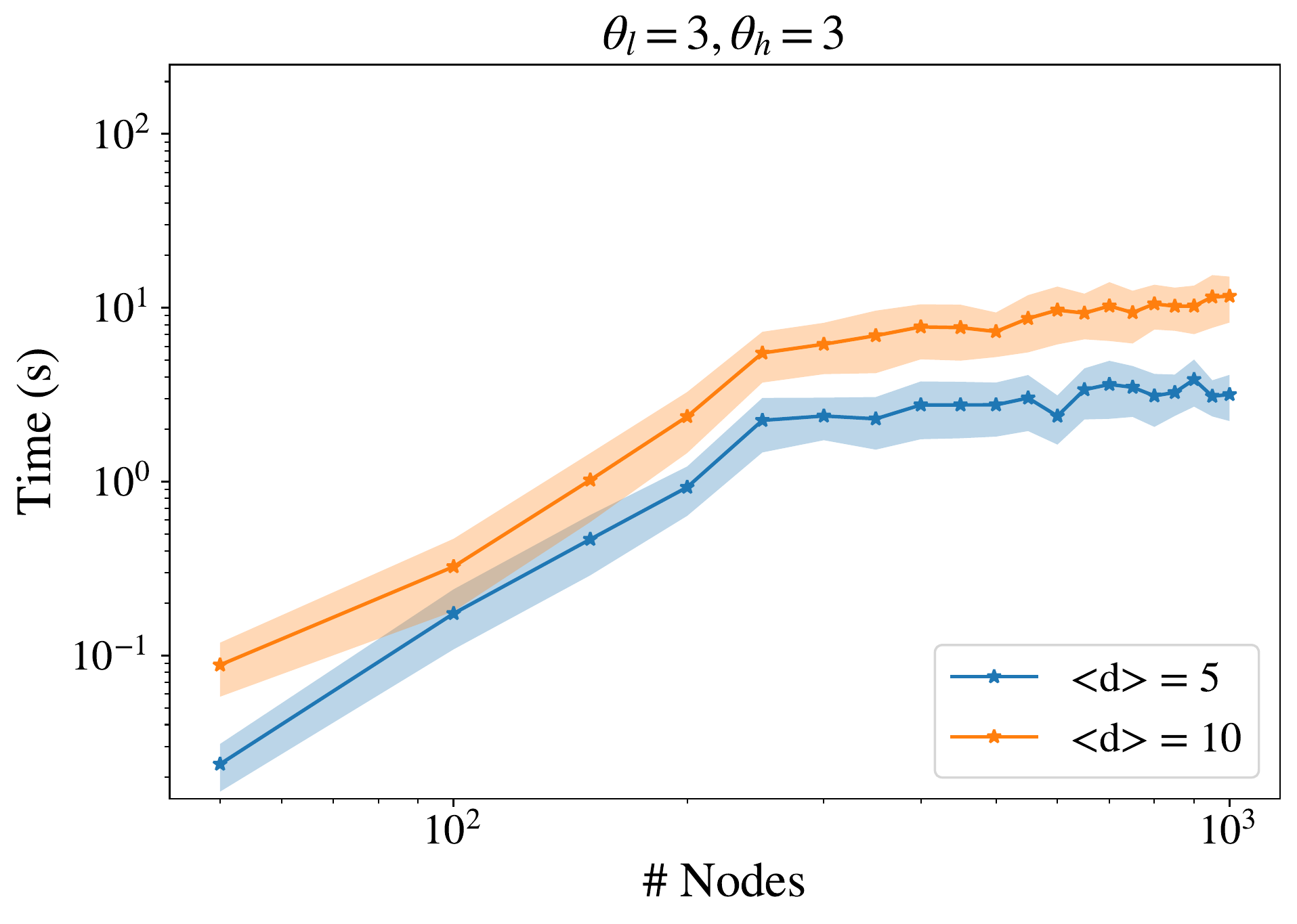} & 
		\includegraphics[width=.35\textwidth]{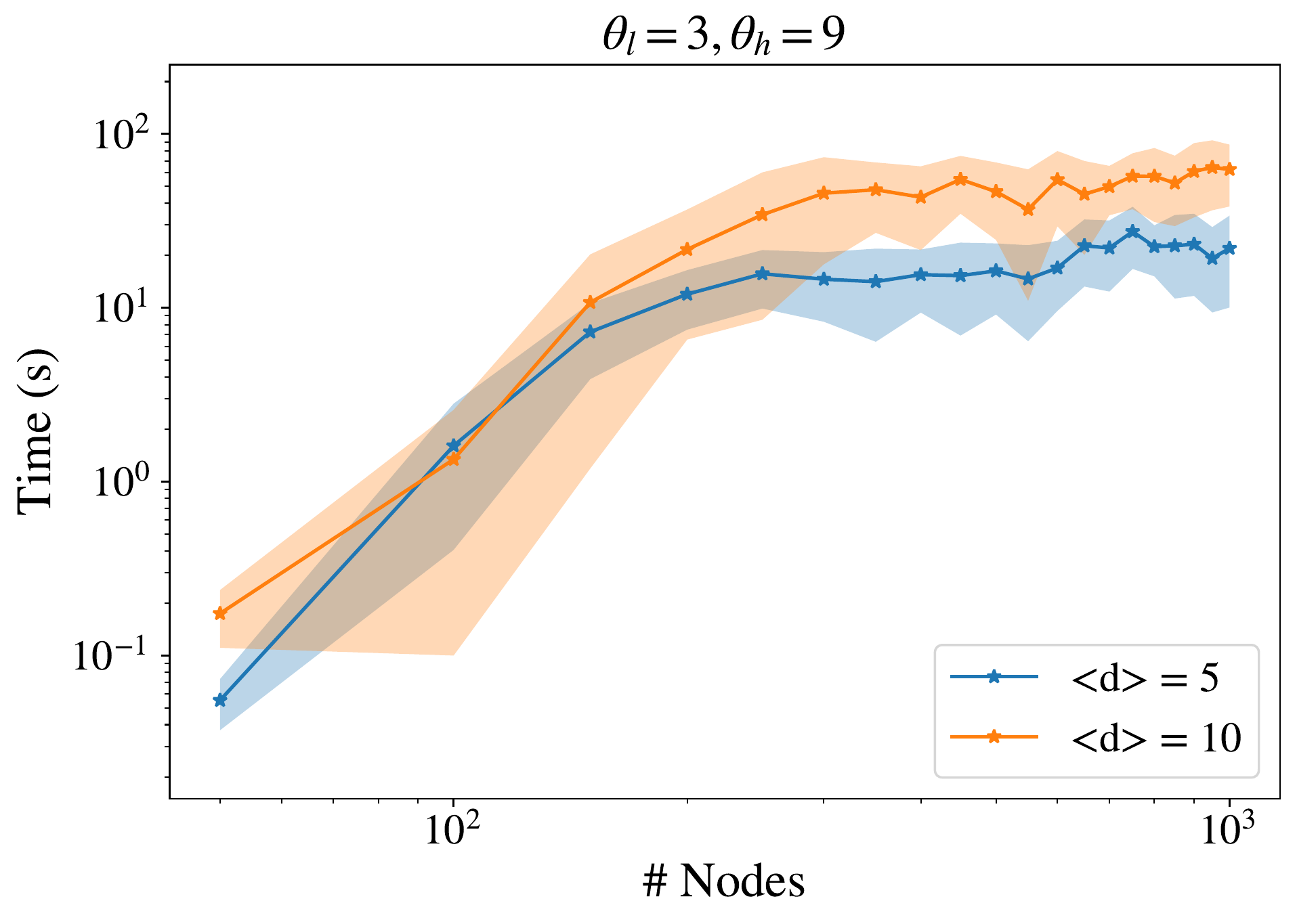}
	\end{tabular}
	\caption{Time consumption of the CDS method on ER random graphs of increasing sizes ($x$-axis) when $k=5$, where we consider different combinations of parameters ($\theta_l, \theta_h$) while maintaining the same mean degree $5$ (blue) and $10$ (orange), and the results are obtained from $50$ samples of the ER graphs of each size $n$, with dots showing the mean value and shades indicating the standard deviation.}
	\label{fig:app_timecomplex_K5}
\end{figure*}
Overall, the time increases linearly when the network is small, and once the network size reaches a critical value, the increase slows down; see Fig.~\ref{fig:app_timecomplex_K3} for $k=3$ and Fig.~\ref{fig:app_timecomplex_K5} for $k=5$, with different combinations of parameters ($\theta_l, \theta_h$). Having different mean degree, upper bounds, lower bounds, or budget sizes normally causes shifts in the trends but not the overall shape.

We now integrate the experimental results with the theoretical understanding. For (I), we know from Algorithm \ref{alg:influ_evaluate} that only neighbors of currently active nodes are considered when calculating the state value in the next step. Since we maintain the same mean degree in our setting, the difference in time lies in $t_\epsilon$ which is at least $O(1)$. For (i), the size of the feasible discrete neighborhood $k(n-k)$ increases linearly in $n$, and we also note that when $\theta_l > 1$, to include nodes that share common neighbors with other initially activated nodes is the only way to improve the overall influence (on networks with uniform edge weights), which is also applied in the algorithm (in refining the mesh at each time step) to reduce the complexity to be lower than $O(n)$. From the experiments, we observe linear increase, and also that the increase slows down significantly after certain critical value of $n$. Therefore, it is reasonable to conjecture that (iib), the number of evaluations required by the CDS method divided by the feasible neighborhood size, is approximately $O(1)$ in the IM problem in real networks. 

These all together result in the value of (II) to be approximately $O(n)$. To measure the goodness of this complexity value, we compare it with a global algorithm, the brute-force method, where each node set will be evaluated. Then, overall $\binom{n}{k}$ evaluations are needed, which is proportional to $O(n^k)$ since again $k = O(1)$. Therefore, $O(n^{k-1})$ more multiples of evaluations are needed in the global algorithm than the CDS method.

The time complexity of the CDS method can be further improved by parallel programming, although it is described in Algorithm \ref{alg:mv-mads_custom} and currently performed serially. For instance, the tasks to evaluate the objective function of different candidates in a discrete neighborhood can run in parallel. There are also other optimization techniques to further reduce the complexity, e.g.,~coarse evaluation of the objective function with a higher tolerance at early stage. All these can be included in the future work.

\section{\label{app:SBM} SBMs}
In this section, we discuss the properties of SBMs mentioned in the main text in more detail. We first show that when increasing the upper bounds in the GIP model, the increase in the expected influence (in one time step) from the node set in one community is higher than that from the other evenly distributed in the two communities, as in Sec.~\ref{sec:diff_props-gen}, and then prove that a node is expected to have higher Katz centrality if it has higher degree centrality than others, as in Sec.~\ref{sec:experiments_comparison}. 

\begin{claim}
	With $l_{j,0} = h_{j,0} = l_0$, and $l_{j,1} = 2\alpha l_0,\ \forall v_j\in V$ in the GIP model, if $SBM(p_{in}, p_{out})$ of two equally-sized \footnote{\footnotesize{Note the same results can be obtained with arbitrary community sizes, but extra conditions on both the community sizes and the probabilities are required.}} communities, $\mathcal{B}_1, \mathcal{B}_2$, and uniform edge weight $\alpha$ satisfies
	\begin{align}
		p_{in} > p_{out},\ (1 - p_{in}) > p_{out},
		\label{equ:app_condi-1}
	\end{align}
	or
	\begin{align}
		p_{out} > p_{in},\ (1 - p_{out}) > p_{in},
		\label{equ:app_condi-2}
	\end{align}
	then the increase in the expected influence, $\mathbb{E}[\sum_{j}(1-\gamma)x_j(1)]$, from the initially activated node set (i) $\mathcal{A}_0 = \{v_{i_1}, v_{i_2}, v_{i_3}, v_{i_4}\} \subset \mathcal{B}_1$, is larger than that from (ii) $\mathcal{A}_0 = \{v_{j_1}, v_{j_2}, v_{j_3}, v_{j_4}\}$ with $v_{j_1}, v_{j_2}\in \mathcal{B}_1$ and $v_{j_3}, v_{j_4}\in \mathcal{B}_2$, when $h_{j,1}$ rises from $h_{j,1} = l_{j,1} = 2\alpha l_0$ to $h_{j,1} = 2l_{j,1} = 4\alpha l_0,\ \forall v_j\in V$.
\end{claim}
\begin{proof}
	For the SBM, $\mathbf{W} = \alpha\mathbf{A}$, where $\mathbf{A}$ is the (unweighted) adjacency matrix. We maintain the same notations for the block membership of each node $v_i$, $\sigma_i\in\{1,2\}$, the linear part of the state vector, $\mathbf{y}(t+1) = \mathbf{W}^T\mathbf{x}(t)$, and the size of each community $n_b$, as in the proof of Claim \ref{cla:linear_difference} in Appendix \ref{app:proofs_model}.
	
	For set (i),
	\begin{align*}
		y_j(1) = \alpha l_0\left[Bin(4, p_{in})\delta(\sigma_j, 1) + Bin(4, p_{out})\delta(\sigma_j, 2)\right], 
	\end{align*}
	while for set (ii),
	\begin{align*}
		y_j(1) = \alpha l_0\left[Bin(2, p_{in}) + Bin(2, p_{out})\right]. 
	\end{align*}
	When $h_{j,1} = l_{j,1} = 2\alpha l_0,\ \forall v_j\in V$, 
	\begin{align}
	    \begin{split}
	        \mathbb{E}&\left[\sum_{j}x_j(1)\right] \\
	        =& \sum_{j} 0P\left[y_j(1) < l_{j,1}\right] + 2\alpha l_0P\left[y_j(1)\ge 2\alpha l_0\right].
	    \end{split}
	    \label{equ:app_exp1}
	\end{align}
	When $h_{j,1} = 4\alpha l_0,\ \forall v_j\in V$, 
	\begin{align}
		\begin{split}
		    \mathbb{E}&\left[\sum_{j}x_j(1)\right] \\
		    =& \sum_{j} \{0P\left[y_j(1) < l_{j,1}\right] + 2\alpha l_0P\left[y_j(1) = 2\alpha l_0\right] \\
    		&+ 3\alpha l_0P\left[y_j(1) = 3\alpha l_0\right] + 4\alpha l_0P\left[y_j(1) = 4\alpha l_0\right]\}.
    	\end{split}
    	\label{equ:app_exp2}
	\end{align}
	Hence, the increase in the expected influence, ignoring the factor $1-\gamma$, is the difference between the two equations (\ref{equ:app_exp2})  and (\ref{equ:app_exp1}),
	\begin{align*}
		\Delta = \sum_{j}\alpha l_0P\left[y_j(1) = 3\alpha l_0\right] + 2\alpha l_0P\left[y_j(1) = 4\alpha l_0\right]. 
	\end{align*}
	For set (i),
	\begin{widetext}
	    \begin{align*}
    		\Delta_{(i)} &= \sum_{j}\alpha l_0\left[\binom{4}{3}p_{in}^3(1-p_{in})\delta(\sigma_j, 1) + \binom{4}{3}p_{out}^3(1-p_{out})\delta(\sigma_j, 2)\right] + 2\alpha l_0\left[p_{in}^4\delta(\sigma_j, 1) + p_{out}^4\delta(\sigma_j, 2)\right]\\
    		&= n_b\times\alpha l_0\times 2\left[2p_{in}^3(1-p_{in}) + 2p_{out}^3(1-p_{out}) + p_{in}^4 + p_{out}^4\right],
	\end{align*}
	\end{widetext}
	while for set (ii), 
	\begin{widetext}
	    \begin{align*}
		\Delta_{(ii)} &= \sum_{j}\alpha l_0\left[\binom{2}{1}p_{in}(1-p_{in})p_{out}^2 + p_{in}^2\binom{2}{1}p_{out}(1-p_{out})\right] + 2\alpha l_0\left(p_{in}^2 p_{out}^2\right)\\
		&= 2n_b\times\alpha l_0\times \left[2p_{in}^2p_{out}(1-p_{out}) + 2p_{out}^2p_{in}(1-p_{in}) + 2p_{in}^2p_{out}^2\right].
	\end{align*}
	\end{widetext}
	Hence, 
	\begin{align*}
		&\Delta_{(i)} - \Delta_{(ii)} \\
		&= 2n_b\alpha l_0\{2\left(p_{in}^2 - p_{out}^2\right)\left[p_{in}(1 - p_{in}) - p_{out}(1 - p_{out})\right]\\
		&\quad + \left(p_{in}^2 - p_{out}^2\right)^2\},
	\end{align*}
	which is positive given conditions (\ref{equ:app_condi-1}) or (\ref{equ:app_condi-2}). 
\end{proof}

\begin{theorem}
	In a two-block SBM with the connecting probabilities in the two communities, $\mathcal{B}_1, \mathcal{B}_2$, being $p_1,p_2$, respectively, and between the two being $p_{12}$, if a node $v_i$ is expected to have higher degree centrality than another node $v_j$, 
	\begin{align}
    	\mathbb{E}\left[\sum_rA_{ir}\right] > \mathbb{E}\left[\sum_rA_{jr}\right],
    	\label{equ:app_deg}
	\end{align}
	where $\mathbf{A} = (A_{ij})$ is the (unweighted) adjacency matrix, then node $v_i$ is expected to have higher Katz centrality than node $v_j$, 
	\begin{align}
    	\mathbb{E}\left[\sum_{t=1}^{\infty}\alpha_{katz}^t\sum_rA^t_{ri}\right] > \mathbb{E}\left[\sum_{t=1}^{\infty}\alpha_{katz}^t\sum_rA^t_{rj}\right],
    	\label{equ:app_katz}
	\end{align} 
	where $\alpha_{katz}$ is the discounting factor.
\end{theorem}
\begin{proof}
	In such SBM, for each pair of nodes $v_i,v_j$, $A_{ij}$ is an independently distributed Bernoulli random variable, with success probability $p_1$ if $v_i, v_j\in \mathcal{B}_1$, $p_2$ if $v_i, v_j\in \mathcal{B}_2$, and $p_{12}$ otherwise.
	Hence, nodes in the same communities are equivalent, and there are only two distinct expected values in both centralities, one for each community. 
	
	For the degree centrality \footnote{\footnotesize{Note that the expected values could be slightly different due to the common assumption of no self-edges. However, we assume $n \gg 1$, thus ignore such differences.}},
	\begin{align*}
		\mathbb{E}\left[\sum_rA_{ir}\right] = \begin{cases}
		n_1p_1 + n_2p_{12},\quad v_i\in\mathcal{B}_1,\\
		n_1p_{12} + n_2p_2,\quad v_i\in\mathcal{B}_2,
	\end{cases}
	\end{align*}
	where $n_1, n_2$ are the sizes of communities $\mathcal{B}_1, \mathcal{B}_2$, respectively. Hence, condition (\ref{equ:app_deg}) can only happen when nodes $v_i$ and $v_j$ are in different communities. 
	
	Without loss of generality, we assume $v_i\in \mathcal{B}_1$, and then $v_j\in \mathcal{B}_2$. We show that (\ref{equ:app_katz}) holds true by proving the following stronger relationship where for each $t > 0$, 
	\begin{align}
    	\mathbb{E}\left[\sum_rA^t_{ri}\right] > \mathbb{E}\left[\sum_rA^t_{rj}\right].
    	\label{equ:app_induction}
	\end{align}
	We show it by induction on $t$. (i) When $t=1$, 
	\begin{align*}
	    \mathbb{E}\left[\sum_rA_{ri}\right] &= n_1p_1 + n_2p_{12} \\
	    & > n_1p_{12} + n_2p_2 = \mathbb{E}\left[\sum_rA_{rj}\right],
	\end{align*}
	where the inequality is by condition (\ref{equ:app_deg}). (ii) Suppose (\ref{equ:app_induction}) is true for all $t\le t'$. Then when $t = t'+1$, 
	\begin{align}
		&\mathbb{E}\left[\sum_rA^{t'+1}_{ri}\right]\\ 
		&= \mathbb{E}\left[\sum_r\sum_qA^{t'}_{rq}A_{qi}\right]\nonumber\\
		&= \sum_{v_q\in\mathcal{B}_1}\mathbb{E}\left[\sum_rA^{t'}_{rq}\right]p_1 + \sum_{v_q\in\mathcal{B}_2}\mathbb{E}\left[\sum_rA^{t'}_{rq}\right]p_{12}\nonumber\\
		&= \mathbb{E}\left[\sum_rA^{t'}_{ri}\right]n_1p_1 + \mathbb{E}\left[\sum_rA^{t'}_{rj}\right]n_2p_{12},
		\label{equ:app_induct1}
	\end{align}
	where the second equality is by independence, and the last equality is by equivalence among nodes in the same communities. Similarly, 
	\begin{align}
		\mathbb{E}\left[\sum_rA^{t'+1}_{rj}\right] 
		= \mathbb{E}\left[\sum_rA^{t'}_{ri}\right]n_1p_{12} + \mathbb{E}\left[\sum_rA^{t'}_{rj}\right]n_2p_{2}. 
		\label{equ:app_induct2}
	\end{align}
	Hence, the difference (\ref{equ:app_induct1}) - (\ref{equ:app_induct2}) is
	\begin{align*}
		&\mathbb{E}\left[\sum_rA^{t'}_{ri}\right]n_1(p_1 - p_{12}) + \mathbb{E}\left[\sum_rA^{t'}_{rj}\right]n_2(p_{12} - p_{2}) \\
		&> \mathbb{E}\left[\sum_rA^{t'}_{rj}\right](n_1(p_1 - p_{12}) + n_2(p_{12} - p_{2})) > 0,
	\end{align*}
	where the first inequality is by induction hypothesis, and the last inequality is by condition (\ref{equ:app_deg}). 
\end{proof}


\bibliography{apssamp}

\providecommand{\noopsort}[1]{}\providecommand{\singleletter}[1]{#1}%
\begin{thebibliography}{56}%
\makeatletter
\providecommand \@ifxundefined [1]{%
 \@ifx{#1\undefined}
}%
\providecommand \@ifnum [1]{%
 \ifnum #1\expandafter \@firstoftwo
 \else \expandafter \@secondoftwo
 \fi
}%
\providecommand \@ifx [1]{%
 \ifx #1\expandafter \@firstoftwo
 \else \expandafter \@secondoftwo
 \fi
}%
\providecommand \natexlab [1]{#1}%
\providecommand \enquote  [1]{``#1''}%
\providecommand \bibnamefont  [1]{#1}%
\providecommand \bibfnamefont [1]{#1}%
\providecommand \citenamefont [1]{#1}%
\providecommand \href@noop [0]{\@secondoftwo}%
\providecommand \href [0]{\begingroup \@sanitize@url \@href}%
\providecommand \@href[1]{\@@startlink{#1}\@@href}%
\providecommand \@@href[1]{\endgroup#1\@@endlink}%
\providecommand \@sanitize@url [0]{\catcode `\\12\catcode `\$12\catcode
  `\&12\catcode `\#12\catcode `\^12\catcode `\_12\catcode `\%12\relax}%
\providecommand \@@startlink[1]{}%
\providecommand \@@endlink[0]{}%
\providecommand \url  [0]{\begingroup\@sanitize@url \@url }%
\providecommand \@url [1]{\endgroup\@href {#1}{\urlprefix }}%
\providecommand \urlprefix  [0]{URL }%
\providecommand \Eprint [0]{\href }%
\providecommand \doibase [0]{https://doi.org/}%
\providecommand \selectlanguage [0]{\@gobble}%
\providecommand \bibinfo  [0]{\@secondoftwo}%
\providecommand \bibfield  [0]{\@secondoftwo}%
\providecommand \translation [1]{[#1]}%
\providecommand \BibitemOpen [0]{}%
\providecommand \bibitemStop [0]{}%
\providecommand \bibitemNoStop [0]{.\EOS\space}%
\providecommand \EOS [0]{\spacefactor3000\relax}%
\providecommand \BibitemShut  [1]{\csname bibitem#1\endcsname}%
\let\auto@bib@innerbib\@empty
\bibitem [{\citenamefont {Bakshy}\ \emph {et~al.}(2012)\citenamefont {Bakshy},
  \citenamefont {Rosenn}, \citenamefont {Marlow},\ and\ \citenamefont
  {Adamic}}]{Bakshy_InfoRole_2012}%
  \BibitemOpen
  \bibfield  {author} {\bibinfo {author} {\bibfnamefont {E.}~\bibnamefont
  {Bakshy}}, \bibinfo {author} {\bibfnamefont {I.}~\bibnamefont {Rosenn}},
  \bibinfo {author} {\bibfnamefont {C.}~\bibnamefont {Marlow}},\ and\ \bibinfo
  {author} {\bibfnamefont {L.}~\bibnamefont {Adamic}},\ }\bibfield  {title}
  {\bibinfo {title} {The role of social networks in information diffusion},\
  }in\ \href {https://doi.org/10.1145/2187836.2187907} {\emph {\bibinfo
  {booktitle} {Proceedings of the 21st International Conference on World Wide
  Web}}}\ (\bibinfo  {publisher} {ACM},\ \bibinfo {address} {New York},\
  \bibinfo {year} {2012})\ pp.\ \bibinfo {pages} {519--528}\BibitemShut
  {NoStop}%
\bibitem [{\citenamefont {Centola}(2010)}]{Centola_InfoDiffBehav_2010}%
  \BibitemOpen
  \bibfield  {author} {\bibinfo {author} {\bibfnamefont {D.}~\bibnamefont
  {Centola}},\ }\bibfield  {title} {\bibinfo {title} {The spread of behavior in
  an online social network experiment},\ }\href@noop {} {\bibfield  {journal}
  {\bibinfo  {journal} {Science}\ }\textbf {\bibinfo {volume} {329}},\ \bibinfo
  {pages} {1194} (\bibinfo {year} {2010})}\BibitemShut {NoStop}%
\bibitem [{\citenamefont {Nekovee}\ \emph {et~al.}(2007)\citenamefont
  {Nekovee}, \citenamefont {Moreno}, \citenamefont {Bianconi},\ and\
  \citenamefont {Marsili}}]{Nekovee_RumorDiff_2007}%
  \BibitemOpen
  \bibfield  {author} {\bibinfo {author} {\bibfnamefont {M.}~\bibnamefont
  {Nekovee}}, \bibinfo {author} {\bibfnamefont {Y.}~\bibnamefont {Moreno}},
  \bibinfo {author} {\bibfnamefont {G.}~\bibnamefont {Bianconi}},\ and\
  \bibinfo {author} {\bibfnamefont {M.}~\bibnamefont {Marsili}},\ }\bibfield
  {title} {\bibinfo {title} {Theory of rumour spreading in complex social
  networks},\ }\href@noop {} {\bibfield  {journal} {\bibinfo  {journal} {Phys.
  A}\ }\textbf {\bibinfo {volume} {374}},\ \bibinfo {pages} {457} (\bibinfo
  {year} {2007})}\BibitemShut {NoStop}%
\bibitem [{\citenamefont {Bovet}\ and\ \citenamefont
  {Makse}(2019)}]{Bovet_Twitter_2019}%
  \BibitemOpen
  \bibfield  {author} {\bibinfo {author} {\bibfnamefont {A.}~\bibnamefont
  {Bovet}}\ and\ \bibinfo {author} {\bibfnamefont {H.}~\bibnamefont {Makse}},\
  }\bibfield  {title} {\bibinfo {title} {Influence of fake news in {T}witter
  during the 2016 {US} presidential election},\ }\bibfield  {journal} {\bibinfo
   {journal} {Nat. Commun.}\ }\textbf {\bibinfo {volume} {10}},\ \href
  {https://doi.org/10.1038/s41467-018-07761-2} {10.1038/s41467-018-07761-2}
  (\bibinfo {year} {2019})\BibitemShut {NoStop}%
\bibitem [{\citenamefont {Chen}\ \emph
  {et~al.}(2010{\natexlab{a}})\citenamefont {Chen}, \citenamefont {Wang},\ and\
  \citenamefont {Wang}}]{Chen_InflMaxViral_2010}%
  \BibitemOpen
  \bibfield  {author} {\bibinfo {author} {\bibfnamefont {W.}~\bibnamefont
  {Chen}}, \bibinfo {author} {\bibfnamefont {C.}~\bibnamefont {Wang}},\ and\
  \bibinfo {author} {\bibfnamefont {Y.}~\bibnamefont {Wang}},\ }\bibfield
  {title} {\bibinfo {title} {Scalable influence maximization for prevalent
  viral marketing in large-scale social networks},\ }in\ \href
  {https://doi.org/10.1145/1835804.1835934} {\emph {\bibinfo {booktitle}
  {Proceedings of the 16th ACM SIGKDD International Conference on Knowledge
  Discovery and Data Mining}}}\ (\bibinfo  {publisher} {ACM},\ \bibinfo
  {address} {New York},\ \bibinfo {year} {2010})\ pp.\ \bibinfo {pages}
  {1029--1038}\BibitemShut {NoStop}%
\bibitem [{\citenamefont {Leskovec}\ \emph
  {et~al.}(2007{\natexlab{a}})\citenamefont {Leskovec}, \citenamefont
  {Adamic},\ and\ \citenamefont {Huberman}}]{Leskovec_ViralMarket_2007}%
  \BibitemOpen
  \bibfield  {author} {\bibinfo {author} {\bibfnamefont {J.}~\bibnamefont
  {Leskovec}}, \bibinfo {author} {\bibfnamefont {L.}~\bibnamefont {Adamic}},\
  and\ \bibinfo {author} {\bibfnamefont {B.}~\bibnamefont {Huberman}},\
  }\bibfield  {title} {\bibinfo {title} {The dynamics of viral marketing},\
  }\href@noop {} {\bibfield  {journal} {\bibinfo  {journal} {ACM Trans. Web}\
  }\textbf {\bibinfo {volume} {1}},\ \bibinfo {pages} {5} (\bibinfo {year}
  {2007}{\natexlab{a}})}\BibitemShut {NoStop}%
\bibitem [{\citenamefont {Mossel}\ and\ \citenamefont
  {Roch}(2010)}]{Mossel_Roch_2010}%
  \BibitemOpen
  \bibfield  {author} {\bibinfo {author} {\bibfnamefont {E.}~\bibnamefont
  {Mossel}}\ and\ \bibinfo {author} {\bibfnamefont {S.}~\bibnamefont {Roch}},\
  }\bibfield  {title} {\bibinfo {title} {Submodularity of influence in social
  networks: From local to global},\ }\href {https://doi.org/10.1137/080714452}
  {\bibfield  {journal} {\bibinfo  {journal} {SIAM J. Comput.}\ }\textbf
  {\bibinfo {volume} {39}},\ \bibinfo {pages} {2176} (\bibinfo {year}
  {2010})}\BibitemShut {NoStop}%
\bibitem [{\citenamefont {Pastor-Satorras}\ \emph {et~al.}(2015)\citenamefont
  {Pastor-Satorras}, \citenamefont {Castellano}, \citenamefont {Van~Mieghem},\
  and\ \citenamefont {Vespignani}}]{Pastor-Satorras_epidemic_2015}%
  \BibitemOpen
  \bibfield  {author} {\bibinfo {author} {\bibfnamefont {R.}~\bibnamefont
  {Pastor-Satorras}}, \bibinfo {author} {\bibfnamefont {C.}~\bibnamefont
  {Castellano}}, \bibinfo {author} {\bibfnamefont {P.}~\bibnamefont
  {Van~Mieghem}},\ and\ \bibinfo {author} {\bibfnamefont {A.}~\bibnamefont
  {Vespignani}},\ }\bibfield  {title} {\bibinfo {title} {Epidemic processes in
  complex networks},\ }\href@noop {} {\bibfield  {journal} {\bibinfo  {journal}
  {Rev. Mod. Phys.}\ }\textbf {\bibinfo {volume} {87}},\ \bibinfo {pages} {925}
  (\bibinfo {year} {2015})}\BibitemShut {NoStop}%
\bibitem [{\citenamefont {Pastor-Satorras}\ and\ \citenamefont
  {Vespignani}(2004)}]{pastor-satorras_internet_2004}%
  \BibitemOpen
  \bibfield  {author} {\bibinfo {author} {\bibfnamefont {R.}~\bibnamefont
  {Pastor-Satorras}}\ and\ \bibinfo {author} {\bibfnamefont {A.}~\bibnamefont
  {Vespignani}},\ }\bibfield  {title} {\bibinfo {title} {Epidemics in the
  internet},\ }in\ \href {https://doi.org/10.1017/CBO9780511610905} {\emph
  {\bibinfo {booktitle} {Evolution and Structure of the Internet: A Statistical
  Physics Approach}}}\ (\bibinfo  {publisher} {Cambridge University Press},\
  \bibinfo {address} {Cambridge},\ \bibinfo {year} {2004})\ pp.\ \bibinfo
  {pages} {180--210}\BibitemShut {NoStop}%
\bibitem [{\citenamefont {Kempe}\ \emph {et~al.}(2003)\citenamefont {Kempe},
  \citenamefont {Kleinberg},\ and\ \citenamefont
  {Tardos}}]{Kempe_influence_2003}%
  \BibitemOpen
  \bibfield  {author} {\bibinfo {author} {\bibfnamefont {D.}~\bibnamefont
  {Kempe}}, \bibinfo {author} {\bibfnamefont {J.}~\bibnamefont {Kleinberg}},\
  and\ \bibinfo {author} {\bibfnamefont {E.}~\bibnamefont {Tardos}},\
  }\bibfield  {title} {\bibinfo {title} {Maximizing the spread of influence
  through a social network},\ }in\ \href
  {https://doi.org/10.1145/956750.956769} {\emph {\bibinfo {booktitle}
  {Proceedings of the 9th ACM SIGKDD International Conference on Knowledge
  Discovery and Data Mining}}}\ (\bibinfo  {publisher} {ACM},\ \bibinfo
  {address} {New York},\ \bibinfo {year} {2003})\ pp.\ \bibinfo {pages}
  {137--146}\BibitemShut {NoStop}%
\bibitem [{\citenamefont {Shakarian}\ \emph {et~al.}(2015)\citenamefont
  {Shakarian}, \citenamefont {Bhatnagar}, \citenamefont {Aleali}, \citenamefont
  {Shaabani},\ and\ \citenamefont {Guo}}]{Shakarian_models_2015}%
  \BibitemOpen
  \bibfield  {author} {\bibinfo {author} {\bibfnamefont {P.}~\bibnamefont
  {Shakarian}}, \bibinfo {author} {\bibfnamefont {A.}~\bibnamefont
  {Bhatnagar}}, \bibinfo {author} {\bibfnamefont {A.}~\bibnamefont {Aleali}},
  \bibinfo {author} {\bibfnamefont {E.}~\bibnamefont {Shaabani}},\ and\
  \bibinfo {author} {\bibfnamefont {R.}~\bibnamefont {Guo}},\ }\bibfield
  {title} {\bibinfo {title} {The independent cascade and linear threshold
  models},\ }in\ \href {https://doi.org/10.1007/978-3-319-23105-1_4} {\emph
  {\bibinfo {booktitle} {Diffusion in Social Networks}}}\ (\bibinfo
  {publisher} {Springer International Publishing},\ \bibinfo {address} {Cham},\
  \bibinfo {year} {2015})\ pp.\ \bibinfo {pages} {35--48}\BibitemShut {NoStop}%
\bibitem [{\citenamefont {Centola}\ and\ \citenamefont
  {Macy}(2007)}]{Centola_ccontagion_2007}%
  \BibitemOpen
  \bibfield  {author} {\bibinfo {author} {\bibfnamefont {D.}~\bibnamefont
  {Centola}}\ and\ \bibinfo {author} {\bibfnamefont {M.}~\bibnamefont {Macy}},\
  }\bibfield  {title} {\bibinfo {title} {Complex contagions and the weakness of
  long ties},\ }\href@noop {} {\bibfield  {journal} {\bibinfo  {journal} {Amer.
  J. Sociol.}\ }\textbf {\bibinfo {volume} {113}},\ \bibinfo {pages} {702}
  (\bibinfo {year} {2007})}\BibitemShut {NoStop}%
\bibitem [{\citenamefont {Guilbeault}\ \emph {et~al.}(2018)\citenamefont
  {Guilbeault}, \citenamefont {Becker},\ and\ \citenamefont
  {Centola}}]{Guibeault_Complcontag_2018}%
  \BibitemOpen
  \bibfield  {author} {\bibinfo {author} {\bibfnamefont {D.}~\bibnamefont
  {Guilbeault}}, \bibinfo {author} {\bibfnamefont {J.}~\bibnamefont {Becker}},\
  and\ \bibinfo {author} {\bibfnamefont {D.}~\bibnamefont {Centola}},\
  }\bibfield  {title} {\bibinfo {title} {Complex contagions: {A} decade in
  review},\ }in\ \href {https://doi.org/10.1007/978-3-319-77332-2_1} {\emph
  {\bibinfo {booktitle} {Complex Spreading Phenomena in Social Systems:
  {I}nfluence and Contagion in Real-World Social Networks}}},\ \bibinfo
  {editor} {edited by\ \bibinfo {editor} {\bibfnamefont {S.}~\bibnamefont
  {Lehmann}}\ and\ \bibinfo {editor} {\bibfnamefont {Y.}~\bibnamefont {Ahn}}}\
  (\bibinfo  {publisher} {Springer},\ \bibinfo {address} {Cham},\ \bibinfo
  {year} {2018})\ pp.\ \bibinfo {pages} {3--25}\BibitemShut {NoStop}%
\bibitem [{\citenamefont {Ma}\ \emph {et~al.}(2008)\citenamefont {Ma},
  \citenamefont {Yang}, \citenamefont {Lyu},\ and\ \citenamefont
  {King}}]{Ma_heat_2008}%
  \BibitemOpen
  \bibfield  {author} {\bibinfo {author} {\bibfnamefont {H.}~\bibnamefont
  {Ma}}, \bibinfo {author} {\bibfnamefont {H.}~\bibnamefont {Yang}}, \bibinfo
  {author} {\bibfnamefont {M.}~\bibnamefont {Lyu}},\ and\ \bibinfo {author}
  {\bibfnamefont {I.}~\bibnamefont {King}},\ }\bibfield  {title} {\bibinfo
  {title} {Mining social networks using heat diffusion processes for marketing
  candidates selection},\ }in\ \href {https://doi.org/10.1145/1458082.1458115}
  {\emph {\bibinfo {booktitle} {Proceedings of the 17th ACM Conference on
  Information and Knowledge Management}}}\ (\bibinfo  {publisher} {ACM},\
  \bibinfo {address} {New York},\ \bibinfo {year} {2008})\ pp.\ \bibinfo
  {pages} {233--242}\BibitemShut {NoStop}%
\bibitem [{\citenamefont {Degroot}(1974)}]{degroot_model_1974}%
  \BibitemOpen
  \bibfield  {author} {\bibinfo {author} {\bibfnamefont {M.}~\bibnamefont
  {Degroot}},\ }\bibfield  {title} {\bibinfo {title} {Reaching a consensus},\
  }\href@noop {} {\bibfield  {journal} {\bibinfo  {journal} {J. Amer. Statist.
  Assoc.}\ }\textbf {\bibinfo {volume} {69}},\ \bibinfo {pages} {118} (\bibinfo
  {year} {1974})}\BibitemShut {NoStop}%
\bibitem [{\citenamefont {Deffuant}\ \emph {et~al.}(2000)\citenamefont
  {Deffuant}, \citenamefont {Neau}, \citenamefont {Amblard},\ and\
  \citenamefont {Weisbuch}}]{Deffuant_boundconf_2000}%
  \BibitemOpen
  \bibfield  {author} {\bibinfo {author} {\bibfnamefont {G.}~\bibnamefont
  {Deffuant}}, \bibinfo {author} {\bibfnamefont {D.}~\bibnamefont {Neau}},
  \bibinfo {author} {\bibfnamefont {F.}~\bibnamefont {Amblard}},\ and\ \bibinfo
  {author} {\bibfnamefont {G.}~\bibnamefont {Weisbuch}},\ }\bibfield  {title}
  {\bibinfo {title} {Mixing beliefs among interacting agents},\ }\href@noop {}
  {\bibfield  {journal} {\bibinfo  {journal} {Adv. Complex Syst.}\ }\textbf
  {\bibinfo {volume} {03}},\ \bibinfo {pages} {87} (\bibinfo {year}
  {2000})}\BibitemShut {NoStop}%
\bibitem [{\citenamefont {Leskovec}\ \emph
  {et~al.}(2007{\natexlab{b}})\citenamefont {Leskovec}, \citenamefont {Krause},
  \citenamefont {Guestrin}, \citenamefont {Faloutsos}, \citenamefont
  {VanBriesen},\ and\ \citenamefont {Glance}}]{Leskovec_outbreak_2007}%
  \BibitemOpen
  \bibfield  {author} {\bibinfo {author} {\bibfnamefont {J.}~\bibnamefont
  {Leskovec}}, \bibinfo {author} {\bibfnamefont {A.}~\bibnamefont {Krause}},
  \bibinfo {author} {\bibfnamefont {C.}~\bibnamefont {Guestrin}}, \bibinfo
  {author} {\bibfnamefont {C.}~\bibnamefont {Faloutsos}}, \bibinfo {author}
  {\bibfnamefont {J.}~\bibnamefont {VanBriesen}},\ and\ \bibinfo {author}
  {\bibfnamefont {N.}~\bibnamefont {Glance}},\ }\bibfield  {title} {\bibinfo
  {title} {Cost-effective outbreak detection in networks},\ }in\ \href
  {https://doi.org/10.1145/1281192.1281239} {\emph {\bibinfo {booktitle}
  {Proceedings of the 13th ACM SIGKDD International Conference on Knowledge
  Discovery and Data Mining}}}\ (\bibinfo  {publisher} {ACM},\ \bibinfo
  {address} {New York},\ \bibinfo {year} {2007})\ pp.\ \bibinfo {pages}
  {420--429}\BibitemShut {NoStop}%
\bibitem [{\citenamefont {Song}\ \emph {et~al.}(2006)\citenamefont {Song},
  \citenamefont {Tseng}, \citenamefont {Lin},\ and\ \citenamefont
  {Sun}}]{Song_recomm_2006}%
  \BibitemOpen
  \bibfield  {author} {\bibinfo {author} {\bibfnamefont {X.}~\bibnamefont
  {Song}}, \bibinfo {author} {\bibfnamefont {B.}~\bibnamefont {Tseng}},
  \bibinfo {author} {\bibfnamefont {C.}~\bibnamefont {Lin}},\ and\ \bibinfo
  {author} {\bibfnamefont {M.}~\bibnamefont {Sun}},\ }\bibfield  {title}
  {\bibinfo {title} {Personalized recommendation driven by information flow},\
  }in\ \href {https://doi.org/10.1145/1148170.1148258} {\emph {\bibinfo
  {booktitle} {Proceedings of the 29th Annual International ACM SIGIR
  Conference on Research and Development in Information Retrieval}}}\ (\bibinfo
   {publisher} {ACM},\ \bibinfo {address} {New York},\ \bibinfo {year} {2006})\
  pp.\ \bibinfo {pages} {509--516}\BibitemShut {NoStop}%
\bibitem [{\citenamefont {Goyal}\ \emph {et~al.}(2011)\citenamefont {Goyal},
  \citenamefont {Lu},\ and\ \citenamefont {Lakshmanan}}]{Goyal_CELF_2011}%
  \BibitemOpen
  \bibfield  {author} {\bibinfo {author} {\bibfnamefont {A.}~\bibnamefont
  {Goyal}}, \bibinfo {author} {\bibfnamefont {W.}~\bibnamefont {Lu}},\ and\
  \bibinfo {author} {\bibfnamefont {L.}~\bibnamefont {Lakshmanan}},\ }\bibfield
   {title} {\bibinfo {title} {{CELF}++: {O}ptimizing the greedy algorithm for
  influence maximization in social networks},\ }in\ \href
  {https://doi.org/10.1145/1963192.1963217} {\emph {\bibinfo {booktitle}
  {Proceedings of the 20th International Conference Companion on World Wide
  Web}}}\ (\bibinfo  {publisher} {ACM},\ \bibinfo {address} {New York},\
  \bibinfo {year} {2011})\ pp.\ \bibinfo {pages} {47--48}\BibitemShut {NoStop}%
\bibitem [{\citenamefont {Banerjee}\ \emph {et~al.}(2020)\citenamefont
  {Banerjee}, \citenamefont {Jenamani},\ and\ \citenamefont
  {Pratihar}}]{Banerjee_IMsurvey_2020}%
  \BibitemOpen
  \bibfield  {author} {\bibinfo {author} {\bibfnamefont {S.}~\bibnamefont
  {Banerjee}}, \bibinfo {author} {\bibfnamefont {M.}~\bibnamefont {Jenamani}},\
  and\ \bibinfo {author} {\bibfnamefont {D.}~\bibnamefont {Pratihar}},\
  }\bibfield  {title} {\bibinfo {title} {A survey on influence maximization in
  a social network},\ }\href@noop {} {\bibfield  {journal} {\bibinfo  {journal}
  {Knowledge Inf. Sist.}\ }\textbf {\bibinfo {volume} {62}},\ \bibinfo {pages}
  {3417} (\bibinfo {year} {2020})}\BibitemShut {NoStop}%
\bibitem [{\citenamefont {Li}\ \emph {et~al.}(2018)\citenamefont {Li},
  \citenamefont {Fan}, \citenamefont {Wang},\ and\ \citenamefont
  {Tan}}]{Li_IMsurvey_2018}%
  \BibitemOpen
  \bibfield  {author} {\bibinfo {author} {\bibfnamefont {Y.}~\bibnamefont
  {Li}}, \bibinfo {author} {\bibfnamefont {J.}~\bibnamefont {Fan}}, \bibinfo
  {author} {\bibfnamefont {Y.}~\bibnamefont {Wang}},\ and\ \bibinfo {author}
  {\bibfnamefont {K.}~\bibnamefont {Tan}},\ }\bibfield  {title} {\bibinfo
  {title} {Influence maximization on social graphs: A survey},\ }\href@noop {}
  {\bibfield  {journal} {\bibinfo  {journal} {IEEE Trans. Knowledge Data
  Engrg.}\ }\textbf {\bibinfo {volume} {30}},\ \bibinfo {pages} {1852}
  (\bibinfo {year} {2018})}\BibitemShut {NoStop}%
\bibitem [{\citenamefont {Belotti}\ \emph {et~al.}(2013)\citenamefont
  {Belotti}, \citenamefont {Kirches}, \citenamefont {Leyffer}, \citenamefont
  {Linderoth}, \citenamefont {Luedtke},\ and\ \citenamefont
  {Mahajan}}]{Belotti_MINLP_2013}%
  \BibitemOpen
  \bibfield  {author} {\bibinfo {author} {\bibfnamefont {P.}~\bibnamefont
  {Belotti}}, \bibinfo {author} {\bibfnamefont {C.}~\bibnamefont {Kirches}},
  \bibinfo {author} {\bibfnamefont {S.}~\bibnamefont {Leyffer}}, \bibinfo
  {author} {\bibfnamefont {J.}~\bibnamefont {Linderoth}}, \bibinfo {author}
  {\bibfnamefont {J.}~\bibnamefont {Luedtke}},\ and\ \bibinfo {author}
  {\bibfnamefont {A.}~\bibnamefont {Mahajan}},\ }\bibfield  {title} {\bibinfo
  {title} {Mixed-integer nonlinear optimization},\ }\href@noop {} {\bibfield
  {journal} {\bibinfo  {journal} {Acta Numer.}\ }\textbf {\bibinfo {volume}
  {22}},\ \bibinfo {pages} {1} (\bibinfo {year} {2013})}\BibitemShut {NoStop}%
\bibitem [{\citenamefont {Boukouvala}\ \emph {et~al.}(2016)\citenamefont
  {Boukouvala}, \citenamefont {Misener},\ and\ \citenamefont
  {Floudas}}]{Boukouvala_GOMINLPCDFM_2016}%
  \BibitemOpen
  \bibfield  {author} {\bibinfo {author} {\bibfnamefont {F.}~\bibnamefont
  {Boukouvala}}, \bibinfo {author} {\bibfnamefont {R.}~\bibnamefont
  {Misener}},\ and\ \bibinfo {author} {\bibfnamefont {C.}~\bibnamefont
  {Floudas}},\ }\bibfield  {title} {\bibinfo {title} {Global optimization
  advances in {M}ixed-{I}nteger {N}onlinear {P}rogramming, {MINLP}, and
  {C}onstrained {D}erivative-{F}ree {O}ptimization, {CDFO}},\ }\href@noop {}
  {\bibfield  {journal} {\bibinfo  {journal} {European J. Oper. Res.}\ }\textbf
  {\bibinfo {volume} {252}},\ \bibinfo {pages} {701} (\bibinfo {year}
  {2016})}\BibitemShut {NoStop}%
\bibitem [{\citenamefont {Burer}\ and\ \citenamefont
  {Letchford}(2012)}]{Burer_MINLPs_2012}%
  \BibitemOpen
  \bibfield  {author} {\bibinfo {author} {\bibfnamefont {S.}~\bibnamefont
  {Burer}}\ and\ \bibinfo {author} {\bibfnamefont {A.}~\bibnamefont
  {Letchford}},\ }\bibfield  {title} {\bibinfo {title} {Non-convex
  mixed-integer nonlinear programming: {A} survey},\ }\href
  {https://doi.org/10.1016/j.sorms.2012.08.001} {\bibfield  {journal} {\bibinfo
   {journal} {Surv. Oper. Res. Manag. Sci.}\ }\textbf {\bibinfo {volume}
  {17}},\ \bibinfo {pages} {97} (\bibinfo {year} {2012})}\BibitemShut {NoStop}%
\bibitem [{\citenamefont {Abramson}\ \emph {et~al.}(2009)\citenamefont
  {Abramson}, \citenamefont {Audet}, \citenamefont {Chrissis},\ and\
  \citenamefont {Walston}}]{Abramson_MADSmv_2009}%
  \BibitemOpen
  \bibfield  {author} {\bibinfo {author} {\bibfnamefont {M.}~\bibnamefont
  {Abramson}}, \bibinfo {author} {\bibfnamefont {C.}~\bibnamefont {Audet}},
  \bibinfo {author} {\bibfnamefont {J.}~\bibnamefont {Chrissis}},\ and\
  \bibinfo {author} {\bibfnamefont {J.}~\bibnamefont {Walston}},\ }\bibfield
  {title} {\bibinfo {title} {Mesh adaptive direct search algorithms for mixed
  variable optimization},\ }\href@noop {} {\bibfield  {journal} {\bibinfo
  {journal} {Optim. Lett.}\ }\textbf {\bibinfo {volume} {3}},\ \bibinfo {pages}
  {35} (\bibinfo {year} {2009})}\BibitemShut {NoStop}%
\bibitem [{Note1()}]{Note1}%
  \BibitemOpen
  \bibinfo {note} {\relax \protect \fontsize {8}{9.5pt}\protect \selectfont
  \abovedisplayskip 6\p@ plus2\p@ minus4\p@ \belowdisplayskip \abovedisplayskip
  \abovedisplayshortskip \z@ plus\p@ \belowdisplayshortskip 3\p@ plus\p@
  minus2\p@ \def \leftmargin \leftmargini \parsep 4\p@ plus2\p@ minus\p@
  \topsep 8\p@ plus2\p@ minus4\p@ \itemsep 4\p@ plus2\p@ minus\p@ {\leftmargin
  \leftmargini \topsep 3\p@ plus\p@ minus\p@ \parsep 2\p@ plus\p@ minus\p@
  \itemsep \parsep }Note that the cases when the network is undirected or
  disconnected or unweighted can be treated similarly.}\BibitemShut {Stop}%
\bibitem [{\citenamefont {Chen}\ \emph {et~al.}(2013)\citenamefont {Chen},
  \citenamefont {Castillo},\ and\ \citenamefont {Lakshmanan}}]{Chen_Info_2013}%
  \BibitemOpen
  \bibfield  {author} {\bibinfo {author} {\bibfnamefont {W.}~\bibnamefont
  {Chen}}, \bibinfo {author} {\bibfnamefont {C.}~\bibnamefont {Castillo}},\
  and\ \bibinfo {author} {\bibfnamefont {L.}~\bibnamefont {Lakshmanan}},\
  }\href {https://doi.org/10.1007/978-3-031-01850-3} {\emph {\bibinfo {title}
  {Information and Influence Propagation in Social Networks}}}\ (\bibinfo
  {publisher} {Springer Cham},\ \bibinfo {year} {2013})\BibitemShut {NoStop}%
\bibitem [{\citenamefont {Borgs}\ \emph {et~al.}(2014)\citenamefont {Borgs},
  \citenamefont {Brautbar}, \citenamefont {Chayes},\ and\ \citenamefont
  {Lucier}}]{Borgs_IMoptimal_2014}%
  \BibitemOpen
  \bibfield  {author} {\bibinfo {author} {\bibfnamefont {C.}~\bibnamefont
  {Borgs}}, \bibinfo {author} {\bibfnamefont {M.}~\bibnamefont {Brautbar}},
  \bibinfo {author} {\bibfnamefont {J.}~\bibnamefont {Chayes}},\ and\ \bibinfo
  {author} {\bibfnamefont {B.}~\bibnamefont {Lucier}},\ }\bibfield  {title}
  {\bibinfo {title} {Maximizing social influence in nearly optimal time},\ }in\
  \href {https://doi.org/10.1137/1.9781611973402.70} {\emph {\bibinfo
  {booktitle} {Proceedings of the 25th Annual ACM-SIAM Symposium on Discrete
  Algorithms}}},\ \bibinfo {editor} {edited by\ \bibinfo {editor}
  {\bibfnamefont {C.}~\bibnamefont {Chekuri}}}\ (\bibinfo  {publisher} {SIAM},\
  \bibinfo {address} {Philadelphia},\ \bibinfo {year} {2014})\ pp.\ \bibinfo
  {pages} {946--957}\BibitemShut {NoStop}%
\bibitem [{\citenamefont {Chen}\ \emph
  {et~al.}(2010{\natexlab{b}})\citenamefont {Chen}, \citenamefont {Yuan},\ and\
  \citenamefont {Zhang}}]{Chen_scalable_2010}%
  \BibitemOpen
  \bibfield  {author} {\bibinfo {author} {\bibfnamefont {W.}~\bibnamefont
  {Chen}}, \bibinfo {author} {\bibfnamefont {Y.}~\bibnamefont {Yuan}},\ and\
  \bibinfo {author} {\bibfnamefont {L.}~\bibnamefont {Zhang}},\ }\bibfield
  {title} {\bibinfo {title} {Scalable influence maximization in social networks
  under the linear threshold model},\ }in\ \href
  {https://doi.org/10.1109/ICDM.2010.118} {\emph {\bibinfo {booktitle} {2010
  IEEE International Conference on Data Mining}}},\ \bibinfo {editor} {edited
  by\ \bibinfo {editor} {\bibfnamefont {G.}~\bibnamefont {Webb}}, \bibinfo
  {editor} {\bibfnamefont {B.}~\bibnamefont {Liu}}, \bibinfo {editor}
  {\bibfnamefont {C.}~\bibnamefont {Zhang}}, \bibinfo {editor} {\bibfnamefont
  {D.}~\bibnamefont {Gunopulos}},\ and\ \bibinfo {editor} {\bibfnamefont
  {X.}~\bibnamefont {Wu}}}\ (\bibinfo  {publisher} {IEEE},\ \bibinfo {address}
  {Los Alamitos},\ \bibinfo {year} {2010})\ pp.\ \bibinfo {pages}
  {88--97}\BibitemShut {NoStop}%
\bibitem [{\citenamefont {Wang}\ \emph {et~al.}(2012)\citenamefont {Wang},
  \citenamefont {Chen},\ and\ \citenamefont {Wang}}]{Wang_scalable_2012}%
  \BibitemOpen
  \bibfield  {author} {\bibinfo {author} {\bibfnamefont {C.}~\bibnamefont
  {Wang}}, \bibinfo {author} {\bibfnamefont {W.}~\bibnamefont {Chen}},\ and\
  \bibinfo {author} {\bibfnamefont {Y.}~\bibnamefont {Wang}},\ }\bibfield
  {title} {\bibinfo {title} {Scalable influence maximization for independent
  cascade model in large-scale social networks},\ }\href@noop {} {\bibfield
  {journal} {\bibinfo  {journal} {Data Min. Knowl. Discov.}\ }\textbf {\bibinfo
  {volume} {25}},\ \bibinfo {pages} {545} (\bibinfo {year} {2012})}\BibitemShut
  {NoStop}%
\bibitem [{\citenamefont {Chen}(2009)}]{Chen_approximability_2009}%
  \BibitemOpen
  \bibfield  {author} {\bibinfo {author} {\bibfnamefont {N.}~\bibnamefont
  {Chen}},\ }\bibfield  {title} {\bibinfo {title} {On the approximability of
  influence in social networks},\ }\href@noop {} {\bibfield  {journal}
  {\bibinfo  {journal} {SIAM J. Discrete Math.}\ }\textbf {\bibinfo {volume}
  {23}},\ \bibinfo {pages} {1400} (\bibinfo {year} {2009})}\BibitemShut
  {NoStop}%
\bibitem [{\citenamefont {Kempe}\ \emph {et~al.}(2015)\citenamefont {Kempe},
  \citenamefont {Kleinberg},\ and\ \citenamefont
  {Tardos}}]{Kempe_influence_2015}%
  \BibitemOpen
  \bibfield  {author} {\bibinfo {author} {\bibfnamefont {D.}~\bibnamefont
  {Kempe}}, \bibinfo {author} {\bibfnamefont {J.}~\bibnamefont {Kleinberg}},\
  and\ \bibinfo {author} {\bibfnamefont {E.}~\bibnamefont {Tardos}},\
  }\bibfield  {title} {\bibinfo {title} {Maximizing the spread of influence
  through a social network},\ }\href@noop {} {\bibfield  {journal} {\bibinfo
  {journal} {Theory Comput.}\ }\textbf {\bibinfo {volume} {11}},\ \bibinfo
  {pages} {105} (\bibinfo {year} {2015})}\BibitemShut {NoStop}%
\bibitem [{\citenamefont {Even-Dar}\ and\ \citenamefont
  {Shapira}(2007)}]{Even-Dar_voter_2007}%
  \BibitemOpen
  \bibfield  {author} {\bibinfo {author} {\bibfnamefont {E.}~\bibnamefont
  {Even-Dar}}\ and\ \bibinfo {author} {\bibfnamefont {A.}~\bibnamefont
  {Shapira}},\ }\bibfield  {title} {\bibinfo {title} {A note on maximizing the
  spread of influence in social networks},\ }in\ \href
  {https://doi.org/10.1007/978-3-540-77105-0_27} {\emph {\bibinfo {booktitle}
  {Internet and Network Economics}}},\ \bibinfo {editor} {edited by\ \bibinfo
  {editor} {\bibfnamefont {X.}~\bibnamefont {Deng}}\ and\ \bibinfo {editor}
  {\bibfnamefont {F.}~\bibnamefont {Graham}}}\ (\bibinfo  {publisher}
  {Springer},\ \bibinfo {address} {Berlin},\ \bibinfo {year} {2007})\ pp.\
  \bibinfo {pages} {281--286}\BibitemShut {NoStop}%
\bibitem [{\citenamefont {Demaine}\ \emph {et~al.}(2014)\citenamefont
  {Demaine}, \citenamefont {Hajiaghayi}, \citenamefont {Mahini}, \citenamefont
  {Malec}, \citenamefont {Raghavan}, \citenamefont {Sawant},\ and\
  \citenamefont {Zadimoghadam}}]{Demanine_ContinInf_2014}%
  \BibitemOpen
  \bibfield  {author} {\bibinfo {author} {\bibfnamefont {E.}~\bibnamefont
  {Demaine}}, \bibinfo {author} {\bibfnamefont {M.}~\bibnamefont {Hajiaghayi}},
  \bibinfo {author} {\bibfnamefont {H.}~\bibnamefont {Mahini}}, \bibinfo
  {author} {\bibfnamefont {D.}~\bibnamefont {Malec}}, \bibinfo {author}
  {\bibfnamefont {S.}~\bibnamefont {Raghavan}}, \bibinfo {author}
  {\bibfnamefont {A.}~\bibnamefont {Sawant}},\ and\ \bibinfo {author}
  {\bibfnamefont {M.}~\bibnamefont {Zadimoghadam}},\ }\bibfield  {title}
  {\bibinfo {title} {How to influence people with partial incentives},\ }in\
  \href {https://doi.org/10.1145/2566486.2568039} {\emph {\bibinfo {booktitle}
  {Proceedings of the 23rd International Conference on World Wide Web}}}\
  (\bibinfo  {publisher} {ACM},\ \bibinfo {address} {New York},\ \bibinfo
  {year} {2014})\ pp.\ \bibinfo {pages} {937--948}\BibitemShut {NoStop}%
\bibitem [{\citenamefont {Srivastava}\ \emph {et~al.}(2010)\citenamefont
  {Srivastava}, \citenamefont {Moehlis},\ and\ \citenamefont
  {Bullo}}]{Srivastava_bifur_2010}%
  \BibitemOpen
  \bibfield  {author} {\bibinfo {author} {\bibfnamefont {V.}~\bibnamefont
  {Srivastava}}, \bibinfo {author} {\bibfnamefont {J.}~\bibnamefont
  {Moehlis}},\ and\ \bibinfo {author} {\bibfnamefont {F.}~\bibnamefont
  {Bullo}},\ }\bibfield  {title} {\bibinfo {title} {On bifurcations in
  nonlinear consensus networks},\ }in\ \href
  {https://doi.org/10.1109/ACC.2010.5531534} {\emph {\bibinfo {booktitle}
  {Proceedings of the 2010 American Control Conference}}}\ (\bibinfo
  {publisher} {IEEE},\ \bibinfo {year} {2010})\ pp.\ \bibinfo {pages}
  {1647--1652}\BibitemShut {NoStop}%
\bibitem [{\citenamefont {Asllani}\ \emph {et~al.}(2018)\citenamefont
  {Asllani}, \citenamefont {Carletti}, \citenamefont {Di~Patti}, \citenamefont
  {Fanelli},\ and\ \citenamefont {Piazza}}]{Asllani_crowd_2018}%
  \BibitemOpen
  \bibfield  {author} {\bibinfo {author} {\bibfnamefont {M.}~\bibnamefont
  {Asllani}}, \bibinfo {author} {\bibfnamefont {T.}~\bibnamefont {Carletti}},
  \bibinfo {author} {\bibfnamefont {F.}~\bibnamefont {Di~Patti}}, \bibinfo
  {author} {\bibfnamefont {D.}~\bibnamefont {Fanelli}},\ and\ \bibinfo {author}
  {\bibfnamefont {F.}~\bibnamefont {Piazza}},\ }\bibfield  {title} {\bibinfo
  {title} {Hopping in the crowd to unveil network topology},\ }\href@noop {}
  {\bibfield  {journal} {\bibinfo  {journal} {Phys. Rev. Lett.}\ }\textbf
  {\bibinfo {volume} {120}},\ \bibinfo {pages} {158301} (\bibinfo {year}
  {2018})}\BibitemShut {NoStop}%
\bibitem [{\citenamefont {Fanelli}\ and\ \citenamefont
  {McKane}(2010)}]{Fanelli_crowd_2010}%
  \BibitemOpen
  \bibfield  {author} {\bibinfo {author} {\bibfnamefont {D.}~\bibnamefont
  {Fanelli}}\ and\ \bibinfo {author} {\bibfnamefont {A.}~\bibnamefont
  {McKane}},\ }\bibfield  {title} {\bibinfo {title} {Diffusion in a crowded
  environment},\ }\href@noop {} {\bibfield  {journal} {\bibinfo  {journal}
  {Phys. Rev. E}\ }\textbf {\bibinfo {volume} {82}},\ \bibinfo {pages} {021113}
  (\bibinfo {year} {2010})}\BibitemShut {NoStop}%
\bibitem [{Note2()}]{Note2}%
  \BibitemOpen
  \bibinfo {note} {\relax \protect \fontsize {8}{9.5pt}\protect \selectfont
  \abovedisplayskip 6\p@ plus2\p@ minus4\p@ \belowdisplayskip \abovedisplayskip
  \abovedisplayshortskip \z@ plus\p@ \belowdisplayshortskip 3\p@ plus\p@
  minus2\p@ \def \leftmargin \leftmargini \parsep 4\p@ plus2\p@ minus\p@
  \topsep 8\p@ plus2\p@ minus4\p@ \itemsep 4\p@ plus2\p@ minus\p@ {\leftmargin
  \leftmargini \topsep 3\p@ plus\p@ minus\p@ \parsep 2\p@ plus\p@ minus\p@
  \itemsep \parsep }With the threshold-type bounds, the condition $\DOTSB \sum@
  \slimits@ _iW_{ij}x_i(0) \ge l_{j,1}$ at $t=1$ is equivalent to $\DOTSB \sum@
  \slimits@ _i(W_{ij}/\alpha )(x_i(0)/l_{j,0}) \ge \theta _{l,j}$. Hence,
  $\alpha $ will not affect the activation so long as the relative weight
  $W_{ij}/\alpha $ does not change (e.g.,, $W_{ij}/\alpha = 1$ if the network
  has uniform edge weight).}\BibitemShut {Stop}%
\bibitem [{Note3()}]{Note3}%
  \BibitemOpen
  \bibinfo {note} {\relax \protect \fontsize {8}{9.5pt}\protect \selectfont
  \abovedisplayskip 6\p@ plus2\p@ minus4\p@ \belowdisplayskip \abovedisplayskip
  \abovedisplayshortskip \z@ plus\p@ \belowdisplayshortskip 3\p@ plus\p@
  minus2\p@ \def \leftmargin \leftmargini \parsep 4\p@ plus2\p@ minus\p@
  \topsep 8\p@ plus2\p@ minus4\p@ \itemsep 4\p@ plus2\p@ minus\p@ {\leftmargin
  \leftmargini \topsep 3\p@ plus\p@ minus\p@ \parsep 2\p@ plus\p@ minus\p@
  \itemsep \parsep }Here we consider the distribution of the SBM, and the
  expectation is take over this distribution.}\BibitemShut {Stop}%
\bibitem [{Note4()}]{Note4}%
  \BibitemOpen
  \bibinfo {note} {\relax \protect \fontsize {8}{9.5pt}\protect \selectfont
  \abovedisplayskip 6\p@ plus2\p@ minus4\p@ \belowdisplayskip \abovedisplayskip
  \abovedisplayshortskip \z@ plus\p@ \belowdisplayshortskip 3\p@ plus\p@
  minus2\p@ \def \leftmargin \leftmargini \parsep 4\p@ plus2\p@ minus\p@
  \topsep 8\p@ plus2\p@ minus4\p@ \itemsep 4\p@ plus2\p@ minus\p@ {\leftmargin
  \leftmargini \topsep 3\p@ plus\p@ minus\p@ \parsep 2\p@ plus\p@ minus\p@
  \itemsep \parsep }{In the specific case here, the upper bound $2l_1^*$ is
  equivalent to require that at most two initially activated neighbors are
  needed to activate a node.}}\BibitemShut {Stop}%
\bibitem [{\citenamefont {Lu}\ \emph {et~al.}(2012)\citenamefont {Lu},
  \citenamefont {Zhang}, \citenamefont {Wu}, \citenamefont {Kim},\ and\
  \citenamefont {Fu}}]{Lu_compldLT_2012}%
  \BibitemOpen
  \bibfield  {author} {\bibinfo {author} {\bibfnamefont {Z.}~\bibnamefont
  {Lu}}, \bibinfo {author} {\bibfnamefont {W.}~\bibnamefont {Zhang}}, \bibinfo
  {author} {\bibfnamefont {W.}~\bibnamefont {Wu}}, \bibinfo {author}
  {\bibfnamefont {J.}~\bibnamefont {Kim}},\ and\ \bibinfo {author}
  {\bibfnamefont {B.}~\bibnamefont {Fu}},\ }\bibfield  {title} {\bibinfo
  {title} {The complexity of influence maximization problem in the
  deterministic linear threshold model},\ }\href@noop {} {\bibfield  {journal}
  {\bibinfo  {journal} {J. Comb. Optim.}\ }\textbf {\bibinfo {volume} {24}},\
  \bibinfo {pages} {374} (\bibinfo {year} {2012})}\BibitemShut {NoStop}%
\bibitem [{\citenamefont {Vicente}\ and\ \citenamefont
  {Custódio}(2012)}]{Vicente_discont_2012}%
  \BibitemOpen
  \bibfield  {author} {\bibinfo {author} {\bibfnamefont {L.}~\bibnamefont
  {Vicente}}\ and\ \bibinfo {author} {\bibfnamefont {A.}~\bibnamefont
  {Custódio}},\ }\bibfield  {title} {\bibinfo {title} {Analysis of direct
  searches for discontinuous functions},\ }\href@noop {} {\bibfield  {journal}
  {\bibinfo  {journal} {Math. Program.}\ }\textbf {\bibinfo {volume} {133}},\
  \bibinfo {pages} {299} (\bibinfo {year} {2012})}\BibitemShut {NoStop}%
\bibitem [{\citenamefont {Audet}\ \emph {et~al.}(2021)\citenamefont {Audet},
  \citenamefont {Le~Digabel}, \citenamefont {Montplaisir},\ and\ \citenamefont
  {Tribes}}]{Audet_NOMAD4_2021}%
  \BibitemOpen
  \bibfield  {author} {\bibinfo {author} {\bibfnamefont {C.}~\bibnamefont
  {Audet}}, \bibinfo {author} {\bibfnamefont {S.}~\bibnamefont {Le~Digabel}},
  \bibinfo {author} {\bibfnamefont {V.}~\bibnamefont {Montplaisir}},\ and\
  \bibinfo {author} {\bibfnamefont {C.}~\bibnamefont {Tribes}},\ }\bibfield
  {title} {\bibinfo {title} {{NOMAD} version 4: {N}onlinear optimization with
  the {MADS} algorithm},\ }\href {https://arxiv.org/abs/2104.11627} {\bibfield
  {journal} {\bibinfo  {journal} {preprint}\ }\textbf {\bibinfo {volume}
  {arXiv:2104.11627v2}} (\bibinfo {year} {2021})}\BibitemShut {NoStop}%
\bibitem [{\citenamefont {Le~Digabel}(2011)}]{LeDigabel_NOMAD_2011}%
  \BibitemOpen
  \bibfield  {author} {\bibinfo {author} {\bibfnamefont {S.}~\bibnamefont
  {Le~Digabel}},\ }\bibfield  {title} {\bibinfo {title} {Algorithm 909:
  {NOMAD}: {N}onlinear optimization with the {MADS} algorithm},\ }\bibfield
  {journal} {\bibinfo  {journal} {ACM Trans. Math. Software}\ }\textbf
  {\bibinfo {volume} {37}},\ \href {https://doi.org/10.1145/1916461.1916468}
  {10.1145/1916461.1916468} (\bibinfo {year} {2011})\BibitemShut {NoStop}%
\bibitem [{\citenamefont {Laguna}\ \emph {et~al.}(2014)\citenamefont {Laguna},
  \citenamefont {Gortázar}, \citenamefont {Gallego}, \citenamefont {Duarte},\
  and\ \citenamefont {Martí}}]{Laguna_scatsearch_2014}%
  \BibitemOpen
  \bibfield  {author} {\bibinfo {author} {\bibfnamefont {M.}~\bibnamefont
  {Laguna}}, \bibinfo {author} {\bibfnamefont {F.}~\bibnamefont {Gortázar}},
  \bibinfo {author} {\bibfnamefont {M.}~\bibnamefont {Gallego}}, \bibinfo
  {author} {\bibfnamefont {A.}~\bibnamefont {Duarte}},\ and\ \bibinfo {author}
  {\bibfnamefont {R.}~\bibnamefont {Martí}},\ }\bibfield  {title} {\bibinfo
  {title} {A black-box scatter search for optimization problems with integer
  variables},\ }\href@noop {} {\bibfield  {journal} {\bibinfo  {journal} {J.
  Global Optim.}\ }\textbf {\bibinfo {volume} {58}},\ \bibinfo {pages} {497}
  (\bibinfo {year} {2014})}\BibitemShut {NoStop}%
\bibitem [{Note5()}]{Note5}%
  \BibitemOpen
  \bibinfo {note} {\relax \protect \fontsize {8}{9.5pt}\protect \selectfont
  \abovedisplayskip 6\p@ plus2\p@ minus4\p@ \belowdisplayskip \abovedisplayskip
  \abovedisplayshortskip \z@ plus\p@ \belowdisplayshortskip 3\p@ plus\p@
  minus2\p@ \def \leftmargin \leftmargini \parsep 4\p@ plus2\p@ minus\p@
  \topsep 8\p@ plus2\p@ minus4\p@ \itemsep 4\p@ plus2\p@ minus\p@ {\leftmargin
  \leftmargini \topsep 3\p@ plus\p@ minus\p@ \parsep 2\p@ plus\p@ minus\p@
  \itemsep \parsep }{It corresponds to the node sets of sizes less than $\theta
  _l$ in networks with uniform weights, while in weighted networks, we also
  need to incorporate the exact weights around each node compared to the mean
  weight $\alpha $.}}\BibitemShut {Stop}%
\bibitem [{\citenamefont {Newman}(2001)}]{Newman_collabor_2001}%
  \BibitemOpen
  \bibfield  {author} {\bibinfo {author} {\bibfnamefont {M.}~\bibnamefont
  {Newman}},\ }\bibfield  {title} {\bibinfo {title} {The structure of
  scientific collaboration networks},\ }\href@noop {} {\bibfield  {journal}
  {\bibinfo  {journal} {Proc. Natl. Acad. Sci.}\ }\textbf {\bibinfo {volume}
  {98}},\ \bibinfo {pages} {404} (\bibinfo {year} {2001})}\BibitemShut
  {NoStop}%
\bibitem [{\citenamefont {Yang}\ and\ \citenamefont
  {Leskovec}(2015)}]{Yang_dblpdata_2018}%
  \BibitemOpen
  \bibfield  {author} {\bibinfo {author} {\bibfnamefont {J.}~\bibnamefont
  {Yang}}\ and\ \bibinfo {author} {\bibfnamefont {J.}~\bibnamefont
  {Leskovec}},\ }\bibfield  {title} {\bibinfo {title} {Defining and evaluating
  network communities based on ground-truth},\ }\href@noop {} {\bibfield
  {journal} {\bibinfo  {journal} {Knowl. Inf. Syst.}\ }\textbf {\bibinfo
  {volume} {42}},\ \bibinfo {pages} {181} (\bibinfo {year} {2015})}\BibitemShut
  {NoStop}%
\bibitem [{Note6()}]{Note6}%
  \BibitemOpen
  \bibinfo {note} {\relax \protect \fontsize {8}{9.5pt}\protect \selectfont
  \abovedisplayskip 6\p@ plus2\p@ minus4\p@ \belowdisplayskip \abovedisplayskip
  \abovedisplayshortskip \z@ plus\p@ \belowdisplayshortskip 3\p@ plus\p@
  minus2\p@ \def \leftmargin \leftmargini \parsep 4\p@ plus2\p@ minus\p@
  \topsep 8\p@ plus2\p@ minus4\p@ \itemsep 4\p@ plus2\p@ minus\p@ {\leftmargin
  \leftmargini \topsep 3\p@ plus\p@ minus\p@ \parsep 2\p@ plus\p@ minus\p@
  \itemsep \parsep }Note that even through higher-order interactions are likely
  to occur in collaboration networks, the simple version (as what we considered
  here) still remains a classic example in the IM problem.}\BibitemShut {Stop}%
\bibitem [{\citenamefont {Gursoy}\ and\ \citenamefont
  {Gunnec}(2018)}]{Gursoy_IMdet_2018}%
  \BibitemOpen
  \bibfield  {author} {\bibinfo {author} {\bibfnamefont {F.}~\bibnamefont
  {Gursoy}}\ and\ \bibinfo {author} {\bibfnamefont {D.}~\bibnamefont
  {Gunnec}},\ }\bibfield  {title} {\bibinfo {title} {Influence maximization in
  social networks under deterministic linear threshold model},\ }\href@noop {}
  {\bibfield  {journal} {\bibinfo  {journal} {Knowl. Based Syst.}\ }\textbf
  {\bibinfo {volume} {161}},\ \bibinfo {pages} {111} (\bibinfo {year}
  {2018})}\BibitemShut {NoStop}%
\bibitem [{Note7()}]{Note7}%
  \BibitemOpen
  \bibinfo {note} {\relax \protect \fontsize {8}{9.5pt}\protect \selectfont
  \abovedisplayskip 6\p@ plus2\p@ minus4\p@ \belowdisplayskip \abovedisplayskip
  \abovedisplayshortskip \z@ plus\p@ \belowdisplayshortskip 3\p@ plus\p@
  minus2\p@ \def \leftmargin \leftmargini \parsep 4\p@ plus2\p@ minus\p@
  \topsep 8\p@ plus2\p@ minus4\p@ \itemsep 4\p@ plus2\p@ minus\p@ {\leftmargin
  \leftmargini \topsep 3\p@ plus\p@ minus\p@ \parsep 2\p@ plus\p@ minus\p@
  \itemsep \parsep }{The expressions of initially activated nodes can be
  different from others in the same community, due to the common assumption of
  no self-edges. However, noting that $k\ll n$, we allow self-edges for
  illustrative purposes, thus ignore such differences.}}\BibitemShut {Stop}%
\bibitem [{\citenamefont {Giovannelli}\ \emph {et~al.}(0221)\citenamefont
  {Giovannelli}, \citenamefont {Liuzzi}, \citenamefont {Lucidi},\ and\
  \citenamefont {Rinaldi}}]{Giovannelli_DFL_2021}%
  \BibitemOpen
  \bibfield  {author} {\bibinfo {author} {\bibfnamefont {T.}~\bibnamefont
  {Giovannelli}}, \bibinfo {author} {\bibfnamefont {G.}~\bibnamefont {Liuzzi}},
  \bibinfo {author} {\bibfnamefont {S.}~\bibnamefont {Lucidi}},\ and\ \bibinfo
  {author} {\bibfnamefont {F.}~\bibnamefont {Rinaldi}},\ }\bibfield  {title}
  {\bibinfo {title} {Derivative-free methods for mixed-integer nonsmooth
  constrained optimization},\ }\href
  {https://doi.org/10.1007/s10589-022-00363-1} {\bibfield  {journal} {\bibinfo
  {journal} {Comput. Optim. Appl.preprint}\ }\textbf {\bibinfo {volume} {82}},\
  \bibinfo {pages} {293} (\bibinfo {year} {20221})}\BibitemShut {NoStop}%
\bibitem [{\citenamefont {Zachary}(1977)}]{Zachary_karate_1977}%
  \BibitemOpen
  \bibfield  {author} {\bibinfo {author} {\bibfnamefont {W.}~\bibnamefont
  {Zachary}},\ }\bibfield  {title} {\bibinfo {title} {An information flow model
  for conflict and fission in small groups},\ }\href@noop {} {\bibfield
  {journal} {\bibinfo  {journal} {J. Anthropol. Res.}\ }\textbf {\bibinfo
  {volume} {33}},\ \bibinfo {pages} {452} (\bibinfo {year} {1977})}\BibitemShut
  {NoStop}%
\bibitem [{\citenamefont {Schoenebeck}\ \emph {et~al.}(2022)\citenamefont
  {Schoenebeck}, \citenamefont {Tao},\ and\ \citenamefont
  {Yu}}]{Schoenebeck_GlobalLocal_2019}%
  \BibitemOpen
  \bibfield  {author} {\bibinfo {author} {\bibfnamefont {G.}~\bibnamefont
  {Schoenebeck}}, \bibinfo {author} {\bibfnamefont {B.}~\bibnamefont {Tao}},\
  and\ \bibinfo {author} {\bibfnamefont {F.}~\bibnamefont {Yu}},\ }\bibfield
  {title} {\bibinfo {title} {Think globally, act locally: {O}n the optimal
  seeding for nonsubmodular influence maximization},\ }\bibfield  {journal}
  {\bibinfo  {journal} {Inf. Comput.}\ }\textbf {\bibinfo {volume} {285}},\
  \href {https://doi.org/10.1016/j.ic.2022.104919} {10.1016/j.ic.2022.104919}
  (\bibinfo {year} {2022})\BibitemShut {NoStop}%
\bibitem [{Note8()}]{Note8}%
  \BibitemOpen
  \bibinfo {note} {\relax \protect \fontsize {8}{9.5pt}\protect \selectfont
  \abovedisplayskip 6\p@ plus2\p@ minus4\p@ \belowdisplayskip \abovedisplayskip
  \abovedisplayshortskip \z@ plus\p@ \belowdisplayshortskip 3\p@ plus\p@
  minus2\p@ \def \leftmargin \leftmargini \parsep 4\p@ plus2\p@ minus\p@
  \topsep 8\p@ plus2\p@ minus4\p@ \itemsep 4\p@ plus2\p@ minus\p@ {\leftmargin
  \leftmargini \topsep 3\p@ plus\p@ minus\p@ \parsep 2\p@ plus\p@ minus\p@
  \itemsep \parsep }{Note the same results can be obtained with arbitrary
  community sizes, but extra conditions on both the community sizes and the
  probabilities are required.}}\BibitemShut {Stop}%
\bibitem [{Note9()}]{Note9}%
  \BibitemOpen
  \bibinfo {note} {\relax \protect \fontsize {8}{9.5pt}\protect \selectfont
  \abovedisplayskip 6\p@ plus2\p@ minus4\p@ \belowdisplayskip \abovedisplayskip
  \abovedisplayshortskip \z@ plus\p@ \belowdisplayshortskip 3\p@ plus\p@
  minus2\p@ \def \leftmargin \leftmargini \parsep 4\p@ plus2\p@ minus\p@
  \topsep 8\p@ plus2\p@ minus4\p@ \itemsep 4\p@ plus2\p@ minus\p@ {\leftmargin
  \leftmargini \topsep 3\p@ plus\p@ minus\p@ \parsep 2\p@ plus\p@ minus\p@
  \itemsep \parsep }{Note that the expected values could be slightly different
  due to the common assumption of no self-edges. However, we assume $n \gg 1$,
  thus ignore such differences.}}\BibitemShut {Stop}%
\end{thebibliography}%

\end{document}